\documentclass[12pt]{article}
\usepackage{amsmath}
\usepackage{graphicx,psfrag,epsf}
\usepackage{enumerate}
\usepackage{natbib}
\usepackage{url} 
\usepackage{amsthm}

\usepackage{amsfonts}
\usepackage{multirow, booktabs}
\usepackage[table]{xcolor}
\usepackage{cleveref}
\usepackage[ruled,vlined]{algorithm2e}
\usepackage[T1]{fontenc}
\usepackage[utf8]{inputenc}
\usepackage{authblk}
\usepackage[multiple]{footmisc}
\usepackage{blindtext,titlefoot}
\usepackage{sectsty}

\usepackage[ruled,vlined]{algorithm2e}
\usepackage{amsmath, amssymb}
\usepackage{amsfonts, multirow, epsfig, subfig, floatrow}
\usepackage{graphicx, pdflscape, verbatim, enumerate, colortbl, setspace}
\usepackage{setspace, color,bm}
\usepackage[normalem]{ulem}
\usepackage{cite}
\usepackage{multirow}
\usepackage{booktabs,array}
\usepackage{url}
\usepackage{bbm}

\floatname{algorithm}{Procedure}

\newcommand{\indep}{\rotatebox[origin=c]{90}{$\models$}}

\newtheorem{Proposition}{Proposition}
\newtheorem{Lemma}{Lemma}

\newtheorem{Theorem}{Theorem}

\newtheorem{Assumption}{Assumption}

\makeatletter
\newcommand*{\rom}[1]{\expandafter\@slowromancap\romannumeral #1@}
\makeatother

\title{\LARGE\bf Increasing Power for Observational Studies of Aberrant Response: An Adaptive Approach}
\date{}

\author[1]{Siyu Heng}
\author[2]{Hyunseung Kang}
\author[1]{Dylan S. Small}
\author[3,$\dagger$]{Colin B. Fogarty}

\affil[1]{\it University of Pennsylvania}
\affil[2]{\it University of Wisconsin-Madison}
\affil[3]{\it Massachusetts Institute of Technology}

\begin{document}

\sectionfont{\bfseries\large\sffamily}%
%

\subsectionfont{\bfseries\sffamily\normalsize}%
%


\def\spacingset#1{\renewcommand{\baselinestretch}%
{#1}\small\normalsize} \spacingset{1}


\maketitle

\let\thefootnote\relax\footnotetext{$\dagger$ The first two authors contributed equally to this work.}
\let\thefootnote\relax\footnotetext{\textit{Address for correspondence:} Siyu Heng, Graduate Group in Applied Mathematics and Computational Science, University of Pennsylvania, David Rittenhouse Laboratory, 209 South 33rd Street, Philadelphia, PA 19104, USA. E-mail: \textsf{siyuheng@sas.upenn.edu}}

\begin{abstract}
In many observational studies, the interest is in the effect of treatment on bad, aberrant outcomes rather than the average outcome. For such settings, the traditional approach is to define a dichotomous outcome indicating aberration from a continuous score and use the Mantel-Haenszel test with matched data. For example, studies of determinants of poor child growth use the World Health Organization’s definition of child stunting being height-for-age z-score $\leq -2$. The traditional approach may lose power because it discards potentially useful information about the severity of aberration. We develop an adaptive approach that makes use of this information and asymptotically dominates the traditional approach. We develop our approach in two parts. First, we develop an aberrant rank approach in matched observational studies and prove a novel design sensitivity formula enabling its asymptotic comparison with the Mantel-Haenszel test under various settings. Second, we develop a new, general adaptive approach, \textit{the two-stage programming method}, and use it to adaptively combine the aberrant rank test and the Mantel-Haenszel test. We apply our approach to a study of the effect of teenage pregnancy on stunting. 
\end{abstract}

\noindent%
{\it Keywords:}  Aberrant rank; Causal inference; Design sensitivity; Optimization; Sensitivity analysis; Super-adaptivity.
\vfill

\spacingset{1.45} 

\section{Introduction}\label{sec:intro}

\pagenumbering{arabic}

\subsection{Examples of settings where interest is in the effect of treatment on aberrant response not average response}\label{sec:examples}

When evaluating the relative merits of competing treatment regimens, it can sometimes be more appropriate to focus on the effect of the treatment on poor outcomes (aberrant responses) rather than average outcomes. For example, malnutrition in children can cause both short- and long-term negative health outcomes and has been a long-standing global concern. According to the 2018 \textit{Global Nutrition Report}, undernutrition contributes to around $45\%$ of deaths among children under five. In studies on the effect of an exposure on child malnutrition, the most commonly used measurements of malnutrition are (1) stunting, (2) wasting, and (3) underweight. Stunting is defined as a child having a height less than or equal to 2 standard deviations below the mean height for the child's age (i.e. height-for-age z-score $\leq -2$), where the mean and standard deviation come from a reference population such as the World Health Organization (WHO) Multicenter Growth Reference Study (WHO, 2006). Similarly, wasting and underweight are defined as weight-for-age z-score $\leq -2$ and weight-for-height z-score $\leq -2$ respectively; see WHO (1986), Harris et al. (2001) and Bloss et al. (2004). When studying causal determinants of malnutrition, say stunting, researchers typically focus on the pattern of stunting instead of the average treatment effect on the height of children. This is because being slightly below the average height will not cause any serious problems, but stunted growth can lead to adverse consequences for the child including poor cognition and educational performance, low adult wages and lost productivity (WHO, 2017). The standard approach in studies of causal determinants of malnutrition is to consider stunting, wasting or underweight as binary outcomes, and to test the null that the treatment (potential causal determinant) does not affect that binary outcome for each individual through either Fisher's exact test for unstratified data or the Mantel-Haenszel test with stratified data (e.g., Brown et al., 1982; Walker et al., 1991; Bloss et al., 2004; Garrett and Ruel, 2005; Phuka et al., 2008; Null et al., 2018).

Numerous causal problems share a similar structure with that of the causal determinants of malnutrition, where we care about whether a certain treatment would change the pattern of some aberrant response (e.g. stunted growth) rather than the average treatment effect over the whole population; see Appendix G in the online supplementary materials for more examples. Rosenbaum and Silber (2008) referred to this as the aberrant effects of treatment problem. When studying aberrant effects, researchers typically choose a widely-used cut-off to define a dichotomous outcome (e.g. stunted or not; wasted or not) from a continuous response (e.g. height-for-age z-score; weight-for-age z-score), and then perform Fisher's exact test or the Mantel-Haenszel test. These traditional methods are both simple and convenient, but discard potentially useful information on the severity of aberrant response (e.g. exact height-for-age z-scores of children with stunted growth) and thus may fail to detect existing aberrant effects of the treatment.

\subsubsection{A matched observational study on the effect of teenage pregnancy on stunting}\label{example}

It has been realized that mothers who give birth at a very early age are more likely to have stunted children, and some studies have tried to investigate the potential causal relationship between teenage pregnancy and child stunting (e.g., Van de Poel et al., 2007; Darteh et al., 2014). We examine this causal problem with children's level data from the Kenya 2003 Demographic and Health Surveys (DHS). According to the Constitution of Kenya, ``adult” means an individual who has attained the age of eighteen years. Therefore, we define children with mother's age $\leq$ 18 years as treated individuals, and children with mother's age $\geq 19$ years as controls. We then take their height-for-age z-scores as the outcomes, where the z-score is with respect to the WHO Multicenter Reference Growth Study (WHO, 2006). Height-for-age z-scores are expressed in units equal to one standard deviation of the reference population's distribution. Recall that according to the WHO, low child height-for-age, or ``stunting", is defined as height-for-age z-score $\leq -2$. We conduct a matched observational study where we match each treated individual with controls for seven covariates: mother's highest education level; geographic district; household wealth index in quantiles; household's main source of drinking water; household's toilet facilities; sex; and children's age in years. Matching is a transparent and easily understandable way of adjusting for observed covariates and has been widely applied in observational studies (Rosenbaum 2002b; Hansen, 2004; Rubin, 2006; Stuart, 2010; Zubizarreta, 2012; Pimentel et al., 2015). We discarded 1466 records with missing or unspecified height-for-age z-scores, source of drinking water or toilet facilities from Kenya 2003 DHS data, leaving 4483 records. Among these 4483 children, there are 150 treated individuals and we matched each to three controls, 450 controls in total. Mother's education is categorized as no education, primary, secondary and higher. Geographic district is coded as eight dummy variables with respect to eight districts in Kenya: Central, Coast, Eastern, Nairobi, Northeastern, Nyanza, Riftvalley and Western. We use a similar method to Fink et al. (2011) to code quality of source of drinking water and toilet facilities. We used optimal matching using rank-based Mahalanobis distance with a propensity score caliper (Hansen and Klopfer, 2006). To evaluate the balance on baseline covariates before and after matching, we use standardized differences which are defined as a weighted difference in means divided by the pooled standard deviation between the treated and control groups before matching. The absolute standardized differences are all close to zero after matching, indicating good balance; see Figure~\ref{loveplot}.

\begin{figure}
\centering
    \caption{Covariate imbalances before and after matching with three controls. The plot reports the absolute standardized difference before and after matching of each covariate. The two dotted vertical lines are 0.1 and 0.2 cut-offs.}
      \label{loveplot}
  \includegraphics[width=0.6\linewidth]{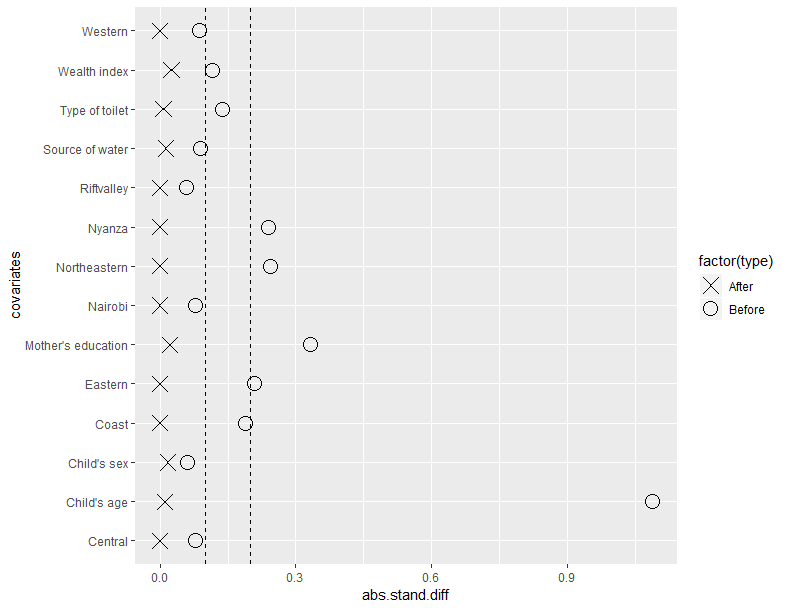}
\end{figure}

\subsection{Our contributions}\label{sec: our contribution}

Previous work on inference for aberrant effects of treatment has considered randomized trials where there is no unmeasured confounding by design (Rosenbaum and Silber, 2008). In an observational study, we typically worry about unmeasured confounding and would like to have an approach that has good power to detect an effect that is insensitive to a moderate amount of unmeasured confounding (Rosenbaum, 2004, 2010a). In this paper, we develop an adaptive approach for inference about aberrant treatment effects from matched observational studies that asymptotically uniformly dominates the traditional approach of performing the Mantel-Haenszel test based on a dichotomous outcome of aberrant/not aberrant. 

Our new approach is developed in two parts. In the first part, we introduce the aberrant null hypothesis of no treatment effect for matched studies which is especially suitable for studying aberrant treatment effects, and then introduce the aberrant rank test for matched studies along with its sensitivity analysis and study its asymptotic power. The aberrant rank test takes the form of the sum of aberrant ranks among all the treated units, with the Wilcoxon rank sum test as a special case. It is more powerful than the Mantel-Haenszel test for testing the aberrant null of no treatment effect in many settings because it considers not only the incidence of aberrant response, but also the severity of aberration. We formally demonstrate this through proving a novel design sensitivity formula. Design sensitivity measures the limiting robustness of a test to hidden bias in an observational study as the sample size increases: larger design sensitivity corresponds to greater asymptotic robustness to hidden bias (Rosenbaum, 2004, 2010a). Our new design sensitivity formula allows us to asymptotically compare the performances of the aberrant rank test and the Mantel-Haenszel test for testing the aberrant null under various settings. We also validate that our asymptotic findings provide good guidance for realistic sample sizes in simulation studies. We illustrate that whether we should use the aberrant rank test or the Mantel-Haenszel test depends on the unknown data generating process, and making the wrong choice can substantially harm the performance of a sensitivity analysis. The proof of the design sensitivity formula involves a new technique that uses empirical process to analyze the asymptotics of matched observational studies as the number of matched strata grows and can be of independent interest.

In the second part, we develop a novel, general adaptive approach called ``the two-stage programming method" to combine two tests in observational studies such that the design sensitivity of the resulting adaptive test is always greater than or equal to maximum of the design sensitivities of the two component tests performed in isolation, regardless of the underlying data generating distribution. We refer to this newly discovered phenomenon as ``super-adaptivity". Thus, applying our new adaptive approach to combine the aberrant rank test and the Mantel-Haenszel test uniformly dominates the traditional approach based solely upon the Mantel-Haenszel test in terms of the design sensitivity. We evaluate our adaptive test via simulations and show that under various settings it achieves close to the better one of its components (i.e., the aberrant rank test and the Mantel-Haenszel test) for realistic sample sizes, and therefore avoids potentially drastic reductions in power stemming from making the wrong choice between the two component tests.

The first adaptive approach for sensitivity analysis was introduced in Rosenbaum (2012). It has been applied to or modified for various settings (e.g., Zubizarreta et al., 2014; Rosenbaum, 2015; Rosenbaum and Small, 2017; Ertefaie et al., 2018; Zhao et al., 2018; Shauly‐Aharonov, 2020). Our new adaptive approach goes beyond this traditional adaptive approach in two aspects. First, the traditional adaptive approach only works if all of the component test statistics are stochastically dominated by a known distribution in a matched observational study. While commonly used tests statistics for binary outcomes or general outcomes in pair-matched studies have this property, most commonly used test statistics for general outcomes in matched studies that allow for multiple controls or full matching do not have this stochastic dominance property, including the aberrant rank test, the Wilcoxon rank sum test, the Hodges-Lehmann aligned rank test and the Huber-Maritz m-tests (Gastwirth et al., 2000; Rosenbaum, 2002b, 2007b). In contrast, our new adaptive approach covers most of the existing testing scenarios in matched studies as it works for any sum statistics and various matching techniques, including pair matching, matching with multiple controls and full matching. Second, our new adaptive approach uniformly dominates the traditional approach in terms of the design sensitivity, which is an immediate consequence of the super-adaptivity property enjoyed by the new adaptive approach.

\section{Notation and reviews}

\subsection{Matching-based randomization inference and sensitivity analysis}

Suppose there are $I$ strata $i=1,\dots,I$. Each stratum contains $m$ ($m\geq 2$) individuals (e.g. children) where one individual received treatment and the other $m-1$ individuals received control. Let $Z_{ij}=1$ if individual $j$ in stratum $i$ received treatment (e.g. mother's age $\leq$ 18 years), otherwise let $Z_{ij}=0$ (e.g. mother's age $\geq$ 19 years). Denote the collection of treatment assignments as $\mathbf{Z} = (Z_{11}, \dots, Z_{Im})^{T}$. Let $\mathcal{Z}$ be the set of all possible values of $\mathbf{Z}$ where $\mathbf{Z} \in \mathcal{Z}$ if and only if $\sum_{j=1}^{m}Z_{ij}=1$ for all $i$. Let $|S|$ denote the number of elements of a finite set $S$, then we have $|\mathcal{Z}|=m^{I}$. Let $\mathbf{x}_{ij}$ and $u_{ij}$ denote the observed covariates and an unobserved covariate respectively for each individual $j$ in stratum $i$. Typically, each stratum $i$ is formed by matching on the observed covariates such that $\mathbf{x}_{ij}= \mathbf{x}_{ij^{\prime}}$ or $\mathbf{x}_{ij}\approx \mathbf{x}_{ij^{\prime}}$. However, matching cannot directly adjust for the unobserved covariate, so $u_{ij} \neq u_{i j^{\prime}}$ is possible. Under the potential outcomes framework, if individual $j$ in stratum $i$ received treatment ($Z_{ij}=1$), we observe the potential response $r_{Tij}$; otherwise ($Z_{ij}=0$), we observe $r_{Cij}$. That is, the observed response (e.g., the observed height-for-age z-score) for individual $ij$ is $R_{ij}=Z_{ij}r_{Tij}+(1-Z_{ij})r_{Cij}$ (Neyman, 1923; Rubin, 1974). Write $\mathcal{F}=\{(r_{Tij}, r_{Cij}, \mathbf{x}_{ij}, u_{ij}), \ i=1,\dots,I, \ j=1,\dots,m\}$. Denote the collection of responses as $\mathbf{R}=(R_{11}, \dots, R_{Im})^{T}$. Fisher's sharp null hypothesis of no treatment effect asserts that $H_{0}: r_{Tij}=r_{Cij},  \ \forall \ i, j$. 

In a randomized experiment, where we can assume random treatment assignment in each stratum, i.e., $\mathbb{P}(\mathbf{Z}=\mathbf{z}\mid \mathcal{F}, \mathcal{Z})=1/|\mathcal{Z}|=m^{-I}$ for all $\mathbf{z} \in \mathcal{Z}$, the significance level of a test statistic $T$ being greater than or equal to the observed value $t$ under the null hypothesis can be calculated via randomization inference:
\begin{equation}\label{test:permutation}
	\mathbb{P}(T \geq t\mid \mathcal{F}, \mathcal{Z})=\sum_{\mathbf{z}\in \mathcal{Z}}\mathbbm{1}\{T(\mathbf{z}, \mathbf{R}) \geq t\} \cdot \mathbb{P}(\mathbf{Z}=\mathbf{z}\mid \mathcal{F}, \mathcal{Z})=\frac{|\{\mathbf{z}\in \mathcal{Z}: T(\mathbf{z}, \mathbf{R}) \geq t\}|}{|\mathcal{Z}|},
\end{equation}
where $\mathbbm{1}(A)=1$ if $A$ is true, and $\mathbbm{1}(A)=0$ otherwise. For large $I$, (\ref{test:permutation}) can be approximated via asymptotic normality of the null distribution of $T$ (Rosenbaum, 2002b).

In an observational study, however, it is unrealistic to assume that the treatment is randomly assigned in each stratum even if we have matched on all the observed covariates, due to the possible presence of unobserved covariates (unmeasured confounders). A sensitivity analysis tries to determine how departures from random assignment of treatment would affect inferences on treatment effects. Let $\pi_{ij}=\mathbb{P}(Z_{ij}=1\mid \mathcal{F})$ denote the probability that, in the population before matching, individual $ij$ will receive treatment. We follow the widely-used Rosenbaum sensitivity analysis framework (Rosenbaum, 2002b) which considers a logit model linking $\pi_{ij}$ to $\mathbf{x}_{ij}$ and normalized $u_{ij}$: 
\begin{equation} \label{eqn: rosenbaum sens model}
    \log \left(\frac{\pi_{ij}}{1-\pi_{ij}} \right)=\theta(\mathbf{x}_{ij})+\gamma u_{ij}, \quad \text{where $u_{ij}\in [0,1]$,}
\end{equation}
where $\theta(\mathbf{x}_{ij})$ is an arbitrary unknown function of $\mathbf{x}_{ij}$ and $\gamma\geq 0$ is a sensitivity parameter. The assumption $u_{ij}\in [0,1]$ is no more restrictive than assuming a bounded support of $u_{ij}$ and is only imposed to make $\gamma$ more interpretable (Rosenbaum, 2002b). Under (\ref{eqn: rosenbaum sens model}), if we imagine that each $(r_{Tij}, r_{Cij})$ is drawn from some super-population model (only for interpretation, not necessary for randomization inference) and is confounded by the set of covariates $(\mathbf{x}_{ij}, u_{ij})$, then the strong ignorability assumption (Rosenbaum and Rubin, 1983) would hold were $(\mathbf{x}_{ij},u_{ij})$ to all be measured (but in fact $u_{ij}$ is not measured), i.e., we have $(r_{Tij}, r_{Cij}) \indep Z_{ij} \mid \mathbf{x}_{ij}, u_{ij}$ and $0<\mathbb{P}(Z_{ij}=1\mid \mathbf{x}_{ij}, u_{ij})<1$ for all $i, j$. Here $u_{ij}$ can also represent an aggregate measurement of all potential, perhaps more than one, unmeasured confounders $u_{ij1}, u_{ij2}, \dots$. For example, if $\log \{ \pi_{ij}/(1-\pi_{ij})\}=\theta(\mathbf{x}_{ij})+g(u_{ij1}, u_{ij2}, \dots)$ for some function $g$ with bounded support $[0, \xi]$, then (\ref{eqn: rosenbaum sens model}) holds with setting $u_{ij}=\xi^{-1}g(u_{ij1}, u_{ij2},\dots)$ and $\gamma=\xi$.

Under model (\ref{eqn: rosenbaum sens model}), it is straightforward to show that for any two individuals $ij$ and $ij^{\prime}$ within the same stratum, the ratio of their odds of receiving the treatment is bounded by the sensitivity parameter $\Gamma=\exp(\gamma) \geq 1$:
\begin{equation*}
	\frac{1}{\Gamma}\leq \frac{\pi_{ij}(1-\pi_{ij^{\prime}})}{\pi_{ij^{\prime}}(1-\pi_{ij})}\leq \Gamma, \quad \text{for all $i, j, j^{\prime}$ with $\mathbf{x}_{ij}=\mathbf{x}_{ij^{\prime}}$}.
\end{equation*}
Note that $\Gamma=1$ is equivalent to random assignment. The more $\Gamma$ departs from 1, the more the treatment assignment potentially departs from random assignment. It is then easy to show that model (\ref{eqn: rosenbaum sens model}) implies the following biased treatment assignment probability after matching, assuming $\mathbf{x}_{ij}=\mathbf{x}_{ij^{\prime}}$: 
\begin{equation*}
	\mathbb{P}(\mathbf{Z}=\mathbf{z}\mid \mathcal{F}, \mathcal{Z})=\prod_{i=1}^{I}\frac{\exp(\gamma \sum_{j=1}^{m}z_{ij}u_{ij})}{\sum_{j=1}^{m}\exp(\gamma u_{ij})}, \quad \mathbf{z} \in \mathcal{Z}, \ 0 \leq u_{ij} \leq 1.
\end{equation*}
Then  researchers usually look at \textit{the worst-case p-value}, which is defined as the largest p-value given the sensitivity parameter $\Gamma$ over all possible arrangements of unmeasured confounders $u_{ij}$. For example, for a one-sided test, the worst-case p-value reported by a test statistic $T$ given its observed value $t$ with the sensitivity parameter $\Gamma=\exp(\gamma)$ is $\max_{0\leq u_{ij} \leq 1}\mathbb{P}(T \geq t \mid \mathcal{F}, \mathcal{Z})=\max_{0 \leq u_{ij}\leq 1}\sum_{\mathbf{z}\in \mathcal{Z}}\mathbbm{1}\{T(\mathbf{z}, \mathbf{R}) \geq t\} \cdot \mathbb{P}(\mathbf{Z}=\mathbf{z} \mid \mathcal{F}, \mathcal{Z})$. In practice, researchers gradually increase $\Gamma$, compute the worst-case p-value for each $\Gamma$, and report the \textit{sensitivity value}, which is defined as the largest $\Gamma$ such that the corresponding worst-case p-value exceeds some prespecified level $\alpha$ and informs the magnitude of hidden bias required to alter the causal conclusion drawn from the primary analysis assuming no unmeasured confounding (Zhao, 2018). For other models of sensitivity analysis, see Shepherd et al. (2006), McCandless et al. (2007), Mitra and Heitjan (2007), Hosman et al. (2010), Keele and Quinn (2017), VanderWeele and Ding (2017) and Zhao (2018).


\subsection{Power of a sensitivity analysis and design sensitivity}

The power of a test is the probability that the test will successfully reject the null hypothesis and is calculated under some alternative. In parallel, the power of a sensitivity analysis is the probability that the test will correctly reject the null, for any possible arrangements of unmeasured confounders given some $\Gamma \geq 1$, under some alternative. To be more specific, for a fixed $\Gamma$, the power of an $\alpha$ level sensitivity analysis using a test statistic $T$ is calculated as the probability that the worst-case p-value corresponding to $T$ falls below $\alpha$ when conducting a sensitivity analysis at $\Gamma$. When calculating the power, we need to specify a data generating process for the alternative. Following previous work, we consider power under the alternative of a ``favorable situation" that there is a treatment effect of a specified magnitude and no hidden bias (Rosenbaum, 2004, 2010a). Even though there is no hidden bias in this favorable situation, we would typically not know that for sure in an observational study and would prefer a test statistic with a higher power of sensitivity analysis for plausible values of $\Gamma \geq 1$ (i.e., with high degree of insensitivity to hidden bias). This strategy of calculating the power is more appropriate than those assuming alternatives of both a treatment effect and a bias in treatment assignment. For example, suppose that we instead use the alternative of a small treatment effect and a large bias in treatment assignment, then a test statistic that has a high chance of rejecting the null hypothesis of no treatment effect (i.e., high statistical power) with small or moderate $\Gamma$ may not be favorable because we cannot tell if its high statistical power results from its detection of the actual treatment effect or the underestimation of the magnitude of hidden bias. This ambiguity will not occur when calculating the power under the favorable situation in which the researchers always seek to reject the null under plausible $\Gamma$. Therefore, calculating the power of a sensitivity analysis under the favorable situation provides a logically consistent way to compare competing test statistics in observational studies. See Hansen et al. (2014) and Rosenbaum (2017; Chapter 10) for detailed discussion on this.


Typically, under some regularity assumptions (varying with different matching structures and tests) on the data generating process of responses $\mathbf{R}$, there is a number $\widetilde{\Gamma}$ called the design sensitivity, such that as the sample size $I \rightarrow \infty$, the power of a sensitivity analysis goes to 1 if the analysis is performed with $\Gamma < \widetilde{\Gamma}$, and the power goes to 0 if performed with $\Gamma > \widetilde{\Gamma}$. That is, $\widetilde{\Gamma}$ refers to the sharp transition of consistency of a test in a sensitivity analysis (Rosenbaum, 2004). The design sensitivity gives us a powerful and elegant tool to asymptotically compare two test statistics or two study designs under each data distribution model - the test or the study design with a larger $\widetilde{\Gamma}$ is asymptotically more robust to unmeasured confounders. Besides its mathematical elegance, the design sensitivity has been shown to be a powerful tool in practical studies (Stuart and Hanna, 2013; Zubizarreta et al., 2013).

\section{The traditional approach: the Mantel-Haenszel test}\label{sec:M-H test}

In settings such as those described in Section~\ref{sec:examples}, there is a subset $\mathcal{A} \subset \mathbb{R}$ such that any $R_{ij} \in \mathcal{A}$ is considered as an aberrant response, and researchers care about whether the treatment would change the pattern of aberrant response instead of the average treatment effect. In these settings, a typical approach is to define a dichotomous outcome $\widetilde{R}_{ij}=\mathbbm{1}(R_{ij}\in \mathcal{A})$ indicating whether individual $ij$ had an aberrant outcome or not. For example, in Section~\ref{example} where we focus on whether child $ij$ showed stunted growth (i.e. $R_{ij}\leq -2$), we can let $\mathcal{A}=(-\infty, -2]$ and the dichotomous observed outcome indicating stunted growth $\widetilde{R}_{ij}=\mathbbm{1}(R_{ij}\leq -2)$ is binary. Let $\widetilde{r}_{Tij}=\mathbbm{1}(r_{Tij}\in \mathcal{A})$ and $\widetilde{r}_{Cij}=\mathbbm{1}(r_{Cij}\in \mathcal{A})$, we have $\widetilde{R}_{ij}=\mathbbm{1}(R_{ij}\in \mathcal{A})=Z_{ij} \mathbbm{1}(r_{Tij}\in \mathcal{A})+(1-Z_{ij}) \mathbbm{1}(r_{Cij}\in \mathcal{A})=Z_{ij}\widetilde{r}_{Tij}+(1-Z_{ij}) \widetilde{r}_{Cij}$. Then researchers focus on a categorized Fisher's sharp null of no treatment effect $\widetilde{H}_{0}: \widetilde{r}_{Tij}=\widetilde{r}_{Cij}, \ \forall \ i,j$, that is, whether individual $ij$ would show aberrant response or not will not be affected by whether he or she received the treatment or not. Note that $\widetilde{H}_{0}$ does not imply any information about the severity of aberration. It is clear that if $H_{0}$ holds true, so does $\widetilde{H}_{0}$, and if $\widetilde{H}_{0}$ is false, so is $H_{0}$. 

The traditional approach then performs the Mantel-Haenszel test (Mantel and Haenszel, 1959), which can be regarded as an analogue of Fisher's exact test when there are two or more stratum, $I \geq 2$. Formally, the Mantel-Haenszel test utilizes the statistic $T_{\text{M-H}}=\sum_{i=1}^{I}\sum_{j=1}^{m}Z_{ij}\mathbbm{1}(R_{ij}\in \mathcal{A})=\sum_{i=1}^{I}\sum_{j=1}^{m}Z_{ij}\widetilde{R}_{ij}$, which is the number of aberrant responses among treated individuals. For matched pairs, the Mantel-Haenszel test reduces to McNemar's test (Cox, 2018). In a randomized experiment, we can use (\ref{test:permutation}) to conduct randomization inference. In an observational study, we can use the related result in Rosenbaum (2002b; Chapter 4) to conduct sensitivity analyses: under matching with $m-1$ controls, for any $t$, the one-sided worst-case p-value $\max_{0 \leq u_{ij}\leq 1} \mathbb{P}(T \geq t\mid \mathcal{F}, \mathcal{Z}) = \mathbb{P}(T^{+} \geq t \mid \mathcal{F}, \mathcal{Z})$ where $T^{+}$ is the sum of $I$ independent Bernoulli random variables $B_{1}, \dots, B_{I}$, with $B_{i}$ taking value one with probability $p_{i}^{+}=\{\Gamma \sum_{j=1}^{m}\mathbbm{1}(R_{ij} \in \mathcal{A})\}/\{(\Gamma-1) \sum_{j=1}^{m}\mathbbm{1}(R_{ij} \in \mathcal{A})+m\}$ and value zero with probability $1-p_{i}^{+}$. Therefore, we have as $I\rightarrow \infty$, for each fixed $t$,
\begin{equation}\label{sensianalysisforM-H}
   \max_{0 \leq u_{ij}\leq 1} \mathbb{P}(T \geq t\mid \mathcal{F}, \mathcal{Z}) = \mathbb{P}(T^{+} \geq t\mid \mathcal{F}, \mathcal{Z}) \simeq 1-\Phi \Bigg(\frac{t-\sum_{i=1}^{I}p_{i}^{+}}{\sqrt{\sum_{i=1}^{I}{p_{i}^{+}(1-p_{i}^{+})}}}\Bigg),
\end{equation}
where $\Phi$ is the distribution function of standard normal distribution and `$\simeq$' denotes that two sequences are asymptotically equal. We can then use (\ref{sensianalysisforM-H}) to report worst-case p-values for various sensitivity parameters $\Gamma$. The design sensitivity of the Mantel-Haenszel test has also been derived in Rosenbaum and Small (2017).

The Mantel-Haenszel test is simple and convenient but can lose power from ignoring information about the magnitude of aberration. 

\section{An aberrant rank approach and its comparison with the traditional approach}

\subsection{The aberrant null and the aberrant rank test}\label{sec: aberrant rank test}
Although $H_{0}$ and $\widetilde{H}_{0}$ are widely used null hypotheses in randomized experiments and observational studies, they do not best capture the hypotheses of interest in studying the causal determinants of aberrant response when severity of aberration matters. To better capture the hypothesis of interest, Rosenbaum and Silber (2008) introduced the aberrant null hypothesis of no effect of treatment on individuals who would have an aberrant response under either the treatment or control. Formally, as in Section~\ref{sec:M-H test}, let $\mathcal{A}$ be a subset of $\mathbb{R}$ that defines an aberrant response. Then the null hypothesis of no aberrant effect states that
\begin{equation*}
	H_{0}^{\mathcal{A}}:r_{Tij}=r_{Cij}, \ \forall \ i,j, \ \text{if either $r_{Tij} \in \mathcal{A}$ or $r_{Cij}\in \mathcal{A}$}.
\end{equation*}
It is easy to see that $H_{0}^{\mathcal{A}}$ is a weaker hypothesis than Fisher's sharp null $H_{0}$, in the sense that $H_{0}$ implies $H_{0}^{\mathcal{A}}$, but the converse is not true. And we can also see that $H_{0}^{\mathcal{A}}$ is a stronger hypothesis than the categorized Fisher's sharp null $\widetilde{H}_{0}$, in the sense that $H_{0}^{\mathcal{A}}$ implies $\widetilde{H}_{0}$, but the converse is not true. That is, $H_{0}^{\mathcal{A}}$ is a null hypothesis that lies between $H_{0}$ and $\widetilde{H}_{0}$. 

Let us consider studying a potential causal determinant of stunting to illustrate why $H_{0}^{\mathcal{A}}$ is a more appropriate null hypothesis to test when the pattern of aberration is our main focus. Note that all the alternatives can be classified into the following four cases: (i) \textbf{Case 1}: $r_{Tij} \in \mathcal{A}$, $r_{Cij} \notin \mathcal{A}$, i.e., treatment will cause stunting for child $ij$; (ii) \textbf{Case 2}: $r_{Tij} \notin \mathcal{A}$, $r_{Cij} \in \mathcal{A}$, i.e., treatment will prevent stunting for child $ij$; (iii) \textbf{Case 3}: $r_{Tij} \in \mathcal{A}$, $r_{Cij} \in \mathcal{A}$, and $r_{Tij} \neq r_{Cij}$, i.e., treatment will not prevent stunting for child $ij$, but it will affect the severity of stunting; (iv) \textbf{Case 4}: $r_{Tij} \notin \mathcal{A}$, $r_{Cij} \notin \mathcal{A}$, and $r_{Tij} \neq r_{Cij}$, i.e., child $ij$ will not show stunted growth no matter whether he or she received treatment or not. Thus, $H_{0}^{\mathcal{A}}$ is against Cases 1-3, while $H_{0}$ is against all the four cases and $\widetilde{H}_{0}$ is against only Cases 1 and 2. Our goal is to decide whether the treatment affects stunted growth. It is clear that in Cases 1 and 2, the treatment affects stunted growth (causing stunting in Case 1, preventing stunting in Case 2). Case 3 also indicates the treatment affects stunted growth, since although the treatment will not prevent a child from being stunted, it will affect the severity of stunting, i.e., it will aggravate or alleviate the child's stunting growth which could have a huge impact on the child. In Case 4, the treatment does not affect stunted growth since the child will be healthy and non-stunted no matter whether he or she is exposed to the treatment or control. Consideration of these four cases shows that $H^{\mathcal{A}}_{0}$ is a more appropriate null hypothesis than $H_{0}$ and $\widetilde{H}_{0}$ because it contains in the alternative the three cases where treatment affects stunted growth but keeps in the null the fourth case where treatment does not affect stunted growth.  

In this paper, our argument focuses on $\mathcal{A}$ with the form of $\mathcal{A}=[c,+\infty)$ (or equivalently, $(c, +\infty$)) for some $c \in \mathbb{R}$. In these settings, there is a threshold value $c$ indicating aberration, which is common in practical research. The argument works in parallel with $\mathcal{A}=(-\infty, c]$ and $\mathcal{A}=(-\infty, c)$. For example, according to the WHO, stunting is defined as height-for-age z-score $\leq -2$, and in this case $\mathcal{A}=(-\infty, c]$ with $c=-2$.

Rosenbaum and Silber (2008) introduced the aberrant rank test for randomized experiments with unmatched data, and Small et al. (2013) considered the aberrant rank in case-referent studies. In this paper, we derive a new aberrant rank test for matched observational cohort studies. We define $q(v \mid \mathbf{R})=\sum_{i^{\prime}=1}^{I}\sum_{j^{\prime}=1}^{m}\mathbbm{1}(v \geq R_{i^{\prime}j^{\prime}}\geq c)$ and refer to $q(R_{ij}\mid \mathbf{R})$ as the aberrant rank of individual $ij$. There are some features worth mentioning. First, the aberrant rank $q(R_{ij}\mid \mathbf{R})$ depends on all the responses, including those that are not in the same stratum as $R_{ij}$. Second, if individual $ij$ did not show aberrant response, $q(R_{ij}\mid \mathbf{R})$ is zero, and if he or she did show aberrant response, $q(R_{ij}\mid \mathbf{R})$ takes the rank of $R_{ij}$ among all the responses of individuals with aberrant response. Third, $q(v\mid \mathbf{R})$ is monotonic in $v$.

Next, we define the aberrant rank test for a stratified (e.g., matched) study as 
	\begin{equation}\label{def:abe}
		T_{\text{abe}}=\sum_{i=1}^{I}\sum_{j=1}^{m}Z_{ij}\ q(R_{ij}\mid \mathbf{R}),
	\end{equation}
which is the sum of all the aberrant ranks over all treated individuals. When $c=-\infty$, $T_{\text{abe}}$ reduces to the Wilcoxon rank sum test. Under the null hypothesis of no aberrant effects $H_{0}^{\mathcal{A}}$, $q(R_{ij}\mid \mathbf{R})$ is fixed. In a randomized experiment, we can use (\ref{test:permutation}) along with its asymptotic approximation to report p-values. In a sensitivity analysis, unlike the Mantel-Haenszel test, in general we cannot find a known distribution to bound the distribution of $T_{\text{abe}}$. However, under pair matching or matching with multiple controls, utilizing the asymptotic separability algorithm in Gastwirth et al. (2000), for any given $t$, we can approximate the one-sided worst-case p-value $\max_{0 \leq u_{ij} \leq 1}\mathbb{P}(T_{\text{abe}} \geq t\mid \mathcal{F}, \mathcal{Z})$ under $H_{0}^{\mathcal{A}}$. Let $b\in \{1,\dots,m-1\}=:[m-1]$, and $\overline{\overline{\mu}}_{ib}$ and $\overline{\overline{\nu}}_{ib}$ be the null expected value and variance of $\sum_{j=1}^{m}Z_{ij}\ q(R_{ij}\mid \mathbf{R})$ under $u_{i1}=\dots =u_{ib}=0$ and $u_{i,b+1}=\dots=u_{im}=1$ with different values of $b$ respectively:
\begin{equation*}
	\overline{\overline{\mu}}_{ib}=\frac{\sum_{j=1}^{b} q(R_{i(j)}\mid \textbf{R})+\Gamma \sum_{j=b+1}^{m} q(R_{i(j)}\mid \textbf{R})}{b+\Gamma(m-b)}, \quad i=1,\dots,I, \ b=1, \dots, m-1,
\end{equation*}
\begin{equation*}
	\overline{\overline{\nu}}_{ib}=\frac{\sum_{j=1}^{b} q^{2}(R_{i(j)}\mid \textbf{R})+\Gamma \sum_{j=b+1}^{m} q^{2}(R_{i(j)}\mid \textbf{R})}{b+\Gamma(m-b)}-\overline{\overline{\mu}}_{ib}^{2}, \quad i=1,\dots,I, \ b=1, \dots, m-1,
\end{equation*}
where we rearrange $R_{i1}, \dots, R_{im}$ as $R_{i(1)} \leq \dots \leq R_{i(m)}$. Let $\overline{\overline{\mu}}_{i}=\max_{b\in [m-1]}\overline{\overline{\mu}}_{ib}$, $B_{i}=\{b: \overline{\overline{\mu}}_{ib}=\overline{\overline{\mu}}_{i}, b\in [m-1] \}$ and $\overline{\overline{\nu}}_{i}=\max_{b \in B_{i}}\overline{\overline{\nu}}_{ib}$. Then the one-sided worst-case p-value can be approximated via
\begin{equation*}
	\max_{0 \leq u_{ij} \leq 1}\mathbb{P}(T_{\text{abe}} \geq t\mid \mathcal{F}, \mathcal{Z})\simeq 1-\Phi \Bigg(\frac{t-\sum_{i=1}^{I} \overline{\overline{\mu}}_{i}}{\sqrt{\sum_{i=1}^{I}\overline{\overline{\nu}}_{i}}} \Bigg) \quad \text{as $I\rightarrow \infty$}.
\end{equation*}
Letting $\xi_{\alpha}=\sum_{i=1}^{I}\overline{\overline{\mu}}_{i}+\Phi^{-1}(1-\alpha)\sqrt{\sum_{i=1}^{I}\overline{\overline{\nu}}_{i}}$, $\Psi_{\Gamma, I}=\mathbb{P}(T_{\text{abe}} \geq \xi_{\alpha}\mid \mathcal{Z})$ is the power of a one-sided $\alpha$-level sensitivity analysis conducted with $\Gamma$ of the aberrant rank test $T_{\text{abe}}$. Typically, $\Psi_{\Gamma, I}$ is computed with respect to draws from a data generating process in the favorable situation in which there is no hidden bias and there is a treatment effect.  

Recall that the aberrant null $H_{0}^{\mathcal{A}}$ is a more appropriate null hypothesis to test than both $\widetilde{H}_{0}$ and $H_{0}$ when there exists a designated cut-off for what constitutes an aberrant outcome but more severely aberrant outcomes are worse than less severely aberrant outcomes. Note that both the aberrant rank test and the Mantel-Haenszel test are valid for testing the aberrant null in the sense that they both have pivotal distributions if the aberrant null holds. When testing the aberrant null, among these two candidate tests, intuitively the aberrant rank test should be more appealing and natural to choose since its null distribution also incorporates the fact that the severity of aberration is fixed under the aberrant null, while the null distribution of the Mantel-Haenszel test under the aberrant null is the same as that under the weaker null $\widetilde{H}_{0}$ which only looks at the dichotomized outcome indicating aberration/non-aberration. However, as we will show in Sections \ref{sec: the design sennsitivity formula} and \ref{sec:computeds}, when testing the aberrant null in observational studies, although the aberrant rank test is indeed more powerful than the Mantel-Haenszel test in many cases, there also exist settings under which the Mantel-Haenszel test is instead more powerful. This motivates us to develop a new adaptive testing procedure in Section~\ref{sec:adaptive} and use it to combine the aberrant rank test and the Mantel-Haenszel test to guarantee that the resulting adaptive approach for testing the aberrant null uniformly dominates the Mantel-Haenszel test in large samples and performs well in finite samples.
 
\subsection{Design sensitivity formula of the aberrant rank test}\label{sec: the design sennsitivity formula}

There is an extensive literature on deriving design sensitivity formulas for various test statistics in matched observational studies. These design sensitivity formulas provide powerful tools for asymptotically evaluating the performances of various tests in a sensitivity analysis. However, all previous approaches either require the use of pair matching (Rosenbaum, 2010a, 2011; Hansen et al., 2014; Rosenbaum and Small, 2017; Fogarty et al., 2019; Howard and Pimentel, 2019), or require a particular structure for the test statistic to which many rank statistics do not conform (Rosenbaum, 2013, 2014, 2016). There are no design sensitivity formulas for popular test statistics, such as the Wilcoxon rank sum test and the Hodges-Lehmann aligned rank test, and the aberrant rank test discussed above; more generally, currently methods cannot handle test statistics where there are matched strata with multiple controls and ranking is done across matched strata as this induces dependence between matched strata that are typically assumed to be independent in many sensitivity analyses. In this section, we derive a novel design sensitivity formula for the aberrant rank test, of which the Wilcoxon rank sum test is a special case. Our proof technique involves applying empirical process to matched data. To the best of our knowledge, this is the first application of such machinery to the sensitivity analysis of matched observational studies and can potentially be used to study many other rank tests in the matched setting. 

We need a few regularity and causal assumptions of responses $\mathbf{R}$ under the alternative for deriving the design sensitivity. Without loss of generality, suppose that in each stratum $i$, unit $j=1$ received treatment and the rest $j=2,\dots,m$ received control. If not, we just need to simply reassign index $j=1$ to the treated in each stratum. 

\begin{Assumption}[i.i.d. strata]\label{ass:iid}
The responses from each stratum $i$ - $(R_{i1}, \dots, R_{im})$ are i.i.d. realizations from a continuous multivariate distribution $F(x_{1}, \dots, x_{m})$.
\end{Assumption}

Assumption~\ref{ass:iid} implies existence of a super-population model for the potential outcomes, which is common in deriving design sensitivity values (Rosenbaum 2004, 2010a). We remark that this assumption as well as others below are only used to derive the design sensitivity formula; they are not required for the validity of the aberrant rank test defined in Section~\ref{sec: aberrant rank test} or the adaptive approach introduced in Section~\ref{sec:adaptive}.

Suppose that $F(x_{1}, \dots, x_{m})$ has marginal cumulative distributions $F_{1}(x_{1}), \dots, F_{m}(x_{m})$ and densities $f_{1}(x_{1}), \dots, f_{m}(x_{m})$. Let $F_{(1)}, \dots, F_{(m)}$ be the associated marginal distribution with densities $f_{(1)}, \dots, f_{(m)}$ of $R_{i(1)}\leq \dots \leq R_{i(m)}$, the ordered responses within stratum $i$. The next two assumptions place regularity conditions on the marginal distributions of $f_{j}$ and its ordered counterpart $f_{(j)}$.

\begin{Assumption}[Connectedness of the support]\label{ass:nonzeromeasure}
	For $j=1,\dots,m$, let $s_{j}=\sup \{t: \mathbb{P}(R_{ij} \geq t)>0\}$ ($s_{j}$ can be $\infty$). Then $s_{j} > c$ and $f_{j}(t) > 0$ for any $t \in [c,s_j)$.
\end{Assumption}

\begin{Assumption}[`Positive' and `non-extreme' treatment effect]\label{ass:treateffect}
Let $s=\max_{j}s_{j}$. Then for any $t \in  [c,s)$, we have $\mathbb{P}(R_{i1}\geq t)\geq \frac{1}{m}\sum_{j=1}^{m}\mathbb{P}(R_{ij}\geq t)$, and there exists an open interval $\mathcal{I} \subset [c,s)$ such that strict inequality holds for any $t \in \mathcal{I}$. Moreover, $\mathbb{P}(R_{i(m)}>R_{i1}\geq t)>0$ for some $t \geq c$.
\end{Assumption}

In words, Assumption~\ref{ass:nonzeromeasure} states that aberrant responses can be observed with non-zero probability, and the support of the distribution function of each individual's aberrant response is a connected set. Assumption~\ref{ass:treateffect} states that the aberrant response of the treated is stochastically larger than the average of the distribution function of all the responses within the same stratum, i.e., $ \mathbb{P}(R_{i1}\geq t)\geq \frac{1}{m} \sum_{j=1}^{m}\mathbb{P}(R_{ij}\geq t)$. The remaining part of Assumption~\ref{ass:treateffect} is only intended to prevent design sensitivities from equaling 1 or going to $\infty$; see the proof of Theorem~\ref{dsformula} in the Appendix A for details.

We give some examples to show that Assumptions 1-3 hold for many widely considered treatment effect models. In Examples 1-3 listed below, we assume that $(R_{i1},\dots,R_{im})$ are i.i.d. continuous random vectors, and we assume that for each $i$, $R_{i2}, \dots, R_{im}$ are identically distributed with the support equalling $\mathbb{R}$, and correlation of $R_{ij_{1}}$ and $R_{ij_{2}}$ is neither $1$ or $-1$ for any two distinct $j_{1}, j_{2} \in \{1,\dots,m\}$. `$\sim$' means two distributions are equal. We consider the following three examples: (i) \textbf{Example 1} (Additive treatment effects): $R_{i1} \sim R_{i2}+\beta$ for some $\beta > 0$ with $c \in \mathbb{R}$; (ii) \textbf{Example 2}(Multiplicative treatment effects): $R_{i1} \sim \delta \cdot R_{i2}$ for some $\delta > 1$ with $c > 0$; (iii) \textbf{Example 3} (Lehmann's alternative): $F_{1} = p \cdot F_{2}^{q}+(1-p)\cdot F_{2}$ for some $0< p < 1$ and $q > 1$ with $c \in \mathbb{R}$, where $R_{i1}\sim F_{1}$ and $R_{i2}\sim F_{2}$. Lehmann's alternative is often used to model some uncommon but dramatic responses to treatment; see Rosenbaum (2007a) and Rosenbaum (2010b, Chapter 16) for some real data examples.




\begin{Proposition}\label{prop:example}
Assumptions 1-3 hold for Examples 1-3.
\end{Proposition}

\begin{Theorem}[Design sensitivity of the aberrant rank test]\label{dsformula}
Define $G(v)=\frac{1}{m}\sum_{j=1}^{m}\max \{F_{j}(v)-F_{j}(c), 0\}$ and $[m-1]=\{1,\dots,m-1\}$. Under Assumptions 1-3, 
\begin{equation}\label{eqn: design sensitivity formula}
		\mathbb{E} \Bigg\{ \max_{b\in [m-1]} \frac{\sum_{j=1}^{b}G(R_{i(j)})+\Gamma \sum_{j=b+1}^{m}G(R_{i(j)})}{b+\Gamma(m-b)} \Bigg \}=\mathbb{E}\{G(R_{i1})\}
\end{equation}
has a unique solution for $\Gamma \in (1,+\infty)$, call it $\widetilde{\Gamma}$. Then $\widetilde{\Gamma}$ is the design sensitivity of the aberrant rank test as in (\ref{def:abe}). That is, as $I \rightarrow \infty$, the power $\Psi_{\Gamma, I}$ of a one-sided $\alpha$-level sensitivity analysis satisfies $\Psi_{\Gamma, I}  \rightarrow 1$ if $\Gamma < \widetilde{\Gamma}$, and $\Psi_{\Gamma, I} \rightarrow 0$ if $\Gamma > \widetilde{\Gamma}$.
\end{Theorem}

Proofs of all propositions and theorems in this paper are provided in Appendix A in the online supplementary materials. Theorem~\ref{dsformula} confirms that the design sensitivity of the aberrant rank test depends only on the underlying data generating distribution $F$ and is independent of the level $\alpha$ and the sample size $I$. Setting $c=-\infty$ gives the design sensitivity formula of the Wilcoxon rank sum test.

\subsection{Asymptotic comparison via the design sensitivity}\label{sec:computeds}

 Theorem~\ref{dsformula} allows us to numerically calculate the design sensitivity of the aberrant rank test in each situation, and compare it with that of the Mantel-Haenszel test. Since the design sensitivity only depends on the data generating process and is independent of the level $\alpha$ and the sample size $I$, it gives us an intrinsic and elegant measurement of how robust a test is to hidden bias, and enables us to asymptotically compare two tests for observational studies. For the calculation of the design sensitivities, as in Theorem~\ref{dsformula}, we assume that for each $i$, without loss of generality, $j=1$ receives treatment and others receive control, and $(R_{i1}, \dots, R_{im})^{T}=(r_{Ti1}, r_{Ci2}, \dots, r_{Cim})^{T}$ is an i.i.d. realization from a multivariate continuous distribution. To make our calculation easier and clearer, we further assume that: First, $r_{Tij}=g(r_{Cij})$ for some deterministic function $g$. That is, given $g$, $r_{Tij}$ is only determined by $r_{Cij}$ and is independent of other individuals' outcomes given $r_{Cij}$; Second, each $r_{Cij}$ is realized from the same distribution $F$; Third, within the same stratum, $R_{i1},\dots, R_{im}$ are independent of each other. In Appendix B, we also examine the cases when $R_{i1}, \dots, R_{im}$ are correlated. Note that these three assumptions merely serve to simplify the simulations and are not necessary for Theorem~\ref{dsformula}. 
 
 We consider the following four models: (i) \textbf{Model 1} (additive treatment effects, normal distribution): $r_{Tij}=r_{Cij}+\beta$, $F$ is the standard normal distribution; (ii) \textbf{Model 2} (additive treatment effects, Laplace distribution): $r_{Tij}=r_{Cij}+\beta$, $F$ is the Laplace distribution with mean zero and variance one; (iii) \textbf{Model 3} (multiplicative treatment effects, normal distribution): $r_{Tij}=\delta\cdot  r_{Cij}$, $F$ is the standard normal distribution; (iv) \textbf{Model 4} (multiplicative treatment effects, Laplace distribution): $r_{Tij}=\delta\cdot r_{Cij}$, $F$ is the Laplace distribution with mean zero and variance one.
 For all four models, we set the aberrant response threshold to be $c=1$, that is, any response $R_{ij}>1$ is considered to be an aberrant response. Table~\ref{tab:designsensitivity} reports the design sensitivities of the Mantel-Haenszel test and the aberrant rank test under Models 1-4 with $m=4$ (i.e., matching with three controls) and various $\beta$ and $\delta$. Calculation is based on Monte-Carlo simulations. Specifically, under each data generating model, we can calculate the left-hand side (LHS) of (\ref{eqn: design sensitivity formula}) for each fixed $\Gamma$ and the right-hand side (RHS) of (\ref{eqn: design sensitivity formula}) using Monte-Carlo simulations. According to Lemma~\ref{lem:monotone} in Appendix A, the RHS of (\ref{eqn: design sensitivity formula}) is a strictly monotonically increasing function of $\Gamma$, therefore we can use the bisection method to find the solution of equation (\ref{eqn: design sensitivity formula}). According to Theorem~\ref{dsformula}, that solution is exactly the design sensitivity of the aberrant rank test given each data generating model.

Two clear patterns emerge in Table~\ref{tab:designsensitivity}. First, the choice of the test statistic has a huge influence on the design sensitivities. For example, under Model 3 with $\delta=2$, the design sensitivity of the aberrant rank test is nearly twice as big as that of the Mantel-Haenszel test. Second, whether or not the aberrant rank test outperforms the Mantel-Haenszel test depends upon the unknown data generating distribution of $\mathcal{F}$. As seen from Table~\ref{tab:designsensitivity}, under Models 1, 3 and 4, the aberrant rank test should be asymptotically less sensitive to unmeasured confounders with larger design sensitivities; instead under Model 2, the Mantel-Haenszel test should be more favorable in a sensitivity analysis with larger $\widetilde{\Gamma}$. These theoretical insights are validated in a simulation study in Section~\ref{sec:simulation}. 

\begin{table}[ht]
\centering \caption{Design sensitivities of the Mantel-Haenszel test and the aberrant rank test under Models 1-4 and matching with three controls with various parameters. The larger of the two design sensitivities of the two tests is in bold in each case.}
\label{tab:designsensitivity}
\small
\centering
\begin{tabular}{cccccccc} 
  \hline
  \multirow{2}{*}{} & \multicolumn{3}{c}{Model 1: additive, normal} & &\multicolumn{3}{c}{Model 2: additive, Laplace}\\
   \cmidrule(r){2-4} \cmidrule(r){6-8}
  Test statistic & $\beta=0.50$ & $\beta=0.75$ & $\beta=1.00$ & & $\beta=0.50$ & $\beta=0.75$ & $\beta=1.00$ \\
  \hline
   Mantel-Haenszel & 2.36 & 3.56  & 5.30 &  & \textbf{2.36} & \textbf{3.91}  & \textbf{7.21 }\\
   Aberrant rank & \textbf{2.63} & \textbf{4.20}  & \textbf{6.50} & & 2.28 & 3.59  & 5.93 \\
 \hline
  \multirow{2}{*}{} & \multicolumn{3}{c}{Model 3: multiplicative, normal} & & \multicolumn{3}{c}{Model 4: multiplicative, Laplace}\\
     \cmidrule(r){2-4} \cmidrule(r){6-8}
  Test statistic & $\delta=1.50$ & $\delta=1.75$ & $\delta=2.00$ & & $\delta=1.50$ & $\delta=1.75$ & $\delta=2.00$\\
  \hline 
   Mantel-Haenszel &  1.80 & 2.11 & 2.37 & & 1.75 & 2.07  & 2.37\\
   Aberrant rank &  \textbf{2.50} &  \textbf{3.28} & \textbf{4.07}& & \textbf{2.15} & \textbf{2.75}  & \textbf{3.36} \\
 \hline 
 \end{tabular}
\end{table}

We give some intuition as to why the aberrant rank test should sometimes be preferred over the Mantel-Haenszel test and other times the Mantel-Haenszel should be preferred. Suppose that $r_{Cij} \stackrel{iid}{\sim} f_{0}(x)$ and $r_{Tij} \stackrel{iid}{\sim} f_{1}(x)$ where $f_{0}$ and $f_{1}$ are two densities. Roughly speaking, the more $f_{1}(x)/f_{0}(x)$ departs from 1, the easier it is to distinguish the treated and control given the outcome value $x$. For Model 1 with $\beta>0$, $f_{1}(x)/f_{0}(x)=\exp(\beta x-\beta^{2}/2)$. For Model 2 with $\beta>0$ and $x > \beta$, $f_{1}(x)/f_{0}(x)= \exp(\sqrt{2} \beta)$. For Model 3 with $\delta>1$ and $x >0$, $f_{1}(x)/f_{0}(x)=\exp((1-1/\delta^{2})x^{2}/2)/\delta$. For Model 4 with $\delta>1$, $f_{1}(x)/f_{0}(x)=\exp(\sqrt{2}(1-1/\delta)x)/\delta $. Thus, for Models 1, 3 and 4 with $\beta>0$ and $\delta>1$, suppose that $c$ is large enough, specially $c\geq \max \{ \frac{\beta}{2}, \sqrt{\frac{2\log \delta}{1-1/\delta^{2}}}, \frac{\log \delta}{\sqrt{2}(1-1/\delta)}\}$, then $f_{1}(x)/f_{0}(x)\geq 1$ and $f_{1}(x)/f_{0}(x)$ is increasing for all $x\geq c$. That is, in these three models, it is easier to detect the true treatment effect at the tail (i.e. larger outcome value $x$) and the aberrant rank test should outperform the Mantel-Haenszel test by assigning larger weights to more aberrant responses (i.e., larger outcome values) via aberrant ranks. For Model 2 with $c\geq \beta$, $f_{1}(x)/f_{0}(x)$ is a constant for $x\geq c$. In this case, the Mantel-Haenszel test should be more powerful than the aberrant rank test since it does not distinguish different magnitudes of severity, while the aberrant rank test loses power by unnecessarily assigning unequal weights based upon the degree of aberration.

\section{A new, general adaptive approach to combine two test statistics in observational studies}

\subsection{Motivation and previous methods}
From the perspectives of design sensitivity and power of sensitivity analysis, neither the Mantel-Haenszel test nor the aberrant rank test uniformly dominates the other. Instead, which test is to be preferred depends upon the data generating process. Unfortunately we typically do not know which one is better for a given setting since we do not typically know the true data generating process. This type of problem is common in observational studies, where we typically have several available tests that we can use, but there is no single choice that can dominate all other choices in all possible situations (Rosenbaum, 2012). To overcome this type of problem in observational studies, various methods have been proposed. Among these, for example, Heller et al. (2009) and Zhang et al. (2011) used a sample splitting method in which a fraction of the data, the planning sample, is used to select a test and the remaining part of the data, the analysis sample, to carry out a test. The sample splitting method throws out the planning sample for carrying out the test which reduces power for small or moderate sample sizes.  

Rosenbaum (2012) proposed an adaptive approach to combine two tests in observational studies that is totally data-driven and does not require dropping samples for design, and can achieve the larger of the two design sensitivities of the component tests. This traditional adaptive approach works for combining different tests within a large class of test statistics for pair matched samples, including any test statistics of the form $T=\sum_{i=1}^{I}\mathbbm{1}(Y_{i}>0)\ h_{i}$, where $Y_{i}=(Z_{i1}-Z_{i2})(R_{i1}-R_{i2})$ is the treated-minus-control difference in response for matched pair $i$ and $h_{i}$ is a function of $|Y_{1}|, \dots, |Y_{I}|$, in which case we can find an uniform upper bound test statistic $\overline{T}_{\Gamma}$ under each sensitivity parameter $\Gamma$ such that $\mathbb{P}(T\geq t \mid \mathcal{F}, \mathcal{Z})\leq \mathbb{P}(\overline{T}_{\Gamma} \geq t \mid \mathcal{F}, \mathcal{Z})$ for any $t$. To combine two different test statistics, this traditional adaptive approach corrects for the correlation between the two test statistics by using the fact that the two upper bound statistics are asymptotically jointly normal under some regularity conditions; see Section 2 in Rosenbaum (2012). The cost for this correction is small compared with, for example, the Bonferroni adjustment since the two tests are typically highly correlated. However, this traditional adaptive approach can only be applied to test statistics that are uniformly bounded by a known distribution, which typically requires either the matching regime to be pair matching or the outcomes to be binary, neither of which would hold for many commonly used tests; see Section~\ref{sec: our contribution} for details. For example, the traditional approach cannot be used for the aberrant rank test, the Wilcoxon rank sum test, the Hodges-Lehmann aligned rank test or the Huber-Maritz m-tests (Gastwirth et al., 2000; Rosenbaum 2002b, 2007b).

\subsection{A new, general adaptive test via two-stage programming}\label{sec:adaptive}

Instead of focusing on matching with $m-1$ ($m\geq 2$) controls as we did in previous sections, in this section we consider a more general matching regime allowing matching with different number of controls across the strata. Suppose that there are $I$ matched strata with $n_{i}$ individuals in the $i$-th stratum, $N=\sum_{i=1}^{I}n_{i}$ individuals in total. $n_{i}=2$ with $Z_{i1}+Z_{i2}=1$ for all $i$ refers to pair matching. $n_{i}=m\geq 3$ with $\sum_{j=1}^{m}Z_{ij}=1$ for all $i$ refers to matching with multiple controls. In full matching, $n_{i}$ can take different values with different $i$, and $\sum_{j=1}^{n_{i}}Z_{ij}\in \{1,n_{i}-1\}$ for all $i$ (i.e., either one treated individual and one or more controls, or one control and one or more treated individuals, within each stratum). As in previous sections, we still let $\mathbf{Z}=(Z_{11}, \dots, Z_{In_{I}})^{T}$ be the binary vector of treatment assignments, and $\mathbf{Z}\in \mathcal{Z}$ if and only if $\sum_{j=1}^{n_{i}}Z_{ij}=1$ for each $i$. The constraint $\sum_{j=1}^{n_{i}}Z_{ij}=1$ for all $i$ is no more restrictive than assuming $\sum_{j=1}^{n_{i}}Z_{ij}\in \{1,n_{i}-1\}$ for all $i$ and is only imposed to make our derivations in this section clearer. See Appendix E for the detailed description of how the procedure derived in this section can be directly extended to allow for $\sum_{j=1}^{n_{i}}Z_{ij}\in \{1,n_{i}-1\}$ for all $i$. We still let $\mathcal{F}$ be the set of all fixed quantities of $r_{Tij}, r_{Cij}, \mathbf{x}_{ij}$ and $u_{ij}$.

Motivated by the demand of performing adaptive inference in much more general settings than the traditional adaptive approach, we develop here a new adaptive approach that can combine any two sum test statistics which refer to any test statistics with the form $T=\mathbf{Z}^{T}\mathbf{q}=\sum_{i=1}^{I}\sum_{j=1}^{n_{i}} Z_{ij}\ q_{ij}$ where each $q_{ij}$ is an arbitrary function of the response vector $\mathbf{R}=(R_{11},\dots, R_{In_{I}})^{T}$ that does not vary with $\mathbf{Z}\in \mathcal{Z}$ under the null hypothesis, and can work under various matching strategies due to the flexibility of the value of each $n_{i}$. We would like the power of the adaptive test to be asymptotically no less than the higher of the two powers of the component tests in sensitivity analysis. The idea is that when the sample size is large, to achieve the higher of the two powers of the component tests is almost equivalent to achieving the larger of the two design sensitivities of the component tests. Consider applying the Bonferroni adjustment to the component tests $T_{k}=\mathbf{Z}^{T}\mathbf{q}_{k}=\sum_{i=1}^{I}\sum_{j=1}^{n_{i}}Z_{ij}q_{ijk}$ with $\mathbf{q}_{k}=(q_{11k}, \dots, q_{In_{I}k})^{T}$ where each $q_{ijk}$ is a function of the response vector $\mathbf{R}$ and $k\in \{1,2\}$. Let $p_{ij}=\mathbb{P}(Z_{ij}=1\mid \mathcal{F}, \mathcal{Z})=\exp(\gamma u_{ij})/\sum_{j^{\prime}=1}^{n_{i}}\exp(\gamma u_{ij^{\prime}})$. For a one-sided testing procedure with level $\alpha$ and given $\Gamma$, we 
\begin{equation}\label{test:bonferroni}
\text{reject the null if}\	\max\limits_{k\in \{1,2\}}\ \min \limits_{\mathbf{u} \in \mathcal{U}}\ \frac{t_{k}-\mu_{k,\mathbf{u}}}{\sigma_{k, \mathbf{u}}} \geq \Phi^{-1}(1-\alpha/2), 
\end{equation}
where $t_{k}$ is the observed value of $T_{k}$, and $\mu_{k, \mathbf{u}}=\mathbb{E}_{\Gamma, \mathbf{u}}(\mathbf{Z}^{T}\mathbf{q}_{k}\mid \mathcal{F}, \mathcal{Z})=\sum_{i=1}^{I}\sum_{j=1}^{n_{i}}p_{ij}q_{ijk}$ and $\sigma_{k, \mathbf{u}}^{2}=Var_{\Gamma, \mathbf{u}}(\mathbf{Z}^{T}\mathbf{q}_{k}\mid \mathcal{F}, \mathcal{Z})=\sum_{i=1}^{I}\sum_{j=1}^{n_{i}}p_{ij}q_{ijk}^{2}-\sum_{i=1}^{I}(\sum_{j=1}^{n_{i}}p_{ij}q_{ijk})^{2}$ are the expectations and variances of $T_{k}$ under the null hypothesis, with a specified $\Gamma$ and given all unobserved covariates $\mathbf{u}=(u_{11}, \dots, u_{In_{I}})^{T}\in [0,1]^{N}=:\mathcal{U}$. The term $\alpha/2$ in the RHS of (\ref{test:bonferroni}) comes from the Bonferroni adjustment with two component tests. Under a normal approximation, the standard deviate of $t_{k}$ follows a standard normal distribution, thus (\ref{test:bonferroni}) is a valid testing procedure with level $\alpha$ and given $\Gamma$ in a sensitivity analysis. Note that the design sensitivity of a test only depends on the data generating distribution and is independent of level $\alpha$. Using an argument parallel to the proof of Proposition 2 in Rosenbaum (2012), it is straightforward to show that applying (\ref{test:bonferroni}) with the two component tests can achieve the larger of the two design sensitivities, where we reject the null as long as one of the two tests rejects the null with significant level $\alpha/2$. However, simply applying (\ref{test:bonferroni}) may lose power due to two significant deficiencies. First, it does not use the fact that the confounder has to impact the treatment assignment in the same way between the two component tests on the same outcome variable. Second, it does not incorporate the information of the correlation between the two component tests. We implement a two-stage programming procedure to overcome these two deficiencies.

In the first stage, we utilize bounds on the correlation between $T_{1}$ and $T_{2}$ to replace $\Phi^{-1}(1-\alpha/2)$ with a smaller rejection threshold under the given $\Gamma$ and level $\alpha<1/2$. Under some mild regularity conditions (see Appendix C for details), $(T_{1}, T_{2})$ is asymptotically bivariate normal in the sense that for large $I$, the distribution function of $\Big(\frac{T_{1}-\mu_{1,\mathbf{u}}}{\sigma_{1, \mathbf{u}}},\ \frac{T_{2}-\mu_{2,\mathbf{u}}}{\sigma_{2, \mathbf{u}}} \Big)$ can be approximated by that of $(X_{1}, X_{2})\sim \mathcal{N}\left(\left(\begin{array}{c}
0\\
0
\end{array}\right),\left(\begin{array}{cc}
1 & \mathbf{\rho}_{\mathbf{u}} \\
\mathbf{\rho}_{\mathbf{u}} & 1 
\end{array}\right)\right)$,
where $\mathbf{\rho}_{\mathbf{u}}=\mathbb{E}\Big(\frac{ T_{1}-\mu_{1, \mathbf{u}}}{\sigma_{1, \mathbf{u}}}\cdot \frac{ T_{2}-\mu_{2, \mathbf{u}}}{\sigma_{2, \mathbf{u}}}\Big \vert \mathcal{F},\mathcal{Z} \Big)$ can be expressed as
\begin{equation}\label{correlation-with_u}
\footnotesize{\mathbf{\rho}_{\mathbf{u}}= \frac{\sum_{i=1}^{I}\sum_{j=1}^{n_{i}} p_{ij}q_{ij1}q_{ij2}-\sum_{i=1}^{I} (\sum_{j=1}^{n_{i}} p_{ij}q_{ij1}) (\sum_{j=1}^{n_{i}} p_{ij}q_{ij2})}{\sqrt{\sum_{i=1}^{I}\sum_{j=1}^{n_{i}}p_{ij}q_{ij1}^{2}-\sum_{i=1}^{I}(\sum_{j=1}^{n_{i}}p_{ij}q_{ij1})^{2}}\sqrt{\sum_{i=1}^{I}\sum_{j=1}^{n_{i}}p_{ij}q_{ij2}^{2}-\sum_{i=1}^{I}(\sum_{j=1}^{n_{i}}p_{ij}q_{ij2})^{2}}}.}
\end{equation}
Let $Q_{\mathbf{\rho}_{\mathbf{u}},\alpha}$ be the quantile such that $\mathbb{P}(X_{1}\leq Q_{\mathbf{\rho}_{\mathbf{u}},\alpha}, X_{2}\leq Q_{\mathbf{\rho}_{\mathbf{u}},\alpha})=1-\alpha$. Note that we would like to derive a valid testing procedure given any $\mathbf{u}$ with the given $\Gamma$ and $\alpha$, we should look at the worst-case rejection threshold $\max_{\mathbf{u}\in \mathcal{U}}Q_{\rho_{\mathbf{u}}, \alpha}$. Invoking Slepian's lemma (Slepian, 1962), to find $\max_{\mathbf{u}\in \mathcal{U}}Q_{\rho_{\mathbf{u}}, \alpha}$, it suffices to find $\min_{\mathbf{u}\in \mathcal{U}} \rho_{\mathbf{u}}$. Through setting $w_{ij}=\exp(\gamma u_{ij})$, we further transform solving $\min_{\mathbf{u}\in \mathcal{U}} \rho_{\mathbf{u}}$ into solving 
\begin{align*}
    &\underset{w_{ij}}{\text{minimize}}\quad \mathbf{\rho}_{\mathbf{u}} \quad \quad  (*) \\
    &\text{subject to}\ \ 1 \leq w_{ij} \leq \Gamma, \quad \forall  i,j
\end{align*}
where $\mathbf{\rho}_{\mathbf{u}}$ is as in (\ref{correlation-with_u}) with $p_{ij}=w_{ij}/\sum_{j^{\prime}=1}^{n_{i}}w_{ij^{\prime}}$. $(*)$ is a large-scale nonlinear optimization problem with box constraints which can be solved approximately in a reasonable amount of time by the well-known L-BFGS-B algorithm, which is a limited-memory Broyden-Fletcher-Goldfarb-Shanno (L-BFGS) algorithm allowing box constraints (Byrd et al., 1995; Zhu et al., 1997). Denote the optimal value of $(*)$ with sensitivity parameter $\Gamma$ as $\mathbf{\rho}^{*}_{\Gamma}$. Then the corresponding worst-case quantile $\max_{\mathbf{u}\in \mathcal{U}}Q_{\rho_{\mathbf{u}}, \alpha}$ equals $Q_{\mathbf{\rho}^{*}_{\Gamma},\alpha}$ by Slepian's lemma. It is well known that $Q_{\mathbf{\rho}^{*}_{\Gamma},\alpha}<\Phi^{-1}(1-\alpha/2)$ as long as $\mathbf{\rho}^{*}_{\Gamma}>-1$. Thus, for two positively correlated test statistics $T_{1}$ and $T_{2}$, especially when the correlation is much greater than zero (which is the case when combining the Mantel-Haenszel test and the aberrant rank test), $Q_{\mathbf{\rho}^{*}_{\Gamma}, \alpha}$ is a much less conservative rejection threshold than $\Phi^{-1}(1-\alpha/2)$. 

In the second stage, we apply the minimax procedure developed in Fogarty and Small (2016) to replace the test statistic $\max\limits_{k\in \{1,2\}}\ \min \limits_{\mathbf{u} \in \mathcal{U}}\ (t_{k}-\mu_{k,\mathbf{u}})/\sigma_{k, \mathbf{u}}$ in (\ref{test:bonferroni}) with a larger one. Note that the following max-min inequality always holds
\begin{align}\label{inequality:minimax}
  \min \limits_{\mathbf{u} \in \mathcal{U}} \  \max\limits_{k\in \{1,2\}} \frac{t_{k}-\mu_{k,\mathbf{u}}}{\sigma_{k, \mathbf{u}}}\geq  \max\limits_{k\in \{1,2\}}\ \min \limits_{\mathbf{u} \in \mathcal{U} }\ \frac{t_{k}-\mu_{k,\mathbf{u}}}{\sigma_{k, \mathbf{u}}},
\end{align}
and strict inequality is possible. (\ref{inequality:minimax}) implies that instead of performing the two sensitivity analyses to solve $\min_{\mathbf{u}\in \mathcal{U} }(t_{k}-\mu_{k,\mathbf{u}})/\sigma_{k, \mathbf{u}}$ for $k\in \{1,2\}$ separately, we should conduct a simultaneous sensitivity analysis to directly check if $\min_{\mathbf{u} \in \mathcal{U}}  \max_{k\in \{1,2\}} (t_{k}-\mu_{k,\mathbf{u}})/\sigma_{k, \mathbf{u}}\geq Q_{\mathbf{\rho}^{*}_{\Gamma},\alpha}$ - if the inequality holds, we reject the null; otherwise, we fail to reject. Adapting the one-sided minimax procedure described in Part B of the Appendices of Fogarty and Small (2016) with our new rejection threshold $Q_{\mathbf{\rho}^{*}_{\Gamma}, \alpha}$, this procedure can be implemented through setting $s_{i}=1/\sum_{j^{\prime}=1}^{n_{i}}\exp(\gamma u_{ij^{\prime}})$ and solving the following quadratically constrained linear program with $M$ being a sufficiently large constant (see Appendix D for the detailed derivation):
\begin{equation*}
     \begin{split}
        \underset{y, p_{ij}, s_i, b_k}{\text{minimize}} \quad &y \quad \quad (**)\\
         \text{subject to}\quad & y\geq (t_k-\mu_{k,\mathbf{u}})^2- Q^{2}_{\mathbf{\rho}^{*}_{\Gamma},\alpha} \sigma^2_{k,\mathbf{u}}-Mb_k \quad \forall k \in \{0,1\}\\
        &\sum_{j=1}^{n_{i}} p_{ij}=1\quad\forall i\\
        &s_i\leq p_{ij}\leq \Gamma s_i\quad\forall i,j\\
        &p_{ij}\geq0 \quad\forall i,j\\
        &b_k\in \{0,1\} \quad \forall k \in \{0,1\}\\
        &-Mb_k\leq t_k-\mu_{k,\mathbf{u}}\leq M(1-b_k),\quad\forall k \in \{0,1\}
     \end{split}
 \end{equation*}
and checking whether the optimal value $y^{*}_{\Gamma}\geq 0$. If it is, we reject the null; otherwise, we fail to reject. The `$M$' constraint here precludes a directional error, as without it one might reject the null if evidence pointed in the opposite direction of the alternative. A quadratically constrained linear program can be efficiently solved with many available solvers. Contrary to implementing $(*)$, from which the gains in power is relatively large when the correlation between $T_{1}$ and $T_{2}$ is strong, implementing $(**)$ (the minimax procedure) typically can have marked improvement of power when the correlation between $T_{1}$ and $T_{2}$ is weak; see Section 8 in Fogarty and Small (2016). That is, by implementing our two-step programming $(*)$ and $(**)$, we can always expect gains in power no matter the correlation between the two component tests are strong or weak.

To conclude, in our new adaptive testing procedure, for a one-sided test with level $\alpha$ and sensitivity parameter $\Gamma$, given the observed value $t_{k}$ of $T_{k}$ $(k=1,2)$ , we
\begin{equation}\label{test:adaptive}
\text{reject the null if}\ \min \limits_{\mathbf{u} \in \mathcal{U}}\ \max\limits_{k\in \{1,2\}} \frac{t_{k}-\mu_{k,\mathbf{u}}}{\sigma_{k, \mathbf{u}}} \geq Q_{\mathbf{\rho}^{*}_{\Gamma},\alpha}.
\end{equation}
When $\Gamma=1$, (\ref{test:adaptive}) reduces to the usual testing procedure with the maximum statistic $\max\{T_{1}, T_{2}\}$ with correcting for $\text{Cor}(T_{1}, T_{2})$. We have shown that (\ref{test:adaptive}) can be implemented through the following two-stage programming method:

\begin{algorithm}[H]\label{algo:adaptive}
\textbf{Input:} Sensitivity parameter $\Gamma$; level $\alpha$ of the sensitivity analysis; treatment assignment indicator vector $\mathbf{Z}=(Z_{11},\dots, Z_{In_{I}})^{T}$; the score vector $\mathbf{q}_{1}=(q_{111},\dots, q_{In_{I}1})^{T}$ associated with $T_{1}=\sum_{i=1}^{I}\sum_{j=1}^{n_{i}}Z_{ij}q_{ij1}$; the score vector $\mathbf{q}_{2}=(q_{112},\dots, q_{In_{I}2})^{T}$ associated with $T_{2}=\sum_{i=1}^{I}\sum_{j=1}^{n_{i}}Z_{ij}q_{ij2}$\;
\textbf{Step 1:} Solve $(*)$ to get the worst-case correlation $\mathbf{\rho}^{*}_{\Gamma}$ along with the corresponding worst-case quantile $Q_{\mathbf{\rho}^{*}_{\Gamma},\alpha}$ \;
\textbf{Step 2:} Solve $(**)$ with $Q_{\mathbf{\rho}^{*}_{\Gamma},\alpha}$ obtained from Step 1, and get the corresponding optimal value $y^{*}_{\Gamma}$ \;
\textbf{Output:} If $y^{*}_{\Gamma}\geq 0$, we reject the null; otherwise, we fail to reject.
\caption{Two-stage programming as the new, general adaptive test}
\end{algorithm}
We provide an implementation of Algorithm~\ref{algo:adaptive} using the R interface to \textsf{Gurobi}, which is a commercial solver but is freely available for academic use. Proposition~\ref{proposition:level} says that the sensitivity analysis with the adaptive testing procedure described in Algorithm~\ref{algo:adaptive} has the correct level $\alpha$ asymptotically.
\begin{Proposition}\label{proposition:level}
For any unknown true $\mathbf{u}_{0}\in \mathcal{U}$ and true $\Gamma_{0}\leq \Gamma$, we have
\begin{equation*}
  \lim_{I\rightarrow \infty}  \mathbb{P}_{\Gamma_{0}, \mathbf{u}_{0}}\Big(\min \limits_{\mathbf{u} \in \mathcal{U} }\ \max\limits_{k\in \{1,2\}} \frac{t_{k}-\mu_{k,\mathbf{u}}}{\sigma_{k, \mathbf{u}}} \geq Q_{\mathbf{\rho}^{*}_{\Gamma}, \alpha}\Big|\mathcal{F}, \mathcal{Z}\Big)\leq \alpha.
\end{equation*}
\end{Proposition}

A nice feature of the traditional adaptive test (Rosenbaum, 2012) is that its design sensitivity is the larger of the two component tests. The Bonferroni adjustment and sample splitting method also have this design sensitivity but sometimes lose power in finite samples to the adaptive test. In Theorem~\ref{th:gainsinds}, we prove that the design sensitivity of our new adaptive approach is always greater than or equal to both two design sensitivities of the component tests, and surprisingly, \textit{strict inequality is possible}. We refer to this new phenomenon as ``super-adaptivity".

\begin{Theorem}[Super-adaptivity]\label{th:gainsinds}
Let $\widetilde{\Gamma}_{1}$ and $\widetilde{\Gamma}_{2}$ be the two design sensitivities of the two tests $T_{1}$ and $T_{2}$, and let $\widetilde{\Gamma}_{1: 2}$ be the design sensitivity of the adaptive testing procedure (\ref{test:adaptive}) implemented by Algorithm~\ref{algo:adaptive} with $T_{1}$ and $T_{2}$ as the two component tests. We have $\widetilde{\Gamma}_{1: 2}\geq \max\{ \widetilde{\Gamma}_{1}, \widetilde{\Gamma}_{2}\}$, and strict inequality is possible. 
\end{Theorem}
Theorem~\ref{th:gainsinds} shows that in terms of the design sensitivity, our new adaptive test dominates all the existing methods, including the traditional adaptive test, the Bonferroni adjustment and sample splitting. Recall that the design sensitivity is a threshold of the consistency of a test in a sensitivity analysis with respect to sensitivity parameter $\Gamma$. When $\widetilde{\Gamma}_{1: 2}=\max\{ \widetilde{\Gamma}_{1}, \widetilde{\Gamma}_{2}\}$, roughly speaking, the new adaptive test is consistent as long as one of the two component tests was consistent, which can also be obtained by the traditional adaptive approach. When $\widetilde{\Gamma}_{1: 2}>\max\{ \widetilde{\Gamma}_{1}, \widetilde{\Gamma}_{2}\}$, the new adaptive test can still be consistent even if neither of the two component tests was consistent, which cannot be achieved from using the traditional adaptive approach. 

Typically, substantial gains in design sensitivity (i.e., gaps between $\widetilde{\Gamma}_{1:2}$ and $\max \{ \widetilde{\Gamma}_{1}, \widetilde{\Gamma}_{2} \}$) resulting from Algorithm~\ref{algo:adaptive} are more likely to be observed with two negatively correlated, independent or weakly positively correlated component statistics than with two strongly positively correlated component statistics (see Table~\ref{tab:checkvariouscor} in Appendix A). When combining two statistics $T_{1}$ and $T_{2}$ on one response vector $\mathbf{R}$ in an adaptive test, we often expect $T_{1}$ and $T_{2}$ to be highly positively correlated, in which case gains in design sensitivity may be hard to see without large samples. But Theorem~\ref{th:gainsinds} is still worth highlighting since it is the first time that an adaptive test can result in a design sensitivity strictly larger than $\max \{ \widetilde{\Gamma}_{1}, \widetilde{\Gamma}_{2}\}$, and it can inspire further studies on designing new adaptive tests with larger design sensitivities.

Note that the design sensitivity only measures limiting insensitivity to hidden bias. In terms of the finite sample power, we need to pay the price for correcting for the two component tests in the adaptive test. That is, Theorem~\ref{th:gainsinds} does not imply that the power of the adaptive test in a sensitivity analysis is always greater than or equal to the maximal power of the two component tests for any sample size. Instead, Theorem~\ref{th:gainsinds} implies that as long as the sample size is sufficiently large, applying Algorithm~\ref{algo:adaptive} to perform an adaptive inference is as good or better than knowing which of the two component tests should be better and using only that test, and is typically much better than incorrectly choosing the worse one among the two component tests, regardless of what the unknown data generating process is and what the two component tests are. Having theoretically derived this favorable asymptotic property of the adaptive test in Theorem~\ref{th:gainsinds}, we turn to examining its performance with realistic sample sizes via simulations in Section~\ref{sec:simulation} and Tables~\ref{tab:checkgainsinds} and \ref{tab:checkvariouscor} in Appendix A.

\section{Simulation studies}\label{sec:simulation}

We examine the finite sample power of sensitivity analyses to check the validity of the theoretical intuitions gained from calculating design sensitivities and compare the performances of (i) the Mantel-Haenszel test, (ii) the aberrant rank test, and (iii) our new adaptive test applying Algorithm~\ref{algo:adaptive} with the Mantel-Haenszel test and the aberrant rank test as components. That is, we use simulations to estimate the probability that the worst-case p-value given by a test statistic in a sensitivity analysis with sensitivity parameter $\Gamma$ will be less than $\alpha=0.05$ under the favorable situation when there is an actual treatment effect and no hidden bias. Table~\ref{tab:powerwithadaptive} summarizes the simulated power of the three tests under Models 1 - 4 discussed in Section~\ref{sec:computeds}, where we match with three controls and number of matched strata $I=100$ or $I=1000$. In Table~\ref{tab:powerwithadaptive}, we set $\beta=1$ for Models 1 and 2 and set $\delta=2$ for Models 3 and 4. For reference, we also give the design sensitivity of each test statistic in the first row of each block. We summarize the simulated size of the above three tests in Table~\ref{tab:sizewithadaptive} in Appendix F.

\begin{table}[H]
\centering \caption{Simulated power of the Mantel-Haenszel test, the aberrant rank test and the adaptive test. We set $\alpha=0.05$, $c=1$ and $m=4$. We set $\beta=1$ for Models 1 and 2 and $\delta=2$ for Models 3 and 4. Each number is based on 2,000 replications. The largest of the three simulated powers in each case is in bold.}
\label{tab:powerwithadaptive}
\scriptsize
\begin{tabular}{ccccccc}
\hline
\multirow{2}{*}{\textbf{Model 1}} & \multicolumn{3}{c}{$I=100$ Matched Strata} & \multicolumn{3}{c}{$I=1000$ Matched Strata}\\
& Mantel-Haenszel & Aberrant rank & Adaptive test & Mantel-Haenszel & Aberrant rank & Adaptive test \\
\hline
$\widetilde{\Gamma}$ & 5.30 & 6.50 &  $\geq 6.50$ & 5.30 & 6.50 &  $\geq 6.50$ \\
\hline
$\Gamma=3.0$ & 0.71 & \textbf{0.87} & 0.83 & 1.00 & \textbf{1.00} & 1.00  \\
$\Gamma=3.5$ & 0.46 & \textbf{0.70} & 0.63 & 1.00 & \textbf{1.00} &  1.00  \\
$\Gamma=4.0$ & 0.28 & \textbf{0.52} & 0.43 & 0.96 & \textbf{1.00} &  1.00  \\
$\Gamma=4.5$ & 0.14 & \textbf{0.32} & 0.25 & 0.61 & \textbf{0.99} & 0.99   \\
$\Gamma=5.0$ & 0.08 & \textbf{0.21} & 0.14 & 0.17 & \textbf{0.89} & 0.82 \\
$\Gamma=5.5$ & 0.04 & \textbf{0.12} & 0.09 & 0.02 & \textbf{0.57} &  0.47  \\
$\Gamma=6.0$ & 0.01 & \textbf{0.06} & 0.04 & 0.00 & \textbf{0.20} &  0.14  \\
\hline
\multirow{2}{*}{\textbf{Model 2}} & \multicolumn{3}{c}{$I=100$ Matched Strata} & \multicolumn{3}{c}{$I=1000$ Matched Strata}\\
& Mantel-Haenszel & Aberrant rank & Adaptive test & Mantel-Haenszel & Aberrant rank & Adaptive test \\
\hline
$\widetilde{\Gamma}$ & 7.21 & 5.93 & $\geq 7.21$ & 7.21 & 5.93 & $\geq 7.21$\\
\hline
$\Gamma=3.0$ & \textbf{0.95} & 0.78  & 0.92  & \textbf{1.00} & 1.00 & 1.00 \\
$\Gamma=3.5$ & \textbf{0.84} & 0.58 & 0.77 & \textbf{1.00} & 1.00 &  1.00 \\
$\Gamma=4.0$ & \textbf{0.70} & 0.38 & 0.60  & \textbf{1.00} & 1.00 & 1.00 \\
$\Gamma=4.5$ & \textbf{0.51} & 0.22  & 0.44 & \textbf{1.00} & 0.91  & 1.00 \\
$\Gamma=5.0$ & \textbf{0.37} & 0.13 & 0.26 & \textbf{1.00} & 0.58 &  0.99 \\
$\Gamma=5.5$ & \textbf{0.22} & 0.07 & 0.16 &\textbf{0.93} & 0.20 & 0.88   \\
$\Gamma=6.0$ & \textbf{0.15} & 0.04 & 0.10 & \textbf{0.64} & 0.03 & 0.58    \\
\hline
\multirow{2}{*}{\textbf{Model 3}} & \multicolumn{3}{c}{$I=100$ Matched Strata} & \multicolumn{3}{c}{$I=1000$ Matched Strata}\\
& Mantel-Haenszel & Aberrant rank & Adaptive test & Mantel-Haenszel & Aberrant rank & Adaptive test \\
\hline
$\widetilde{\Gamma}$ & 2.37 & 4.07 & $\geq 4.07$ & 2.37 & 4.07 & $\geq 4.07$ \\
\hline
$\Gamma=1.0$ & 0.94 & \textbf{1.00} & 0.99 & 1.00 & \textbf{1.00} & 1.00  \\
$\Gamma=1.5$ & 0.52 & \textbf{0.94} & 0.93 & 1.00 & \textbf{1.00}  &  1.00 \\
$\Gamma=2.0$ & 0.15 & \textbf{0.74} & 0.71 & 0.63 & \textbf{1.00} & 1.00   \\
$\Gamma=2.5$ & 0.04 & \textbf{0.47}  & 0.39 & 0.01 & \textbf{1.00}  & 1.00 \\
$\Gamma=3.0$ & 0.01 & \textbf{0.24} & 0.17 & 0.00 & \textbf{0.94} &  0.89 \\
\hline
\multirow{2}{*}{\textbf{Model 4}} & \multicolumn{3}{c}{$I=100$ Matched Strata} & \multicolumn{3}{c}{$I=1000$ Matched Strata}\\
& Mantel-Haenszel & Aberrant rank & Adaptive test & Mantel-Haenszel & Aberrant rank & Adaptive test \\
\hline
$\widetilde{\Gamma}$  & 2.37 & 3.36 & $\geq 3.36$  & 2.37 & 3.36 & $\geq 3.36$\\
\hline
$\Gamma=1.0$ & 0.89 & \textbf{0.97} & 0.95 & 1.00 & \textbf{1.00} & 1.00  \\
$\Gamma=1.5$ & 0.46 & \textbf{0.78} & 0.72 & 1.00 & \textbf{1.00} & 1.00\\
$\Gamma=2.0$ & 0.14 & \textbf{0.47} & 0.39  & 0.55 & \textbf{1.00} &  1.00  \\
$\Gamma=2.5$ & 0.03 & \textbf{0.22}  & 0.17 & 0.02 & \textbf{0.88}  & 0.83 \\
$\Gamma=3.0$ & 0.01 & \textbf{0.08} & 0.05 & 0.00 & \textbf{0.29} &  0.24 \\
\hline
\end{tabular}
\end{table}

In Table~\ref{tab:powerwithadaptive}, in general, the power increases as the number of matched strata $I$ increases, and the power decreases as the bias magnitude $\Gamma$ increases, which agrees with empirical knowledge. The simulated power also verifies the validity of our design sensitivity formula. That is, as $I \rightarrow \infty$, the power of the test in a sensitivity analysis goes to 1 for $\Gamma<\widetilde{\Gamma}$, and the power goes to 0 for $\Gamma>\widetilde{\Gamma}$. For example, see the row $\Gamma=5.5$ for Model 1 in Table~\ref{tab:powerwithadaptive}, as $I$ increases from 100 to 1000, the power of the aberrant rank test with $\widetilde{\Gamma}=6.5>5.5$ is closer to 1, but the power of the Mantel-Haenszel test with $\widetilde{\Gamma}=5.3<5.5$ is closer to 0. From Table~\ref{tab:powerwithadaptive}, we can also observe that in Models 1, 3 and 4, the aberrant rank test is more powerful than the Mantel-Haenszel test; instead, in Model 2 the Mantel-Haenszel test has higher power than the aberrant rank test, and the gap between the two powers of these two tests could be extremely large, especially with large sample size and sensitivity parameter $\Gamma$ considerably greater than 1. For example, see Models 1 and 2 with $\Gamma=5.5$ and $I=1000$, and Models 3 and 4 with $\Gamma=2.5$ and $I=1000$. This confirms the two key insights obtained from calculation of design sensitivities: power of a sensitivity analysis can differ a lot with different choices between the two tests and the optimal choice between the two tests could be different under different data generating processes.

We now examine the asymptotic property of the adaptive test. Let $\widetilde{\Gamma}_{1}$, and $\widetilde{\Gamma}_{2}$ denote the design sensitivities of the Mantel-Haenszel test and the aberrant rank test respectively. As long as the given $\Gamma < \max\{ \widetilde{\Gamma}_{1}, \widetilde{\Gamma}_{2}\}$, the power of the adaptive test in a sensitivity analysis goes to 1 as sample size $I \rightarrow \infty$, even if one of the powers of the two component tests goes to zero if $\min\{ \widetilde{\Gamma}_{1}, \widetilde{\Gamma}_{2}\}<  \Gamma < \max\{ \widetilde{\Gamma}_{1}, \widetilde{\Gamma}_{2}\}$. For example, see the rows $\Gamma=6.0$ of Models 1 and 2 and the rows $\Gamma=3.0$ of Models 3 and 4 in Table~\ref{tab:powerwithadaptive}.

We then examine the finite sample performance of the adaptive test. As discussed in Section~\ref{sec:adaptive}, although the adaptive test uniformly dominates the two component tests in terms of the design sensitivity, in practice we need to pay the price for correcting for the two component tests and the finite sample power of the adaptive test may sit in between those of the two component tests, which is the case in Table~\ref{tab:powerwithadaptive}. If this is the case, the simulation results in Table~\ref{tab:powerwithadaptive} confirm that the price paid for correcting for the two component tests is very much worth it in the sense that if the power of the Mantel-Haenszel test and that of the aberrant rank test differ a lot, then the power of the adaptive test is typically much closer to the higher one of the powers of the two component tests than to the lower one in each case. This favorable finite sample property of the adaptive test holds both for relatively large sample sizes (e.g., see the cases $\Gamma=5, I=1000$ in Models 1 and 2) and relatively small sample sizes (e.g., see the cases $\Gamma=1.5$, $I=100$ in Models 3 and 4). To conclude, the new adaptive test is like a high quality insurance policy: we will lose a little money (the low cost of the insurance) if we bought one but an accident never occurs (i.e., if we were lucky enough to always choose the better one among the two component tests), but we will lose much more if an accident indeed occurs (i.e., if we unfortunately choose the worse one among the two component tests) but we never bought one.



\section{Adaptive inference of the effect of mother's age on child stunted growth}\label{sec:real data}

For the study of the effect of mother's age on child stunting discussed in Section~\ref{example}, we summarize the worst-case p-values of a sensitivity analysis reported by three different test statistics: the Mantel-Haenszel test, the aberrant rank test and the adaptive test applying Algorithm~\ref{algo:adaptive} putting together these two tests with various sensitivity parameters $\Gamma$ ranging from $1.00$ to $1.45$; see Appendix G for more details. From Table~\ref{tab:p-values}, we find that the Mantel-Haenszel test fails to detect a possible treatment effect (i.e., worst-case p-value $>0.05$) with sensitivity parameter $\Gamma=1.17$ under level 0.05. However, the aberrant rank test can detect a possible treatment effect (i.e., worst-case p-value $<0.05$) up to a much larger sensitivity parameter $\Gamma=1.43$. Thus, we can see that when studying causal determinants of aberrant response, the aberrant rank test might be preferred to the Mantel-Haenszel test since it might be less sensitive. However, we did not know this in advance of looking at the data, and choosing the test that is less sensitive on the data will inflate Type I errors in a sensitivity analysis. To use the data in choosing the best test while controlling the Type I error rate, we apply the adaptive approach developed in Section~\ref{sec:adaptive} to combine the aberrant rank test with the Mantel-Haenszel test to guarantee a powerful test in sensitivity analyses. From Table~\ref{tab:p-values}, we can find that if we combine these two tests with the new adaptive approach, we can successfully detect the possible actual treatment effect with $\Gamma=1.36$, which is close to the results obtained by using the more favorable one between the two component tests - the aberrant rank test, and substantially better than the least favorable of the two tests. Therefore, both the aberrant rank test and the adaptive test enable us to detect a significant treatment effect even with a nontrivial magnitude of hidden bias. Meanwhile, for this particular data set, the Mantel-Haenszel test would possibly give an exaggerated report of sensitivity to bias. This agrees with all our theoretical insights and simulations results. 

\begin{table}[H]
  \centering
  \caption{One-sided worst-case p-values under various $\Gamma$. The p-values $\approx 0.05$ are in bold. We also report the approximate sensitivity value of each test with level 0.05.}
  \label{tab:p-values}
  \scriptsize
  \centering
\begin{tabular}{cccc}
\hline\multicolumn{4}{c}{One-sided worst-case p-values under various $\Gamma$}\\
& Mantel-Haenszel & Aberrant rank & Adaptive test \\
\hline
$\Gamma=1.00$ & 0.010 & 0.001 & 0.001  \\
$\Gamma=1.05$ & 0.017 & 0.001 & 0.003  \\
$\Gamma=1.10$ & 0.028 & 0.003 & 0.005   \\
$\Gamma=1.15$ & 0.043 & 0.005  & 0.008  \\
$\Gamma=1.17$ & \textbf{0.051} & 0.006 & 0.010 \\
$\Gamma=1.20$ & 0.064 & 0.008 & 0.014  \\
$\Gamma=1.25$ & 0.089 &  0.013 & 0.022  \\
$\Gamma=1.30$ & 0.121 & 0.020 & 0.032  \\
$\Gamma=1.35$ & 0.157 & 0.029 & 0.047  \\
$\Gamma=1.36$ & 0.165 & 0.031 & \textbf{0.050} \\
$\Gamma=1.40$ & 0.198 & 0.040 & 0.064  \\
$\Gamma=1.43$ & 0.225 & \textbf{0.049} & 0.077  \\
$\Gamma=1.45$ & 0.244 & 0.055 & 0.086 \\
\hline
Sensitivity value & 1.17 & 1.43 & 1.36 \\
\hline
\end{tabular}
\end{table}

\section{Discussion}

We have developed an adaptive aberrant rank approach to conducting inference about the effect of a treatment on aberrant (bad) outcomes from matched observational studies when there is an established cutoff for what constitutes an aberrant outcome but more aberrant outcomes are worse than less aberrant ones. We have shown that our new approach asymptotically dominates the traditional approach (performing the Mantel-Haenszel test with the dichotomous outcome indicating aberration/non-aberration) and performs well in simulation studies. To establish the new approach, we have developed an empirical process approach to studying design sensitivity and developed a general adaptive testing procedure. These developments can be applied to other types of general matched observational studies beyond the aberrant outcome setting we have studied. 

There are limitations to this work. For example, we have not discussed how to enable adjustment for measured variables that were not used for matching. This side information, along with the matched observed covariates, can potentially be used to perform covariance adjustment in randomization inference (Rosenbaum, 2002a). However, existing model-based covariance adjustment approaches, e.g., covariance adjustment with robust linear regression considered in Rosenbaum (2002a), may not be directly applicable in our setting since the aberrant rank considers some truncated outcome (zero if not aberrant and multi-valued if aberrant). It might be fruitful for future research to explore how to incorporate both matched and unmatched measured variables to perform some covariance adjustment to further develop the aberrant rank approach.

\begin{center}
{\large\bf Acknowledgement}
\end{center}

We would like to thank Peter Cohen for his kind support for implementing Algorithm~\ref{algo:adaptive} in this article, Karun Adusumilli, Bhaswar Bhattacharya, and Paul Rosenbaum for helpful discussions, and the participants in the causal inference reading group for helpful comments.

\section*{References}%
%



\setlength{\hangindent}{12pt}
\noindent
Bloss, E., Wainaina, F. and Bailey, R. C. (2004). Prevalence and predictors of underweight, stunting, and wasting among children aged 5 and under in western Kenya. \textit{Journal of Tropical Pediatrics.} \textbf{50(5)} 260-270.

\setlength{\hangindent}{12pt}
\noindent
Brown, K. H., Black, R. E. and Becker, S. (1982). Seasonal changes in nutritional status and the prevalence of malnutrition in a longitudinal study of young children in rural Bangladesh. \textit{Am J Clin Nutr.} \textbf{36(2)} 303-13.

\setlength{\hangindent}{12pt}
\noindent
Byrd, R. H., Lu, P., Nocedal, J. and Zhu, C. (1995). A limited memory algorithm for bound constrained optimization. \textit{SIAM Journal on Scientific Computing.} \textbf{16(5)} 1190-1208.


\setlength{\hangindent}{12pt}
\noindent
Cox, D. R. (2018). \textit{Analysis of Binary Data.} Routledge.


\setlength{\hangindent}{12pt}
\noindent
Darteh, E. K. M., Acquah, E. and Kumi-Kyereme, A. (2014). Correlates of stunting among children in Ghana. \textit{BMC Public Health.} \textbf{14(1)} 504.

\setlength{\hangindent}{12pt}
\noindent
Ertefaie, A., Small, D. S. and Rosenbaum, P. R. (2018). Quantitative evaluation of the trade-off of strengthened instruments and sample size in observational studies. \textit{Journal of the American Statistical Association.} \textbf{113(523)} 1122-1134.


\setlength{\hangindent}{12pt}
\noindent
Fink, G., Günther, I. and Hill, K. (2011). The effect of water and sanitation on child health: evidence from the demographic and health surveys 1986–2007. \textit{International Journal of Epidemiology.} \textbf{40(5)} 1196-1204.

\setlength{\hangindent}{12pt}
\noindent
Fogarty, C. B. and Small, D. S. (2016). Sensitivity analysis for multiple comparisons in matched observational studies through quadratically constrained linear programming. \textit{Journal of the American Statistical Association.} \textbf{111(516)} 1820-1830.

\setlength{\hangindent}{12pt}
\noindent
Fogarty, C. B., Lee, K., Kelz, R. R. and Keele, L. J. (2019). Biased Encouragements and Heterogeneous Effects in an Instrumental Variable Study of Emergency General Surgical Outcomes. \textit{arXiv preprint.} arXiv:1909.09533.


\setlength{\hangindent}{12pt}
\noindent
Garrett, J. L. and Ruel, M. T. (2005). Stunted child–overweight mother pairs: prevalence and association with economic development and urbanization. \textit{Food and Nutrition Bulletin.} \textbf{26(2)} 209-221.

\setlength{\hangindent}{12pt}
\noindent
Gastwirth, J. L., Krieger, A. M. and Rosenbaum, P. R. (2000). Asymptotic separability in sensitivity analysis. \textit{Journal of the Royal Statistical Society: Series B (Statistical Methodology).} \textbf{62(3)} 545-555.

\setlength{\hangindent}{12pt}
\noindent
Hansen, B. B. (2004). Full matching in an observational study of coaching for the SAT. \textit{Journal of the American Statistical Association.} \textbf{99(467)} 609-618.

\setlength{\hangindent}{12pt}
\noindent
Hansen, B. B. and Klopfer, S. O. (2006). Optimal full matching and related designs via network flows. \textit{Journal of Computational and Graphical Statistics.} \textbf{15(3)} 609-627.

\setlength{\hangindent}{12pt}
\noindent
Hansen, B. B., Rosenbaum, P. R. and Small, D. S. (2014). Clustered treatment assignments and sensitivity to unmeasured biases in observational studies. \textit{Journal of the American Statistical Association.} \textbf{109(505)} 133-144.

\setlength{\hangindent}{12pt}
\noindent
Harris, N.S., Crawford, P.B., Yangzom, Y., Pinzo, L., Gyaltsen, P. and Hudes, M. (2001).  Nutritional and health status of Tibetan children living at high altitudes.  \textit{New England Journal of Medicine.} \textbf{344} 341-347.

\setlength{\hangindent}{12pt}
\noindent
Heller, R., Rosenbaum, P. R. and Small, D. S. (2009). Split samples and design sensitivity in observational studies. \textit{Journal of the American Statistical Association.} \textbf{104(487)} 1090-1101.

\setlength{\hangindent}{12pt}
\noindent
Hosman, C. A., Hansen, B. B. and Holland, P. W. (2010). The sensitivity of linear regression coefficients’ confidence limits to the omission of a confounder. \textit{The Annals of Applied Statistics.} \textbf{4(2)} 849-870.

\setlength{\hangindent}{12pt}
\noindent
Howard, S. R. and Pimentel, S. D. (2019). The uniform general signed rank test and its design sensitivity. \textit{Biometrika.} Forthcoming.




\setlength{\hangindent}{12pt}
\noindent
Keele, L. and Quinn, K. M. (2017). Bayesian sensitivity analysis for causal effects from $2\times2 $ tables in the presence of unmeasured confounding with application to presidential campaign visits. \textit{The Annals of Applied Statistics.} \textbf{11(4)} 1974-1997.




\setlength{\hangindent}{12pt}
\noindent
Mantel, N. and Haenszel, W. (1959). Statistical aspects of the analysis of data from retrospective studies of disease. \textit{Journal of the National Cancer Institute.} \textbf{22(4)} 719-748.

\setlength{\hangindent}{12pt}
\noindent
McCandless, L. C., Gustafson, P. and Levy, A. (2007). Bayesian sensitivity analysis for unmeasured confounding in observational studies. \textit{Statistics in Medicine.} \textbf{26(11)} 2331-2347.

\setlength{\hangindent}{12pt}
\noindent
Mitra, N. and Heitjan, D. F. (2007). Sensitivity of the hazard ratio to nonignorable treatment assignment in an observational study. \textit{Statistics in Medicine.} \textbf{26(6)} 1398-1414.


\setlength{\hangindent}{12pt}
\noindent
Neyman, J. S. (1923). On the Application of Probability Theory to Agricultural Experiments. Essay on Principles. Section 9. (Tlanslated and edited by DM Dabrowska and TP Speed, Statistical Science (1990), 5, 465-480). \textit{Annals of Agricultural Sciences.} \textbf{10} 1-51.

\setlength{\hangindent}{12pt}
\noindent
Null, C., Stewart, C. P., Pickering, A. J., ... and Hubbard, A. E. (2018). Effects of water quality, sanitation, handwashing, and nutritional interventions on diarrhoea and child growth in rural Kenya: a cluster-randomised controlled trial. \textit{The Lancet Global Health.} \textbf{6(3)} e316-e329.



\setlength{\hangindent}{12pt}
\noindent
Phuka, J.C., Maleta, K., Thakwalakwa, C.,  Bun Cheung, Y., Briend, A., Manary, M.J. and Ashorn, P.  (2008).  Complementary Feeding With Fortified Spread and Incidence of Severe Stunting in 6- to 18-Month-Old Rural Malawians.  \textit{JAMA Pediatrics.} \textbf{162} 1619-1626.

\setlength{\hangindent}{12pt}
\noindent
Pimentel, S. D., Kelz, R. R., Silber, J. H. and Rosenbaum, P. R. (2015). Large, sparse optimal matching with refined covariate balance in an observational study of the health outcomes produced by new surgeons. \textit{Journal of the American Statistical Association.} \textbf{110(510)} 515-527.


\setlength{\hangindent}{12pt}
\noindent
Rosenbaum, P. R. (2002a). Covariance adjustment in randomized experiments and observational studies. \textit{Statistical Science.} \textbf{17(3)} 286-327.

\setlength{\hangindent}{12pt}
\noindent
Rosenbaum, P. R. (2002b). \textit{Observational Studies.} New York: Springer.

\setlength{\hangindent}{12pt}
\noindent
Rosenbaum, P. R. (2004). Design sensitivity in observational studies. \textit{Biometrika.} \textbf{91(1)} 153-164.

\setlength{\hangindent}{12pt}
\noindent
Rosenbaum, P. R. (2007a). Confidence intervals for uncommon but dramatic responses to treatment. \textit{Biometrics.} \textbf{63(4)} 1164-1171.

\setlength{\hangindent}{12pt}
\noindent
Rosenbaum, P. R. (2007b). Sensitivity analysis for m‐estimates, tests, and confidence intervals in matched observational studies. \textit{Biometrics.} \textbf{63(2)} 456-464.

\setlength{\hangindent}{12pt}
\noindent
Rosenbaum, P. R. (2010a). Design sensitivity and efficiency in observational studies. \textit{Journal of the American Statistical Association.} \textbf{105(490)} 692-702.

\setlength{\hangindent}{12pt}
\noindent
Rosenbaum, P. R. (2010b). \textit{Design of Observational Studies.} NY: Springer

\setlength{\hangindent}{12pt}
\noindent
Rosenbaum, P. R. (2011). A new u‐Statistic with superior design sensitivity in matched observational studies. \textit{Biometrics.} \textbf{67(3)} 1017-1027.

\setlength{\hangindent}{12pt}
\noindent
Rosenbaum, P. R. (2012). Testing one hypothesis twice in observational studies. \textit{Biometrika.} \textbf{99(4)} 763-774.

\setlength{\hangindent}{12pt}
\noindent
Rosenbaum, P. R. (2013). Impact of multiple matched controls on design sensitivity in observational studies. \textit{Biometrics.} \textbf{69(1)} 118-127.

\setlength{\hangindent}{12pt}
\noindent
Rosenbaum, P. R. (2014). Weighted M-statistics with superior design sensitivity in matched observational studies with multiple controls. \textit{Journal of the American Statistical Association.} \textbf{109(507)} 1145-1158.

\setlength{\hangindent}{12pt}
\noindent
Rosenbaum, P. R. (2015). Bahadur efficiency of sensitivity analyses in observational studies. \textit{Journal of the American Statistical Association.} \textbf{110(509)} 205-217.

\setlength{\hangindent}{12pt}
\noindent
Rosenbaum, P. R. (2016). The cross‐cut statistic and its sensitivity to bias in observational studies with ordered doses of treatment. \textit{Biometrics.} \textbf{72(1)} 175-183.

\setlength{\hangindent}{12pt}
\noindent
Rosenbaum, P. R. (2017). \textit{Observation and Experiment: An Introduction to Causal Inference.} Harvard University Press.

\setlength{\hangindent}{12pt}
\noindent
Rosenbaum, P. R. and Rubin, D. B. (1983). The central role of the propensity score in observational studies for causal effects. \textit{Biometrika.} \textbf{70(1)} 41-55.

\setlength{\hangindent}{12pt}
\noindent
Rosenbaum, P. R. and Silber, J. H. (2008). Aberrant effects of treatment. \textit{Journal of the American Statistical Association.} \textbf{103(481)} 240-247.

\setlength{\hangindent}{12pt}
\noindent
Rosenbaum, P. R. and Small, D. S. (2017). An adaptive Mantel–Haenszel test for sensitivity analysis in observational studies. \textit{Biometrics.} \textbf{73(2)} 422-430.

\setlength{\hangindent}{12pt}
\noindent
Rubin, D. B. (1974). Estimating causal effects of treatments in randomized and nonrandomized studies. \textit{Journal of Educational Psychology.} \textbf{66(5)} 688.

\setlength{\hangindent}{12pt}
\noindent
Rubin, D. B. (2006). \textit{Matched Sampling for Causal Effects.} Cambridge University Press.

\setlength{\hangindent}{12pt}
\noindent
Shauly‐Aharonov, M. (2020). An exact test with high power and robustness to unmeasured confounding effects. \textit{Statistics in Medicine.} To appear.


\setlength{\hangindent}{12pt}
\noindent
Shepherd, B. E., Gilbert, P. B., Jemiai, Y. and Rotnitzky, A. (2006). Sensitivity analyses comparing outcomes only existing in a subset selected post‐randomization, conditional on covariates, with application to HIV vaccine trials. \textit{Biometrics.} \textbf{62(2)} 332-342.

\setlength{\hangindent}{12pt}
\noindent
Slepian, D. (1962). The one‐sided barrier problem for Gaussian noise. \textit{Bell System Technical Journal.} \textbf{41(2)} 463-501.

\setlength{\hangindent}{12pt}
\noindent
Small, D. S., Cheng, J., Halloran, M. E. and Rosenbaum, P. R. (2013). Case definition and design sensitivity. \textit{Journal of the American Statistical Association.} \textbf{108(504)} 1457-1468.

\setlength{\hangindent}{12pt}
\noindent
Stuart, E. A. (2010). Matching methods for causal inference: A review and a look forward. \textit{Statistical Science} \textbf{25(1)} 1.

\setlength{\hangindent}{12pt}
\noindent
Stuart, E. A. and Hanna, D. B. (2013). Commentary: Should Epidemiologists Be More Sensitive to Design Sensitivity?. \textit{Epidemiology.} \textbf{24(1)} 88-89.




\setlength{\hangindent}{12pt}
\noindent
Van de Poel, E., Hosseinpoor, A. R., Jehu-Appiah, C., Vega, J. and Speybroeck, N. (2007). Malnutrition and the disproportional burden on the poor: the case of Ghana. \textit{International Journal for Equity in Health.} \textbf{6(1)} 21.

\setlength{\hangindent}{12pt}
\noindent
VanderWeele, T. J. and Ding, P. (2017). Sensitivity analysis in observational research: introducing the E-value. \textit{Annals of Internal Medicine.} \textbf{167(4)} 268-274.


\setlength{\hangindent}{12pt}
\noindent
Walker, S. P., Powell, C. A., Grantham-McGregor, S. M., Himes, J. H. and Chang, S. M. (1991). Nutritional supplementation, psychosocial stimulation, and growth of stunted children: the Jamaican study. \textit{Am. J. Clin. Nutr.} \textbf{54(4)} 642-648.

\setlength{\hangindent}{12pt}
\noindent
WHO (1986).  Use and interpretation of anthropometric indicators of nutritional status.  \textit{Bulletin of the World Health Organization.} \textbf{64} 929-941.

\setlength{\hangindent}{12pt}
\noindent
WHO (2006). WHO child growth standards: length/height-for-age, weight-for-age, weight-for-length, weight-for-height and body mass index-for-age: methods and development.

\setlength{\hangindent}{12pt}
\noindent
WHO (2017). Stunting in a Nutshell. \url{https://www.who.int/nutrition/healthygrowthproj_stunted_videos/en/}. Accessed, 02-20.



\setlength{\hangindent}{12pt}
\noindent
Zhang, K., Small, D. S., Lorch, S., Srinivas, S. and Rosenbaum, P. R. (2011). Using split samples and evidence factors in an observational study of neonatal outcomes. \textit{Journal of the American Statistical Association.} \textbf{106(494)} 511-524.

\setlength{\hangindent}{12pt}
\noindent
Zhao, Q. (2018). On sensitivity value of pair-matched observational studies. \textit{Journal of the American Statistical Association.} 1-10.

\setlength{\hangindent}{12pt}
\noindent
Zhao, Q., Small, D. S. and Rosenbaum, P. R. (2018). Cross-screening in observational studies that test many hypotheses. \textit{Journal of the American Statistical Association.} \textbf{113(523)} 1070-1084.

\setlength{\hangindent}{12pt}
\noindent
Zhu, C., Byrd, R. H., Lu, P. and Nocedal, J. (1997). Algorithm 778: L-BFGS-B: Fortran subroutines for large-scale bound-constrained optimization. \textit{ACM Transactions on Mathematical Software (TOMS).} \textbf{23(4)} 550-560.

\setlength{\hangindent}{12pt}
\noindent
Zubizarreta, J. R. (2012). Using mixed integer programming for matching in an observational study of kidney failure after surgery. \textit{Journal of the American Statistical Association.} \textbf{107(500)} 1360-1371.

\setlength{\hangindent}{12pt}
\noindent
Zubizarreta, J. R., Cerdá, M. and Rosenbaum, P. R. (2013). Effect of the 2010 Chilean Earthquake on posttraumatic stress reducing sensitivity to unmeasured bias through study design. \textit{Epidemiology} \textbf{24(1)} 79.

\setlength{\hangindent}{12pt}
\noindent
Zubizarreta, J. R., Paredes, R. D. and Rosenbaum, P. R. (2014). Matching for balance, pairing for heterogeneity in an observational study of the effectiveness of for-profit and not-for-profit high schools in Chile. \textit{The Annals of Applied Statistics.} \textbf{8(1)} 204-231.

\clearpage

  \begin{center}
    {\LARGE\bf Online Supplementary Materials for ``Increasing Power for Observational Studies of Aberrant Response: An Adaptive Approach"}
\end{center}
  
 \begin{abstract}
    Appendix A contains proofs for Propositions \ref{prop:example} and \ref{proposition:level}, Theorems \ref{dsformula} and \ref{th:gainsinds}, and related simulations. Appendix B contains additional calculations of design sensitivities with correlation within matched strata. Appendix C contains more details on the regularity assumptions used when deriving the adaptive test. Appendix D contains more details on the second stage of the two-stage programming method. Appendix E shows how the adaptive testing procedure described in the main text can be directly extended to work for general full matching case. Appendix F investigates the simulated size of a sensitivity analysis with the Mantel-Haenszel test, the aberrant rank test and the adaptive test. Appendix G gives more details on Section~\ref{sec:intro} and Section~\ref{sec:real data}.
 \end{abstract}

\begin{center}
{\large\bf Appendix A: Proofs and Related Simulations}
\end{center}

\subsection*{Proof of Proposition~\ref{prop:example}}

\begin{proof}
The validity of Assumptions~\ref{ass:iid} and \ref{ass:nonzeromeasure} for each example follows immediately from the general assumptions on $(R_{i1},\dots, R_{im})$, so we just need to check the validity of Assumption~\ref{ass:treateffect}. For Example 1 with $\beta>0$ and $c \in \mathbb{R}$, $\mathbb{P}(R_{i1}\geq t)=\mathbb{P}(R_{i2} \geq t-\beta)>\mathbb{P}(R_{i2}\geq t)$ for all $t\geq c$. For Example 2 with $\delta > 1$ and $c > 0$, $\mathbb{P}(R_{i1}\geq t)=\mathbb{P}(R_{i2} \geq \frac{t}{\delta})>\mathbb{P}(R_{i2}\geq t)$ for all $t\geq c$. For Example 3 with $0< p < 1$, $q > 1$ and $c \in \mathbb{R}$, $\mathbb{P}(R_{i1}\geq t)=1-F_{1}(t)=1-p\cdot F_{2}^{q}(t)-(1-p)\cdot F_{2}(t)>1-F_{2}(t)=\mathbb{P}(R_{i2}\geq t)$ for all $t \geq c$. Thus, in Examples 1-3, $\mathbb{P}(R_{i1}\geq t)>\frac{1}{m}\cdot \{ \mathbb{P}(R_{i1}\geq t)+(m-1)\cdot \mathbb{P}(R_{i2}\geq t)\}=\frac{1}{m}\sum_{j=1}^{m}\mathbb{P}(R_{ij}\geq t)$ holds true for any $t\geq c$. From the general assumptions on $(R_{i1}, \dots, R_{im})$, $\mathbb{P}(R_{i(m)}> R_{i1} \geq t)$ for some $t\geq c$ is trivially true. Thus, Assumption~\ref{ass:treateffect} also holds for Examples 1-3.
\end{proof}

\subsection*{Proof of Theorem~\ref{dsformula}}

\begin{Lemma}\label{lem:uniformapprox}
Let $G(v)=\frac{1}{m}\sum_{j=1}^{m}\max \{ F_{j}(v)-F_{j}(c), 0\}$. Under Assumption~\ref{ass:iid}, we have as $I \rightarrow \infty$,
\begin{equation*}
	\sup_{v}\Big|\frac{q(v\mid \mathbf{R})}{mI}-G(v)\Big|\xrightarrow{a.s.} 0.
\end{equation*}
\end{Lemma}

\begin{proof}
	We have the following expression
	\begin{align*}
		\sup_{v}\Big|\frac{q(v\mid \mathbf{R})}{mI}-G(v)\Big|&=\sup_{v}\Big|\frac{1}{m}\sum_{j^{\prime}=1}^{m}\Big[\frac{1}{I}\sum_{i^{\prime}=1}^{I}\mathbbm{1}(v \geq R_{i^{\prime}j^{\prime}}>c)-\max\{F_{j^{\prime}}(v)-F_{j^{\prime}}(c),0\}\Big]\Big|\\
		&\leq \sup_{v}\frac{1}{m}\sum_{j^{\prime}=1}^{m}\Big|\frac{1}{I}\sum_{i^{\prime}=1}^{I}\mathbbm{1}(v \geq R_{i^{\prime}j^{\prime}}>c)-\max\{F_{j^{\prime}}(v)-F_{j^{\prime}}(c),0\}\Big|\\
		&\leq \frac{1}{m}\sum_{j^{\prime}=1}^{m}\sup_{v} \Big|\frac{1}{I}\sum_{i^{\prime}=1}^{I}\mathbbm{1}(v \geq R_{i^{\prime}j^{\prime}}>c)-\max\{F_{j^{\prime}}(v)-F_{j^{\prime}}(c),0\}\Big|.
\end{align*}
First, for each $j^{\prime}$, we have
\begin{equation*}
	\mathbb{E}\{\mathbbm{1}(v \geq R_{i^{\prime}j^{\prime}}>c)\}=\mathbb{P}(v \geq R_{i^{\prime}j^{\prime}}>c)=\max\{F_{j^{\prime}}(v)-F_{j^{\prime}}(c),0\}.
\end{equation*}
Second, for each $j^{\prime}$ in the above sum, the bracketing number of $\mathbbm{1}(v \geq R_{i^{\prime}j^{\prime}}>c)$ is bounded by Example 19.6 in Van der Vaart (2000) where we replace $t_{0}=-\infty$ with $t_{0}=c$. Combining these two facts together, for each $j^{\prime}$, each $\sup_{v}$ term goes to zero a.s. and we have the desired result.
\end{proof}

\begin{Lemma}\label{lem:asapprox}
    Under Assumption~\ref{ass:iid}, we have as $I \rightarrow \infty$,
    \begin{equation*}
    	\frac{q(R_{ij}\mid \mathbf{R})}{mI} \xrightarrow{a.s.} G(R_{ij})\quad \text{and} \quad	\frac{q(R_{i(j)}\mid \mathbf{R})}{mI} \xrightarrow{a.s.} G(R_{i(j)}).
    \end{equation*}
\end{Lemma}

\begin{proof}
The conclusion follows immediately from Lemma~\ref{lem:uniformapprox} and the fact that $|\frac{q(R_{ij}\mid \mathbf{R})}{mI}-G(R_{ij})|\leq \sup_{v}\Big|\frac{q(v\mid \mathbf{R})}{mI}-G(v)\Big|$ and $|\frac{q(R_{i(j)}|\mathbf{R})}{mI}-G(R_{i(j)})|\leq \sup_{v}\Big|\frac{q(v\mid \mathbf{R})}{mI}-G(v)\Big|$.
\end{proof}

\begin{Lemma}\label{lem:expectationunderthealter}
Under Assumption~\ref{ass:iid}, we have as $I \rightarrow \infty$,
\begin{equation*}
	\frac{1}{I} \sum_{i=1}^{I}\frac{q(R_{i1}\mid \mathbf{R})}{mI}\xrightarrow{a.s.}  \mathbb{E}\{G(R_{i1})\}.
\end{equation*}
\end{Lemma}

\begin{proof}
Note that
\begin{equation}\label{twoterminGRi1}
    \frac{1}{I} \sum_{i=1}^{I}\frac{q(R_{i1}\mid \mathbf{R})}{mI}=\frac{1}{I} \sum_{i=1}^{I}G(R_{i1})+\frac{1}{I}\sum_{i=1}^{I}\Big\{ \frac{q(R_{i1}\mid \mathbf{R})}{mI}-G(R_{i1})\Big\}.
\end{equation}
Since $G(R_{i1})$, $i=1,2,\dots$ are bounded and iid, by the law of large numbers, the first term in the RHS of (\ref{twoterminGRi1}) converges to $\mathbb{E}\{G(R_{i1})\}$ a.s.. By Lemma~\ref{lem:asapprox}, the second term in the RHS of (\ref{twoterminGRi1}) converges to 0 a.s.. So the desired result follows.
\end{proof}

\begin{Lemma}\label{lem:expectationunderthenull}
Under Assumption~\ref{ass:iid}, we have as $I \rightarrow \infty$,
\begin{equation*}
	\frac{1}{I}\sum_{i=1}^{I}\frac{\overline{\overline{\mu_{i}}}}{mI} \xrightarrow{a.s.} \mathbb{E} \Big\{ \max_{b\in [m-1]} \frac{\sum_{j=1}^{b}G(R_{i(j)})+\Gamma \sum_{j=b+1}^{m}G(R_{i(j)})}{b+\Gamma(m-b)} \Big\}.
\end{equation*}
\end{Lemma}

\begin{proof} 
Note that
	\begin{align*}
	    &\quad \Bigg| \frac{ \overline{\overline{\mu}}_{i}}{mI}-\max_{b\in [m-1]} \frac{\sum_{j=1}^{b}G(R_{i(j)})+\Gamma \sum_{j=b+1}^{m}G(R_{i(j)})}{b+\Gamma(m-b)} \Bigg|\\
		&=\Bigg| \max_{b\in [m-1]}\frac{\sum_{j=1}^{b}\frac{q(R_{i(j)}\mid \textbf{R})}{mI}+\Gamma \sum_{j=b+1}^{m}\frac{q(R_{i(j)}\mid \textbf{R})}{mI}}{b+\Gamma(m-b)}-\max_{b\in [m-1]}\frac{\sum_{j=1}^{b}G(R_{i(j)})+\Gamma \sum_{j=b+1}^{m}G(R_{i(j)})}{b+\Gamma(m-b)} \Bigg|\\
		&\leq \max_{b\in [m-1]} \Bigg| \frac{\sum_{j=1}^{b}\frac{q(R_{i(j)}\mid \textbf{R})}{mI}+\Gamma \sum_{j=b+1}^{m}\frac{q(R_{i(j)}\mid \textbf{R})}{mI}}{b+\Gamma(m-b)}- \frac{\sum_{j=1}^{b}G(R_{i(j)})+\Gamma \sum_{j=b+1}^{m}G(R_{i(j)})}{b+\Gamma(m-b)}  \Bigg| \\
		&\leq \max_{b\in [m-1]} \frac{\sum_{j=1}^{b}|\frac{q(R_{i(j)}\mid \textbf{R})}{mI}-G(R_{i(j)})| +\Gamma \sum_{j=b+1}^{m}|\frac{ q(R_{i(j)}\mid \textbf{R})}{mI}-G(R_{i(j)})|} {b+\Gamma(m-b)}, 
	\end{align*}
together with Lemma~\ref{lem:asapprox}, we have as $I \rightarrow \infty$,
\begin{equation*}
    \frac{ \overline{\overline{\mu}}_{i}}{mI}\xrightarrow{a.s.} \max_{b\in [m-1]} \frac{\sum_{j=1}^{b}G(R_{i(j)})+\Gamma \sum_{j=b+1}^{m}G(R_{i(j)})}{b+\Gamma(m-b)}.
\end{equation*}

Note that
\begin{align}\label{equa:twotermunderthenull}
	\frac{1}{I}\sum_{i=1}^{I}\frac{\overline{\overline{\mu_{i}}}}{mI}
	&=\frac{1}{I}\sum_{i=1}^{I}\max_{b\in [m-1]} \frac{\sum_{j=1}^{b}G(R_{i(j)})+\Gamma \sum_{j=b+1}^{m}G(R_{i(j)})}{b+\Gamma(m-b)}\nonumber \\
	&\quad +\frac{1}{I} \sum_{i=1}^{I}\Big\{ \frac{ \overline{\overline{\mu}}_{i}}{mI}-\max_{b\in [m-1]} \frac{\sum_{j=1}^{b}G(R_{i(j)})+\Gamma \sum_{j=b+1}^{m}G(R_{i(j)})}{b+\Gamma(m-b)}\Big\}.
\end{align}

For the first term in the RHS of (\ref{equa:twotermunderthenull}), note that $\max_{b\in [m-1]} \frac{\sum_{j=1}^{b}G(R_{i(j)})+\Gamma \sum_{j=b+1}^{m}G(R_{i(j)})}{b+\Gamma(m-b)}$, $i=1,2,\dots$ are bounded iid random variables, by the strong law of large numbers, we have as $I \rightarrow \infty$,
\begin{equation*}
\frac{1}{I}\sum_{i=1}^{I}\max_{b\in [m-1]} \frac{\sum_{j=1}^{b}G(R_{i(j)})+\Gamma \sum_{j=b+1}^{m}G(R_{i(j)})}{b+\Gamma(m-b)}\xrightarrow{a.s.}\mathbb{E} \Big\{ \max_{b\in [m-1]} \frac{\sum_{j=1}^{b}G(R_{i(j)})+\Gamma \sum_{j=b+1}^{m}G(R_{i(j)})}{b+\Gamma(m-b)} \Big\}.
\end{equation*}

The second term in the RHS of (\ref{equa:twotermunderthenull}) has been shown to converge to zero almost surely. So the desired conclusion follows.
\end{proof}

For simplicity, from now on, let
\begin{equation*}
	\varphi(\Gamma)=\mathbb{E} \Big\{ \max_{b\in [m-1]} \frac{\sum_{j=1}^{b}G(R_{i(j)})+\Gamma \sum_{j=b+1}^{m}G(R_{i(j)})}{b+\Gamma(m-b)} \Big\}.
\end{equation*}

\begin{Lemma}\label{lem:continuity}
	$\varphi(\Gamma)$ is continuous on $[1,+\infty)$.
\end{Lemma}

\begin{proof}
For any $\Gamma_{1}, \Gamma_{2} \in [1,+\infty)$,
\begin{align*}
	&\quad |\varphi(\Gamma_{1})-\varphi(\Gamma_{2})|\\
	&\leq \mathbb{E} \Big| \max_{b\in [m-1]} \frac{\sum_{j=1}^{b}G(R_{i(j)})+\Gamma_{1} \sum_{j=b+1}^{m}G(R_{i(j)})}{b+\Gamma_{1}(m-b)}-\max_{b\in [m-1]} \frac{\sum_{j=1}^{b}G(R_{i(j)})+\Gamma_{2} \sum_{j=b+1}^{m}G(R_{i(j)})}{b+\Gamma_{2}(m-b)} \Big|\\
	&\leq  \mathbb{E}\Big\{ \max_{b\in [m-1]} \Big|  \frac{\sum_{j=1}^{b}G(R_{i(j)})+\Gamma_{1} \sum_{j=b+1}^{m}G(R_{i(j)})}{b+\Gamma_{1}(m-b)}- \frac{\sum_{j=1}^{b}G(R_{i(j)})+\Gamma_{2} \sum_{j=b+1}^{m}G(R_{i(j)})}{b+\Gamma_{2}(m-b)} \Big| \Big\}  \\
	&=  \mathbb{E}\Big \{\max_{b\in [m-1]} \Big|\Big[ \frac{\sum_{j=b+1}^{m}G(R_{i(j)})}{m-b}+\Big\{ \sum_{j=1}^{b}G(R_{i(j)})-\frac{b\cdot \sum_{j=b+1}^{m}G(R_{i(j)})}{m-b} \Big \} \cdot \frac{1}{b+\Gamma_{1}(m-b)}\Big] \\
	&\quad \quad \quad \quad - \Big[ \frac{\sum_{j=b+1}^{m}G(R_{i(j)})}{m-b}+\Big \{ \sum_{j=1}^{b}G(R_{i(j)})-\frac{b\cdot \sum_{j=b+1}^{m}G(R_{i(j)})}{m-b} \Big\} \cdot \frac{1}{b+\Gamma_{2}(m-b)} \Big] \Big| \Big\}\\
	&=\mathbb{E}\Big\{ \max_{b\in [m-1]} \Big| \frac{b}{b+\Gamma_{1}(m-b)}-\frac{b}{b+\Gamma_{2}(m-b)} \Big|\cdot \Big|\frac{\sum_{j=1}^{b}G(R_{i(j)})}{b}-\frac{\sum_{j=b+1}^{m}G(R_{i(j)})}{m-b} \Big| \Big\} \\
    &\leq 2 \max_{b\in [m-1]} \Big| \frac{b}{b+\Gamma_{1}(m-b)}-\frac{b}{b+\Gamma_{2}(m-b)} \Big|, \quad (\text{since $G(R_{i(j)})\leq 1$})
	\end{align*}
then continuity of $\varphi (\Gamma)$ follows from the fact that $g_{b}(\Gamma)=\frac{b}{b+\Gamma(m-b)}$ is continuous on $[1, +\infty)$ for each $b\in [m-1]$.
\end{proof}

\begin{Lemma}\label{lem:similardurrett}
   Let $X$ and $Y$ be two random variables. Set $s_{X}=\sup \{ t: \mathbb{P}(X \geq t)>0 \}$, $s_{Y}=\sup \{ t: \mathbb{P}(Y \geq t)>0 \}$. Suppose that function $h:\mathbb{R} \rightarrow \mathbb{R}$ is continuously differentiable on $(c, +\infty)$, $h(c)=0$, $h^{\prime}(t)>0$ for $t \in (c, \ s_{X} \vee s_{Y})$, and $\mathbb{E}|h(X)\mathbbm{1}_{X\geq c}|<\infty$, $\mathbb{E}|h(Y)\mathbbm{1}_{Y \geq c}|<\infty$. If for any $t \in (c, \ s_{X} \vee s_{Y})$,  $\mathbb{P}(X \geq t) \leq \mathbb{P}(Y \geq t)$, we have $\mathbb{E}\{h(X)\mathbbm{1}_{X\geq c}\} \leq \mathbb{E}\{h(Y)\mathbbm{1}_{Y \geq c}\}$, and if there exists an open interval $\mathcal{I}\subset (c, \ s_{X} \vee s_{Y})$ such that $\mathbb{P}(X \geq t) < \mathbb{P}(Y \geq t)$ for any $t\in \mathcal{I}$, we have $\mathbb{E}\{h(X)\mathbbm{1}_{X\geq c}\} < \mathbb{E}\{h(Y)\mathbbm{1}_{Y \geq c}\}$.
\end{Lemma}

\begin{proof}
	We have
	\begin{align*}
		\mathbb{E}\{h(X)\mathbbm{1}_{X\geq c}\} 
		&=\int_{\Omega}\int_{c}^{X}h^{\prime}(t)\mathbbm{1}_{X\geq c} \ dt \ d\mathbb{P} \quad \text{(since $h(c)=0$)} \\
		&=\int_{\Omega}\int_{c}^{+\infty}h^{\prime}(t)\mathbbm{1}_{X\geq t} \ dt \ d\mathbb{P}\\
		&=\int_{c}^{+\infty} \int_{\Omega}\ h^{\prime}(t)\mathbbm{1}_{X\geq t } \ d\mathbb{P} \ dt \quad \text{(by Fubini's theorem)} \\
		&=\int_{c}^{+\infty}h^{\prime}(t)\mathbb{P}(X\geq t)dt\\
		&=\int_{c}^{s_{X} }h^{\prime}(t)\mathbb{P}(X\geq t)dt.
	\end{align*}
		
Similarly, we have
\begin{equation*}
	\mathbb{E}\{h(Y)\mathbbm{1}_{Y\geq c}\}=\int_{c}^{s_{Y}}h^{\prime}(t)\mathbb{P}(Y\geq t)dt.
\end{equation*}

Note that if for any $t \in (c, \ s_{X} \vee s_{Y})$,  $\mathbb{P}(X \geq t) \leq \mathbb{P}(Y \geq t)$, then $s_{X} \leq s_{Y}$. So the desired conclusion follows immediately from the above two equalities and the assumption that $h^{\prime}(t)>0$ for $t \in (c, \ s_{X} \vee s_{Y})$.  
\end{proof}

\begin{Lemma} \label{lem:boundingEG_i1}
Under Assumptions~\ref{ass:iid}-\ref{ass:treateffect}, we have
 \begin{equation*}
 \frac{1}{m}\sum_{j=1}^{m} \mathbb{E}\{G(R_{ij})\} < \mathbb{E}\{G(R_{i1})\}< \mathbb{E}\{G(R_{i(m)})\}.
 \end{equation*}
 \end{Lemma}

\begin{proof}
	The second inequality follows immediately from applying Lemma~\ref{lem:similardurrett} with $h(t)=G(t)=\frac{1}{m}\sum_{j=1}^{m}\max \{F_{j}(t)-F_{j}(c),0\}$, $X=R_{i1}$ and $Y=R_{i(m)}$. For the first inequality, let $s_{j}=\sup \{ t: \mathbb{P}(R_{ij} \geq t)>0 \}$. Assumption~\ref{ass:treateffect} implies that $s_{1}=s=\max_{j}s_{j}$. Follow a similar calculation as in Lemma~\ref{lem:similardurrett}, by Assumption~\ref{ass:treateffect} we have
	\begin{align*}
	    	\frac{1}{m}\sum_{j=1}^{m}\mathbb{E}\{G(R_{ij})\}&=\frac{1}{m}\sum_{j=1}^{m}\int_{c}^{s_{j}}G^{\prime}(t)\mathbb{P}(R_{ij}\geq t)dt\\
	    	&\leq  \frac{1}{m}\sum_{j=1}^{m} \int_{c}^{s}G^{\prime}(t)\mathbb{P}(R_{ij}\geq t)dt\\
	    	&=  \int_{c}^{s}G^{\prime}(t)\cdot \frac{1}{m}\sum_{j=1}^{m}\mathbb{P}(R_{ij}\geq t)dt\\
	    	&< \int_{c}^{s}G^{\prime}(t)\mathbb{P}(R_{i1}\geq t)dt\\
	    	&=\mathbb{E}\{G(R_{i1})\}.
	\end{align*}
\end{proof}

\begin{Lemma}\label{lem:boundinglimitofphi}
	Under Assumptions~\ref{ass:iid}-\ref{ass:treateffect}, we have
	\begin{equation*}
	 \lim_{\Gamma \rightarrow 1^{+}}\varphi(\Gamma)< \mathbb{E}\{G(R_{i1})\}, \quad \lim_{\Gamma \rightarrow +\infty}\varphi(\Gamma) > \mathbb{E}\{G(R_{i1})\}.
	 \end{equation*}
\end{Lemma}

\begin{proof}

For any $\Gamma\geq 1$, since $G(R_{i(j)})\leq 1$, we have
\begin{equation*}
0 \leq \max_{b\in [m-1]} \frac{\sum_{j=1}^{b}G(R_{i(j)})+\Gamma \sum_{j=b+1}^{m}G(R_{i(j)})}{b+\Gamma(m-b)} \leq 1.
\end{equation*}
By continuity of $\varphi$ (Lemma~\ref{lem:continuity}) and bounded convergence theorem, 
\begin{align*}
	\lim_{\Gamma \rightarrow 1^{+}}\varphi(\Gamma)&= \lim_{n \rightarrow \infty}\varphi\big(\frac{n+1}{n} \big)\\
	&=\mathbb{E} \Big\{ \lim_{n \rightarrow \infty} \max_{b\in [m-1]} \frac{\sum_{j=1}^{b}G(R_{i(j)})+\frac{n+1}{n} \sum_{j=b+1}^{m}G(R_{i(j)})}{b+\frac{n+1}{n}(m-b)} \Big\}\\
	&=\frac{\sum_{j=1}^{m}\mathbb{E}\{G(R_{i(j)})\}}{m}\\
	&=\frac{\sum_{j=1}^{m}\mathbb{E}\{G(R_{ij})\}}{m}\\
	&< \mathbb{E}\{G(R_{i1})\}, \quad \text{(by Lemma~\ref{lem:boundingEG_i1})}
	\end{align*}
\begin{align*}
	\lim_{\Gamma \rightarrow +\infty}\varphi(\Gamma)&=\lim_{n \rightarrow \infty}\varphi(n)\\
	&=\mathbb{E} \Big\{ \lim_{n \rightarrow \infty} \max_{b\in [m-1]} \frac{\sum_{j=1}^{b}G(R_{i(j)})+n \sum_{j=b+1}^{m}G(R_{i(j)})}{b+n(m-b)} \Big\}\\
	&=\mathbb{E} \Big\{ \max_{b\in [m-1]}\frac{\sum_{j=b+1}^{m}G(R_{i(j)})}{m-b}\Big\}\\
	&\geq \mathbb{E}\{G(R_{i(m)})\}\\
	&> \mathbb{E}\{G(R_{i1})\}. \quad \text{(by Lemma~\ref{lem:boundingEG_i1})}
\end{align*}
 \end{proof}
 
 \begin{Lemma}\label{lem:monotone}
   Under Assumptions~\ref{ass:iid} and~\ref{ass:nonzeromeasure}, $\varphi(\Gamma)$ is a strictly monotonically increasing function of $\Gamma$ on $[1, +\infty)$. 
 \end{Lemma}
 
 \begin{proof}
 	For any $b\in [m-1]$, and for any $1\leq \Gamma_{1} <\Gamma_{2}<+\infty$, since
 \begin{equation*}
 	\frac{\sum_{j=1}^{b}G(R_{i(j)})}{b}-\frac{\sum_{j=b+1}^{m}G(R_{i(j)})}{m-b}\leq 0,
 \end{equation*}
 we have
\begin{align*}
&\quad \frac{\sum_{j=1}^{b}G(R_{i(j)})+\Gamma_{2} \sum_{j=b+1}^{m}G(R_{i(j)})}{b+\Gamma_{2}(m-b)}-\frac{\sum_{j=1}^{b}G(R_{i(j)})+\Gamma_{1} \sum_{j=b+1}^{m}G(R_{i(j)})}{b+\Gamma_{1}(m-b)}\\
	&= \frac{\sum_{j=b+1}^{m}G(R_{i(j)})}{m-b}+\Big\{ \sum_{j=1}^{b}G(R_{i(j)})-\frac{b\cdot \sum_{j=b+1}^{m}G(R_{i(j)})}{m-b} \Big\} \cdot \frac{1}{b+\Gamma_{2}(m-b)} \\
	&\quad - \frac{\sum_{j=b+1}^{m}G(R_{i(j)})}{m-b}-\Big\{\sum_{j=1}^{b}G(R_{i(j)})-\frac{b\cdot \sum_{j=b+1}^{m}G(R_{i(j)})}{m-b} \Big\} \cdot \frac{1}{b+\Gamma_{1}(m-b)}\\
	&=b\Big\{\frac{\sum_{j=1}^{b}G(R_{i(j)})}{b}-\frac{\cdot \sum_{j=b+1}^{m}G(R_{i(j)})}{m-b} \Big\}\Big\{ \frac{1}{b+\Gamma_{2}(m-b)}-\frac{1}{b+\Gamma_{1}(m-b)} \Big\}\\
	&\geq  0,
\end{align*}
where equality holds if and only if 
\begin{equation*}
	\frac{\sum_{j=1}^{b}G(R_{i(j)})}{b}-\frac{\sum_{j=b+1}^{m}G(R_{i(j)})}{m-b}= 0.
\end{equation*}

Thus, we have
\begin{equation*}
	\max_{b\in [m-1]} \frac{\sum_{j=1}^{b}G(R_{i(j)})+\Gamma_{2} \sum_{j=b+1}^{m}G(R_{i(j)})}{b+\Gamma_{2}(m-b)} \geq  \max_{b\in [m-1]} \frac{\sum_{j=1}^{b}G(R_{i(j)})+\Gamma_{1} \sum_{j=b+1}^{m}G(R_{i(j)})}{b+\Gamma_{1}(m-b)}.
\end{equation*}

That is, to show that $\varphi(\Gamma_{2})>\varphi(\Gamma_{1})$, it suffices to show that
\begin{equation*}
	\mathbb{P} \Big\{\max_{b\in [m-1]} \frac{\sum_{j=1}^{b}G(R_{i(j)})+\Gamma_{2} \sum_{j=b+1}^{m}G(R_{i(j)})}{b+\Gamma_{2}(m-b)}>  \max_{b\in [m-1]} \frac{\sum_{j=1}^{b}G(R_{i(j)})+\Gamma_{1} \sum_{j=b+1}^{m}G(R_{i(j)})}{b+\Gamma_{1}(m-b)}\Big\}>0.
\end{equation*}

We have
\begin{align*}
&\quad 	\mathbb{P}\Big\{\max_{b\in [m-1]} \frac{\sum_{j=1}^{b}G(R_{i(j)})+\Gamma_{2} \sum_{j=b+1}^{m}G(R_{i(j)})}{b+\Gamma_{2}(m-b)}>  \max_{b\in [m-1]} \frac{\sum_{j=1}^{b}G(R_{i(j)})+\Gamma_{1} \sum_{j=b+1}^{m}G(R_{i(j)})}{b+\Gamma_{1}(m-b)}\Big\}\\
&\geq \mathbb{P}\Big[\bigcap_{b=1}^{m-1}\Big \{ \frac{\sum_{j=1}^{b}G(R_{i(j)})+\Gamma_{2} \sum_{j=b+1}^{m}G(R_{i(j)})}{b+\Gamma_{2}(m-b)}>\frac{\sum_{j=1}^{b}G(R_{i(j)})+\Gamma_{1} \sum_{j=b+1}^{m}G(R_{i(j)})}{b+\Gamma_{1}(m-b)}\Big\} \Big]\\
&= \mathbb{P} \Big[\bigcap_{b=1}^{m-1}\Big \{ \frac{\sum_{j=1}^{b}G(R_{i(j)})}{b}-\frac{\sum_{j=b+1}^{m}G(R_{i(j)})}{m-b}<0 \Big\} \Big]\\
    &\geq \mathbb{P}(R_{i(m)}> R_{i(m-1)} > \dots > R_{i(1)}\geq c ) \\
   	&=\mathbb{P}(R_{i(1)}\geq c) \quad \quad \text{(by Assumption~\ref{ass:iid})}\\
	&>0, \quad \quad \text{(by Assumption~\ref{ass:nonzeromeasure})}
\end{align*}
so the conclusion follows.
\end{proof}

\begin{proof}[Theorem 1]
By Lemma~\ref{lem:continuity}, Lemma~\ref{lem:boundinglimitofphi} and Lemma~\ref{lem:monotone}, it is clear that equation $\varphi(\Gamma)=\mathbb{E}\{G(R_{i1})\}$
has a unique solution $\widetilde{\Gamma}$ on $[1,+\infty)$.
Note that,
\begin{align}
	\Psi_{I, \Gamma}&=\mathbb{P}(T_{\text{abe}}\geq \xi_{\alpha}\mid \mathcal{Z})\nonumber \\
	&= \mathbb{P}\Bigg\{\frac{\sum_{i=1}^{I}q(R_{i1}\mid \mathbf{R})- \sum_{i=1}^{I} \overline{\overline{\mu}}_{i}}{\sqrt{\sum_{i=1}^{I}\overline{\overline{\nu}}_{i}}}\geq  \Phi^{-1}(1-\alpha) \Bigg\}\nonumber \\
	&=\mathbb{P}\Bigg\{\frac{\sqrt{I}\big( \frac{1}{I} \sum_{i=1}^{I}\frac{q(R_{i1}\mid \mathbf{R})}{mI}- \frac{1}{I} \sum_{i=1}^{I} \frac{\overline{\overline{\mu}}_{i}}{mI} \big)}{\sqrt{\frac{1}{I}\sum_{i=1}^{I}\frac{\overline{\overline{\nu}}_{i}}{(mI)^{2}}}} \geq  \Phi^{-1}(1-\alpha) \Bigg\}\label{equa:asymppower} .
\end{align}

Since $q(R_{i(j)}\mid \mathbf{R})\leq mI$, we have
\begin{equation*}
    \frac{1}{I}\sum_{i=1}^{I}\frac{\overline{\overline{\nu}}_{i}}{(mI)^{2}}=\frac{1}{I}\sum_{i=1}^{I} \frac{\max_{b \in B_{i}}\overline{\overline{\nu}}_{ib}}{(mI)^{2}}
    \leq \frac{1}{I}\sum_{i=1}^{I}\max_{b \in B_{i}}  \frac{\sum_{j=1}^{b} \frac{q^{2}(R_{i(j)}|\textbf{R})}{(mI)^{2}}+\Gamma \sum_{j=b+1}^{m}\frac{ q^{2}(R_{i(j)}|\textbf{R})}{(mI)^{2}}}{b+\Gamma(m-b)}\leq 1.
\end{equation*}
   
For $\Gamma < \widetilde{\Gamma}$, by Lemma~\ref{lem:monotone} we have  $\varphi(\widetilde{\Gamma})>\varphi(\Gamma)$. Thus, as $I \rightarrow \infty$, 
\begin{align*}
\frac{\sqrt{I}\big\{ \frac{1}{I} \sum_{i=1}^{I}\frac{q(R_{i1}|\mathbf{R})}{mI}- \frac{1}{I} \sum_{i=1}^{I} \frac{\overline{\overline{\mu}}_{i}}{mI}\big\}}{\sqrt{\frac{1}{I}\sum_{i=1}^{I}\frac{\overline{\overline{\nu}}_{i}}{(mI)^{2}}}}
	&\simeq \frac{\sqrt{I}\big\{ \mathbb{E}(G(R_{i1}))-\varphi(\Gamma) \big\}}{\sqrt{\frac{1}{I}\sum_{i=1}^{I}\frac{\overline{\overline{\nu}}_{i}}{(mI)^{2}}}} \quad \text{(by Lemma~\ref{lem:expectationunderthealter} and Lemma~\ref{lem:expectationunderthenull})}\\
	&=\frac{\sqrt{I}\big\{ \varphi(\widetilde{\Gamma})-\varphi(\Gamma) \big\}}{\sqrt{\frac{1}{I}\sum_{i=1}^{I}\frac{\overline{\overline{\nu}}_{i}}{(mI)^{2}}}}\\
	&\geq \sqrt{I}\big\{ \varphi(\widetilde{\Gamma})-\varphi(\Gamma) \big\}\\
	&\rightarrow +\infty.
\end{align*}

Similarly, we have for $\Gamma > \widetilde{\Gamma}$,
\begin{equation*}
	\frac{\sqrt{I}\big\{ \frac{1}{I} \sum_{i=1}^{I}\frac{q(R_{i1}\mid \mathbf{R})}{mI}- \frac{1}{I} \sum_{i=1}^{I} \frac{\overline{\overline{\mu}}_{i}}{mI}\big\}}{\sqrt{\frac{1}{I}\sum_{i=1}^{I}\frac{\overline{\overline{\nu}}_{i}}{(mI)^{2}}}} \rightarrow -\infty,
\end{equation*}
so the conclusion follows from (\ref{equa:asymppower}).
\end{proof}

\section*{Proof of Proposition~\ref{proposition:level}}

\begin{proof}

Suppose that $\mathbf{u}_{0}\in \mathcal{U}$ is the unknown true vector of unmeasured confounders and $\Gamma_{0}\leq \Gamma$ is the unknown true magnitude of hidden bias. Recall that $\mathbf{\rho}^{*}_{\Gamma}$ and $y^{*}_{\Gamma}$ are the optimal values of $(*)$ and $(**)$ with sensitivity parameter $\Gamma$ respectively. Since the constraint regions of $(*)$ and $(**)$ enlarge as $\Gamma$ increases, we have $y^{*}_{\Gamma_{0}}\geq y^{*}_{\Gamma}$ and $\mathbf{\rho}^{*}_{\Gamma_{0}}\geq \mathbf{\rho}^{*}_{\Gamma}$. By Slepian's lemma, $\mathbf{\rho}^{*}_{\Gamma_{0}}\geq \mathbf{\rho}^{*}_{\Gamma}$ implies $Q_{\mathbf{\rho}^{*}_{\Gamma_{0}}, \alpha}\leq Q_{\mathbf{\rho}^{*}_{\Gamma}, \alpha}$. Thus we have
\begin{align*}
\mathbb{P}_{\Gamma_{0}, \mathbf{u}_{0}}\Big(\min \limits_{\mathbf{u} \in \mathcal{U} }\ \max\limits_{k\in \{1,2\}} \frac{t_{k}-\mu_{k,\mathbf{u}}}{\sigma_{k, \mathbf{u}}} &\geq Q_{\mathbf{\rho}^{*}_{\Gamma}, \alpha}\Big|\mathcal{F}, \mathcal{Z}\Big)\\
&=\mathbb{P}_{\Gamma_{0}, \mathbf{u}_{0}}\big (y^{*}_{\Gamma}\geq Q_{\mathbf{\rho}^{*}_{\Gamma}, \alpha}\mid \mathcal{F}, \mathcal{Z}\big )\\
 &\leq \mathbb{P}_{\Gamma_{0}, \mathbf{u}_{0}}\big( y^{*}_{\Gamma_{0}}\geq Q_{\mathbf{\rho}^{*}_{\Gamma_{0}}, \alpha}\mid \mathcal{F}, \mathcal{Z}\big )\\
  &= \mathbb{P}_{\Gamma_{0},\mathbf{u}_{0}}\Big(\min \limits_{\mathbf{u} \in \mathcal{U} }\ \max\limits_{k\in \{1,2\}} \frac{t_{k}-\mu_{k,\mathbf{u}}}{\sigma_{k, \mathbf{u}}} \geq \max_{\mathbf{u}\in \mathcal{U}} Q_{\rho_{\mathbf{u}}, \alpha}\Big|\mathcal{F}, \mathcal{Z}, \Gamma=\Gamma_{0} \Big)\\
  &\leq \mathbb{P}_{\Gamma_{0}, \mathbf{u}_{0}}\Big(\min \limits_{\mathbf{u} \in \mathcal{U} }\Big \{ \max\limits_{k\in \{1,2\}} \frac{t_{k}-\mu_{k,\mathbf{u}}}{\sigma_{k, \mathbf{u}}}- Q_{\rho_{\mathbf{u}}, \alpha}\Big \}\geq 0 \Big|\mathcal{F}, \mathcal{Z}, \Gamma=\Gamma_{0}\Big)\\
  &\leq \mathbb{P}_{\Gamma_{0}, \mathbf{u}_{0}}\Big( \max\limits_{k\in \{1,2\}} \frac{t_{k}-\mu_{k,\mathbf{u}_{0}}}{\sigma_{k, \mathbf{u}_{0}}}\geq  Q_{\rho_{\mathbf{u}_{0}}, \alpha} \Big|\mathcal{F}, \mathcal{Z}, \Gamma=\Gamma_{0} \Big)\\
  &\rightarrow \alpha, \quad \text{as $I \rightarrow \infty$}
\end{align*}
so the desired conclusion follows.
\end{proof}

\section*{Proof of Theorem~\ref{th:gainsinds}}

\begin{proof}

Recall that using the Bonferroni adjustment to combine $T_{1}$ and $T_{2}$ takes the design sensitivity $\max \{ \widetilde{\Gamma}_{1}, \widetilde{\Gamma}_{2}\}$; see Rosenbaum (2012). Since inequality (\ref{inequality:minimax}) always holds, the test statistic used in testing procedure (\ref{test:adaptive}) implemented by Algorithm~\ref{algo:adaptive} uniformly dominates the one used by the Bonferroni adjustment, which implies $\widetilde{\Gamma}_{1:2} \geq \max \{ \widetilde{\Gamma}_{1}, \widetilde{\Gamma}_{2}
\}$. 

We then construct an example to show that $\widetilde{\Gamma}_{1:2}>\max \{ \widetilde{\Gamma}_{1}, \widetilde{\Gamma}_{2}\}$ is possible. That is, we show that the minimax procedure (developed in Fogarty and Small, 2016) implemented in the Step 2 of Algorithm~\ref{algo:adaptive} can result in improved design sensitivity by enforcing that unmeasured confounder must have the same impact on the probabilities of assignment to treatment for all scores in each component test on the same outcome variable.

Suppose we have $I$ matched pairs and potential responses $r_{Tij}, r_{Cij}\in S= \{a+b\sqrt{2}: a \in \mathbb{Z}, b\in \mathbb{Z}\}$, $i=1, \dots, I, j=1, 2$. For $x\in S$, we define two functions $f_{1}(x)=a$ and $f_{2}(x)=b$ where $x=a+b\sqrt{2}$ for some $a, b\in \mathbb{Z}$. Note that both $f_{1}$ and $f_{2}$ are well-defined since for $a_{1}, b_{1}, a_{2}, b_{2} \in \mathbb{Z}$, we have $a_{1}+b_{1}\sqrt{2}=a_{2}+b_{2}\sqrt{2}$ if and only if $a_{1}=a_{2}$ and $b_{1}=b_{2}$. Consider two test statistics $T_{1}=I^{-1}\sum_{i=1}^{I}D_{i}^{(1)}$ and $T_{2}=I^{-1}\sum_{i=1}^{I}D_{i}^{(2)}$ where $D_{i}^{(1)}=(Z_{i1}-Z_{i2})(f_{1}(R_{i1})-f_{1}(R_{i2}))$ and $D_{i}^{(2)}=(Z_{i1}-Z_{i2})(f_{2}(R_{i1})-f_{2}(R_{i2}))$ are treated-minus-control paired difference of $f_{1}(R_{ij})$ and $f_{2}(R_{ij})$ respectively. Note that $r_{Tij}=f_{1}(r_{Tij})+f_{2}(r_{Tij})\sqrt{2}$ and $r_{Cij}=f_{1}(r_{Cij})+f_{2}(r_{Cij})\sqrt{2}$. Suppose that for all $i=1,\dots,I, j=1, 2$, the vector $(f_{1}(r_{Tij}), f_{1}(r_{Cij}), f_{2}(r_{Tij}), f_{2}(r_{Cij}))$ are i.i.d. realizations from the distribution:
\begin{equation*}
   (f_{1}(r_{Tij}), f_{1}(r_{Cij}), f_{2}(r_{Tij}), f_{2}(r_{Cij}))=
\begin{cases}
	(3, 0, -1, 0) \quad \text{with probability 1/2}\\
	(-1, 0, 3, 0) \quad \text{with probability 1/2}.
\end{cases}
\end{equation*}
Therefore, the vector of treated-minus-control paired differences $\mathbf{D}_{i}=(D_{i}^{(1)}, D_{i}^{(2)})$ are identically distributed as:
\begin{equation}\label{twooutcomes:nulldistibution}
   (D_{i}^{(1)}, D_{i}^{(2)})=
\begin{cases}
	(3,-1) \quad \text{with probability 1/2}\\
	(-1,3) \quad \text{with probability 1/2}.
\end{cases}
\end{equation}
For $k=1,2$, let $\widetilde{\Gamma}_{k}$ denote the design sensitivity for $T_{k}$ under Fisher's sharp null of no treatment effect $H_{0}: r_{Tij}=r_{Cij}$, $i=1,\dots,I, j=1,2$, and let $\widetilde{\Gamma}_{1:2}$ denote the design sensitivity for testing $H_{0}$ with $T_{1}$ and $T_{2}$ as the two component test statistics through the minimax procedure.

We at first show that $\widetilde{\Gamma}_{1}=\widetilde{\Gamma}_{2}=3$. The design sensitivity is the value of $\widetilde{\Gamma}$ such that the worst-case expectation $\overline{\mu}_{\Gamma,k}$ of $T_{k}$ (i.e. the expectation of the limiting bounding distribution of $T_{k}$) with the magnitude of hidden bias $\Gamma=\widetilde{\Gamma}$ equals the true expectation $\mu_{k}$ of $T_{k}$ (i.e. the actual expectation of the limiting distribution of $T_{k}$) based on how the paired differences $D_{i}^{(k)}$ are generated, $k\in \{1,2\}$; see Rosenbaum (2004). In this case, according to (\ref{twooutcomes:nulldistibution}), it is clear that $\mu_{1}=\mu_{2}=1$. To find the worst-case expectation of $T_{k}$ at a given $\Gamma$, it is clear that under Fisher's sharp null $H_{0}$, any paired difference $D_{i}^{(k)}$ where a 3 was observed should be assigned probability $\Gamma/(1+\Gamma)$ for the treated unit in the pair having the higher value of $f_{k}$. Similarly, when a $-1$ is observed the probability that the treated unit had the lower value of $f_{k}$ should be set to $1/(1+\Gamma)$. For any $\Gamma$, the worst-case expectation $\overline{\mu}_{\Gamma,k}$ of $T_{k}$ is then:
\begin{equation*}
    \overline{\mu}_{\Gamma,k}=\frac{1}{2}\cdot 3 \cdot \big(\frac{\Gamma}{1+\Gamma}-\frac{1}{1+\Gamma}\big)+\frac{1}{2}\cdot (-1) \cdot \big(\frac{1}{1+\Gamma}-\frac{\Gamma}{1+\Gamma}\big)=\frac{2\Gamma-2}{1+\Gamma}.
\end{equation*}
To obtain $\widetilde{\Gamma}_{k}$, we just need to solve the equation $\overline{\mu}_{\widetilde{\Gamma}_{k}, 1}=\mu_{k}$ with $\widetilde{\Gamma}_{k}$, which, from the above arguments, can be written as: $(2\widetilde{\Gamma}_{k}-2)/(1+\widetilde{\Gamma}_{k})=1$. Thus, $\widetilde{\Gamma}_{k}=3$ for $k=1,2$.

We then show that $\widetilde{\Gamma}_{1:2}=+\infty$. Note that asymptotically the minimax procedure fails to reject $H_{0}$ if the maximum over the unmeasured confounders of the minimum over the outcomes of the expectations of $T_{1}$ and $T_{2}$ under Fisher's sharp null $H_{0}$ exceeds the true expectation of the test statistic (in our example, 1 for both $T_{1}$ and $T_{2}$). Hence, the design sensitivity is the value of $\Gamma$ such that the worst-case expectations under Fisher's sharp null for both test statistics exceed their true expectations. Asymptotically, of the $I$ matched pairs $I/2$ will have the observed paired difference of $(3,-1)$ for their two score functions $(f_{1}, f_{2})$ and $I/2$ will have an observed paired difference of $(-1,3)$. Separating the observed pairs into two sets according to their paired difference, let $p_{i}$ be the probability that the treated individual receives the treatment in pair $i$ in the $(3,-1)$ group, and let $q_{i}$ be the probability that the treated individual receives the treatment in pair $i$ of the $(-1,3)$ group. For any given $\Gamma$, consider the following optimization problem:
\begin{align*}
    &\underset{y, p_{i}, q_{i}}{\text{maximize}} \quad y \quad \quad \quad \quad \quad (***)\\
    &\text{subject to} \ \ y \leq \frac{1}{I}\sum_{i=1}^{I/2}\{3\cdot(2p_{i}-1)+(-1)\cdot (2q_{i}-1)\}\\
    &\quad \quad \quad \quad \quad y \leq \frac{1}{I}\sum_{i=1}^{I/2}\{(-1)\cdot(2p_{i}-1)+3\cdot (2q_{i}-1)\} \\
    &\quad \quad \quad \quad \quad \frac{1}{1+\Gamma} \leq p_{i} \leq \frac{\Gamma}{1+\Gamma} \quad i=1,\dots, I/2 \\
    &\quad \quad \quad \quad \quad \frac{1}{1+\Gamma} \leq q_{i} \leq \frac{\Gamma}{1+\Gamma}. \quad  i=1,\dots, I/2
\end{align*}
Let $\mathbf{x}=(y,p_{1},q_{1},\dots, p_{I/2}, q_{I/2})$,  the above problem can be rewritten in canonical form:
\begin{align*}
    &\underset{y, p_{i}, q_{i}}{\text{maximize}} \quad f(\mathbf{x})=y \quad \quad \quad \quad \quad (***)\\
    &\text{subject to} \ \  g_{1}(\mathbf{x})=y-\frac{1}{I}\sum_{i=1}^{I/2}(6p_{i}-2q_{i})+1\leq 0 \\
    &\quad \quad \quad \quad \quad g_{2}(\mathbf{x})=y-\frac{1}{I}\sum_{i=1}^{I/2}(-2p_{i}+6q_{i})+1\leq 0 \\
    &\quad \quad \quad \quad \quad   s_{pi}(\mathbf{x})=p_{i}-\frac{\Gamma}{1+\Gamma}\leq 0 \quad i=1,\dots, I/2 \\
    &\quad \quad \quad \quad \quad  s_{qi}(\mathbf{x})=q_{i}-\frac{\Gamma}{1+\Gamma}\leq 0 \quad i=1,\dots, I/2 \\
    &\quad \quad \quad \quad \quad   t_{pi}(\mathbf{x})=-p_{i}+\frac{1}{1+\Gamma}\leq 0 \quad i=1,\dots, I/2 \\
    &\quad \quad \quad \quad \quad   t_{qi}(\mathbf{x})=-q_{i}+\frac{1}{1+\Gamma}\leq 0. \quad i=1,\dots, I/2
\end{align*}

The above problem along with its canonical form considers the maximum over the unmeasured confounders of the minimum over the outcomes of the expectations of $T_{1}$ and $T_{2}$ under Fisher's sharp null $H_{0}$. The design sensitivity would be the value $\Gamma=\widetilde{\Gamma}$ such that the optimal value $y^{*}$ exceeds the true expectation 1. We claim that the optimal solution is $p_{i}^{*}=q_{i}^{*}=\Gamma/(1+\Gamma)$ for each $i=1,\dots,I/2$, yielding $y^{*}=(\Gamma-1)/(1+\Gamma)$. To show this, we proceed by showing that this solution satisfies the Karush–Kuhn–Tucker (KKT) conditions. Since both the objective functions and the constraints are affine, the KKT conditions are sufficient for proving optimality of a solution. 

Associate KKT multipliers $\lambda_{1}$, $\lambda_{2}$, $\alpha_{pi}$, $\alpha_{qi}$, $\beta_{pi}$, $\beta_{qi}$ with the above constraints. Let $\lambda_{1}=\lambda_{2}=1/2$, $\alpha_{pi}=\alpha_{qi}=2/I$, $\beta_{pi}=\beta_{qi}=0$ for all $i$. We just need to check the following four parts of KKT conditions hold:

(1) Stationarity: partial of the objective function equals sum of partials of constraints times their KKT multipliers for each variable. 
\begin{align*}
    &y \ \ 1=\lambda_{1}+\lambda_{2}=1/2+1/2\\
    &p_{i}\ \ 0=-\lambda_{1}\cdot \frac{1}{I}\cdot 6-\lambda_{2}\cdot \frac{1}{I}\cdot (-2)+\alpha_{pi}-\beta_{pi}=-\frac{1}{2}\cdot \frac{1}{I}\cdot 6-\frac{1}{2}\cdot \frac{1}{I}\cdot (-2)+\frac{2}{I}-0\\
    &q_{i}\ \ 0=-\lambda_{1}\cdot \frac{1}{I}\cdot(-2)-\lambda_{2}\cdot \frac{1}{I}\cdot 6+\alpha_{qi}-\beta_{qi}=-\frac{1}{2}\cdot \frac{1}{I}\cdot (-2)-\frac{1}{2}\cdot \frac{1}{I}\cdot 6+\frac{2}{I}-0.
\end{align*}

(2) Primal feasibility: constraints must be satisfied. Let $\mathbf{x}^{*}=(y^{*},p_{1}^{*},q_{1}^{*},\dots, p_{I/2}^{*}, q_{I/2}^{*})=(\frac{\Gamma-1}{1+\Gamma}, \frac{\Gamma}{1+\Gamma},\frac{\Gamma}{1+\Gamma},\dots, \frac{\Gamma}{1+\Gamma},\frac{\Gamma}{1+\Gamma})$,
\begin{align*}
    g_{1}(\mathbf{x}^{*})=g_{2}(\mathbf{x}^{*})=\frac{\Gamma-1}{1+\Gamma}-\frac{1}{I}\cdot \frac{I}{2}\cdot \frac{4\Gamma}{1+\Gamma}+1=0,
\end{align*}
and it is clear that $s_{pi}(\mathbf{x}^{*})$, $s_{qi}(\mathbf{x}^{*}), t_{pi}(\mathbf{x}^{*}), t_{qi}(\mathbf{x}^{*})\leq 0$ are satisfied.

(3) Dual feasibility: KKT multipliers must be non-negative. This clearly holds based on our choices: $\lambda_{1}=\lambda_{2}=1/2$, $\alpha_{pi}=\alpha_{qi}=2/I$, $\beta_{pi}=\beta_{qi}=0$ for each $i$.

(4) Complementary slackness: for each constraint, either the constraint is binding or the KKT multiplier is zero. This clearly holds since the only constraints which are not binding at the solution are $t_{pi}(\mathbf{x}^{*})$ and $t_{qi}(\mathbf{x}^{*})$. For these, $\beta_{pi}=\beta_{qi}=0$. 

Hence, the KKT conditions $(1)-(4)$ are satisfied, which by KKT sufficiency implies optimality of the proposed solution $y^{*}=(\Gamma-1)/(1+\Gamma)<1$ for any $1\leq \Gamma<+\infty$. Hence, $\widetilde{\Gamma}_{1:2}=+\infty$ since for any finite $\Gamma$, the optimal value $y^{*}$ cannot exceed 1. Thus, the proof is complete.
\end{proof}

We study the simulated power to illustrate the example with $\widetilde{\Gamma}_{1:2}>\max
\{ \widetilde{\Gamma}_{1}, \widetilde{\Gamma}_{2} \}$ constructed in the proof of Theorem~\ref{th:gainsinds}. We consider the test statistics $T_{1}$ and $T_{2}$ defined in the proof of Theorem~\ref{th:gainsinds}. In this example constructed in the proof, we have $\widetilde{\Gamma}_{1}=\widetilde{\Gamma}_{2}=3$ and $\widetilde{\Gamma}_{1:2}=+\infty$. In Table~\ref{tab:checkgainsinds}, we report the simulated power of (i) using $T_{1}$ to test Fisher's sharp null of no treatment effect $H_{0}$, (ii) using $T_{2}$ to test $H_{0}$, and (iii) using the minimax procedure to combine $T_{1}$ and $T_{2}$ to test $H_{0}$, with level $\alpha=0.05$, $\Gamma=2,2.5,2.9,3.1,4,6$ and sample size $I=50, 100, 300$. The (one-sided) minimax procedure uses the critical value $\Phi^{-1}(1-\alpha/2)$ as in Fogarty and Small (2016) in simulations. All the numbers are based on 10,000 replications.

\begin{table}[H]
    \centering
    \caption{Simulated power of $T_{1}$, $T_{2}$ and the minimax procedure combining $T_{1}$ and $T_{2}$.}
    \label{tab:checkgainsinds}
    \smallskip
\smallskip
\smallskip
    \footnotesize
   \begin{tabular}{cccccccccc}
\hline
\multirow{3}{*}{} & \multicolumn{3}{c}{$T_{1}$} & \multicolumn{3}{c}{$T_{2}$}& \multicolumn{3}{c}{Minimax} \\
\cmidrule(r){2-4} \cmidrule(r){5-7} \cmidrule(r){8-10}
& $I=50$ &  $I=100$ & $I=300$ & $I=50$ & $I=100$ & $I=300$ & $I=50$ & $I=100$ & $I=300$ \\
\hline
$\Gamma=2.0$ & 0.24 & 0.47 & 0.92 & 0.25 & 0.45 & 0.93  & 1.00  & 1.00  & 1.00  \\
$\Gamma=2.5$ & 0.06 & 0.09 & 0.29 & 0.06 & 0.09  & 0.30 & 1.00 & 1.00  & 1.00  \\
$\Gamma=2.9$ & 0.02 & 0.02 & 0.03 & 0.01 & 0.02 & 0.03 & 0.48 & 1.00  & 1.00  \\
$\Gamma=3.1$ & 0.01 & 0.01 & 0.01 & 0.01 & 0.01 & 0.01 & 0.33 & 1.00  & 1.00  \\
$\Gamma=4.0$ & 0.00 & 0.00 & 0.00 & 0.00 & 0.00 & 0.00 & 0.01 & 1.00  & 1.00  \\
$\Gamma=6.0$ & 0.00 & 0.00 & 0.00 & 0.00 & 0.00 & 0.00 & 0.00 & 0.20  & 1.00  \\
\hline
\end{tabular}
\end{table}

\normalsize

From Table~\ref{tab:checkgainsinds}, we can see that as sample size $I \rightarrow \infty$, the simulated power of all the three tests increases to 1 for $\Gamma < 3$. For $\Gamma >3$, the simulated power of $T_{1}$ and $T_{2}$ is nearly zero, while the simulated power of using the minimax procedure to combine $T_{1}$ and $T_{2}$ in a sensitivity analysis still increases to 1 as $I$ increases. This holds true even for $\Gamma$ much larger than 3, which confirms that $\widetilde{\Gamma}_{1}=\widetilde{\Gamma}_{2}=3$ and $\widetilde{\Gamma}_{1:2}=+\infty$.

In the proof of Theorem~\ref{th:gainsinds}, to find a case in which applying the minimax procedure results in the substantial gains in design sensitivity, we constructed an example with a perfect negative correlation between the two treated-minus-control paired differences of the two score functions $f_{1}$ and $f_{2}$ of the two tests $T_{1}$ and $T_{2}$. In our example, the worst-case unmeasured confounder vector $\mathbf{u}$ for one test is actually the best-case $\mathbf{u}$ for the other test, and it is this conflict that yields the infinite design sensitivity. It is not necessary that the two test statistics be perfectly negatively correlated to attain an improved design sensitivity, and our procedure applied to independently distributed test statistics would also yield a design sensitivity that is larger than the max of the component design sensitivities. As the correlation gets closer to 1, in each pair the worst-case $\mathbf{u}$ for one test statistic is close to or exactly the worst-case $\mathbf{u}$ for the other test statistic with higher and higher probability, which means that the gap of the max-min inequality (\ref{inequality:minimax}) gets smaller. Recall that the gains in design sensitivity are resulted from the gap of the max-min inequality (\ref{inequality:minimax}), we therefore expect the magnitude of the gains in design sensitivity gets larger as the correlation between the two test statistics gets closer to $-1$, and gets smaller as the correlation gets closer to 1. This also explains why the gains in design sensitivity from applying our new adaptive testing procedure to combine the aberrant rank test and the Mantel-Haenszel test are small and hard to observe since these two component tests are highly positively correlated.

We study the simulated power to illustrate how the power of using the minimax procedure to combine $T_{1}$ and $T_{2}$ varies with the correlation between the two test statistics. In particular, we consider the following three settings of the joint distribution of two paired treated-minus-control differences $(D_{i}^{(1)}, D_{i}^{(2)})$:
\begin{itemize}
    \item Setting 1 (perfectly positively correlated): $\mathbb{P}(D_{i}^{(1)}=3, D_{i}^{(2)}=3)=\mathbb{P}(D_{i}^{(1)}=-1, D_{i}^{(2)}=-1)=1/2$
    \item Setting 2 (independent): $\mathbb{P}(D_{i}^{(1)}=3, D_{i}^{(2)}=3)=\mathbb{P}(D_{i}^{(1)}=3, D_{i}^{(2)}=-1)=\mathbb{P}(D_{i}^{(1)}=-1, D_{i}^{(2)}=3)=\mathbb{P}(D_{i}^{(1)}=-1, D_{i}^{(2)}=-1)=1/4$
    \item Setting 3 (perfectly negatively correlated): $\mathbb{P}(D_{i}^{(1)}=3, D_{i}^{(2)}=-1)=\mathbb{P}(D_{i}^{(1)}=-1, D_{i}^{(2)}=3)=1/2$.
\end{itemize}
Settings 1-3 have the same marginal distribution $\mathbb{P}(D_{i}^{(1)}=3)=\mathbb{P}(D_{i}^{(2)}=3)=\mathbb{P}(D_{i}^{(1)}=-1)=\mathbb{P}(D_{i}^{(2)}=-1)=1/2$. That is, the power of using $T_{1}$ and $T_{2}$ to test $H_{0}$ is the same in each setting and has been reported in Table~\ref{tab:checkgainsinds} if we still set $\alpha=0.05$, $\Gamma=2,2.5,2.9,3.1,4,6$ and sample size $I=50, 100, 300$. Note that Setting 3 is exactly the example constructed in the proof. We set the critical value $=\Phi^{-1}(1-\alpha/2)$ in each setting. All the numbers are based on 10,000 replications.

\begin{table}[ht]
    \centering
    \caption{Simulated power of the minimax procedure with the various correlations of the two component test statistics.}
    \label{tab:checkvariouscor}
    \smallskip
\smallskip
\smallskip
    \footnotesize
   \begin{tabular}{cccccccccc}
\hline
\multirow{3}{*}{} & \multicolumn{3}{c}{Setting 1} & \multicolumn{3}{c}{Setting 2}& \multicolumn{3}{c}{Setting 3} \\
\cmidrule(r){2-4} \cmidrule(r){5-7} \cmidrule(r){8-10}
& $I=50$ &  $I=100$ & $I=300$ & $I=50$ & $I=100$ & $I=300$ & $I=50$ & $I=100$ & $I=300$ \\
\hline
$\Gamma=2.0$ & 0.10 & 0.31 & 0.87 & 0.38 & 0.86 & 1.00  & 1.00  & 1.00  & 1.00  \\
$\Gamma=2.5$ & 0.01 & 0.05 & 0.17 & 0.11 & 0.46 & 0.99 & 1.00 & 1.00  & 1.00  \\
$\Gamma=2.9$ & 0.00 & 0.00 & 0.01 & 0.03 & 0.19 & 0.85 & 0.48 & 1.00  & 1.00  \\
$\Gamma=3.1$ & 0.00 & 0.00 & 0.00 & 0.02 & 0.11 & 0.67 & 0.33 & 1.00  & 1.00  \\
$\Gamma=4.0$ & 0.00 & 0.00 & 0.00 & 0.00 & 0.00 & 0.05 & 0.01 & 1.00  & 1.00  \\
$\Gamma=6.0$ & 0.00 & 0.00 & 0.00 & 0.00 & 0.00 & 0.00 & 0.00 & 0.20  & 1.00  \\
\hline
\end{tabular}
\end{table}
\normalsize

From Table~\ref{tab:checkvariouscor}, we can see that the power of using the minimax procedure to combine $T_{1}$ and $T_{2}$ increases as the correlation of the two test statistics decreases, which agrees with the theoretical insight that the magnitude of the gains in design sensitivity gets larger as the correlation between the two test statistics gets closer to $-1$. Also, Table~\ref{tab:checkvariouscor} suggests that substantial gains in the design sensitivities can be expected in both Setting 2 (independent case) and Setting 3 (negatively correlated case). Especially, for Setting 2 (independent case), the power still increases as the sample size increases when $\Gamma=4$, suggesting that the design sensitivity of the adaptive test (resulted from the minimax procedure) should be at least greater than or equal to 4 and is therefore evidently greater than the individual design sensitivity (=3). Therefore, substantial gains in design sensitivity resulted from the adaptive test can be evident with two negatively correlated, independent or weakly positively correlated component test statistics.

\begin{center}
{\large\bf Appendix B: Asymptotic Comparison With Correlation Within Strata}
\end{center}

In Section~\ref{sec:computeds}, to make the simulations easier and clearer, we assume that $R_{i1}, \dots, R_{im}$ are independent of each other within each stratum $i$. In practice, matching may introduce correlation within strata, so here we examine whether the pattern of asymptotic comparisons in Section~\ref{sec:computeds} via the design sensitivity still holds with correlated outcomes within strata or not.

\begin{table}[ht]
\centering \caption{Design sensitivities of the Mantel-Haenszel test and the aberrant rank test under Models 5-8 and matching with three controls with various parameters. The larger of the two design sensitivities of the two tests is in bold in each case.}
\label{tab:designsensitivity with correlation}
\smallskip
\smallskip
\smallskip
\small
\centering
\begin{tabular}{ c c c c } 
  \hline
  \multicolumn{4}{c}{Model 5: additive, multivariate normal}  \\
  Test statistic & $\beta=0.50$ & $\beta=0.75$ & $\beta=1.00$\\
  \hline
   Mantel-Haenszel & 3.47 & 6.45  & 11.99 \\
   Aberrant rank & \textbf{3.98} & \textbf{7.70}  & \textbf{14.74} \\
 \hline
 \multicolumn{4}{c}{Model 6: additive, multivariate Laplace}  \\
  Test statistic & $\beta=0.50$ & $\beta=0.75$ & $\beta=1.00$ \\
 \hline 
   Mantel-Haenszel &  \textbf{3.83} & \textbf{8.12}  & \textbf{17.51}\\
   Aberrant rank &  3.69 & 7.31  & 14.47 \\
 \hline
 \multicolumn{4}{c}{Model 7: multiplicative, multivariate normal}  \\
  Test statistic & $\delta=1.50$ & $\delta=1.75$ & $\delta=2.00$ \\
  \hline 
   Mantel-Haenszel &  2.30 & 2.91 & 3.46 \\
   Aberrant rank &  \textbf{3.62} &  \textbf{5.38} & \textbf{7.26} \\
 \hline
 \multicolumn{4}{c}{Model 8: multiplicative, multivariate Laplace}  \\
  Test statistic &  $\delta=1.50$ & $\delta=1.75$ & $\delta=2.00$ \\
 \hline
   Mantel-Haenszel & 2.39 & 3.13  & 3.82 \\
   Aberrant rank & \textbf{3.35} & \textbf{4.94}  & \textbf{6.66}
    \\
 \hline 
 \end{tabular}
\end{table}

As in Theorem~\ref{dsformula} and Section~\ref{sec:computeds}, without loss of generality, we still assume that $j=1$ receives treatment and others receive control for each $i$, and $(R_{i1}, \dots, R_{im})^{T}=(r_{Ti1}, r_{Ci2}, \dots, r_{Cim})^{T}$ is an i.i.d. realization from a multivariate continuous distribution. Following Section~\ref{sec:computeds}, we still assume that $r_{Tij}=g(r_{Cij})$ for some deterministic function $g$. Instead of assuming $r_{Ci1},\dots, r_{Cim}$ are independent of each other within each stratum $i$ as in Section~\ref{sec:computeds}, here we allow $r_{Ci1},\dots, r_{Cim}$ to be correlated with each other. Parallel with Models 1-4 in Section~\ref{sec:computeds}, we consider the following four models with correlated outcomes within strata:
 \begin{itemize}
	\item Model 5 (additive treatment effects, multivariate normal distribution with correlation): $r_{Tij}=r_{Cij}+\beta$, and $(r_{Ci1}, \dots, r_{Cim})$ follows a multivariate normal distribution with the mean vector being $(0,\dots,0)$ and the covariance matrix being $\Sigma_{m\times m}$ where $\Sigma_{ii}=1$ for $i=1,\dots,m$ and $\Sigma_{ij}=0.5$ for any $i\neq j$.
	\item Model 6 (additive treatment effects, multivariate Laplace distribution with correlation): $r_{Tij}=r_{Cij}+\beta$, and $(r_{Ci1}, \dots, r_{Cim})$ follows a multivariate Laplace distribution with the mean vector being $(0,\dots,0)$ and the covariance matrix being $\Sigma_{m\times m}$ where $\Sigma_{ii}=1$ for $i=1,\dots,m$ and $\Sigma_{ij}=0.5$ for any $i\neq j$.
    \item Model 7 (multiplicative treatment effects, multivariate normal distribution with correlation): $r_{Tij}=\delta\cdot  r_{Cij}$, and $(r_{Ci1}, \dots, r_{Cim})$ follows a multivariate normal distribution with the mean vector being $(0,\dots,0)$ and the covariance matrix being $\Sigma_{m\times m}$ where $\Sigma_{ii}=1$ for $i=1,\dots,m$ and $\Sigma_{ij}=0.5$ for any $i\neq j$.
	\item Model 8 (multiplicative treatment effects, multivariate Laplace distribution with correlation): $r_{Tij}=\delta\cdot r_{Cij}$, and $(r_{Ci1}, \dots, r_{Cim})$ follows a multivariate Laplace distribution with the mean vector being $(0,\dots,0)$ and the covariance matrix being $\Sigma_{m\times m}$ where $\Sigma_{ii}=1$ for $i=1,\dots,m$ and $\Sigma_{ij}=0.5$ for any $i\neq j$.
\end{itemize}
 For Models 5-8, we still set the aberrant response threshold to be $c=1$. Table~\ref{tab:designsensitivity with correlation} reports the design sensitivities of the Mantel-Haenszel test and the aberrant rank test under Models 5-8 with $m=4$ (i.e., matching with three controls) and various $\beta$ and $\delta$. As described in Section~\ref{sec:computeds}, calculation of the design sensitivity can be done via Monte-Carlo simulations.

 The pattern of Table~\ref{tab:designsensitivity with correlation} agrees with that of Table~\ref{tab:designsensitivity}. First, the choice of the test statistic still has a huge influence on the design sensitivities in the presence of correlation within strata; see Models 7 and 8 in Table~\ref{tab:designsensitivity with correlation}. Second, with correlated outcomes within strata, whether or not the aberrant rank test outperforms the Mantel-Haenszel test still depends on the unknown data generating process. From Table~\ref{tab:designsensitivity with correlation}, we can see that under Models 5, 7 and 8, the aberrant rank test outperforms the Mantel-Haenszel test in terms of design sensitivities. While under Model 6, the design sensitivities of the Mantel Haenszel test are larger than those of the aberrant rank test. Note that Models 5-8 are parallel with Models 1-4: both the marginal distributions of $r_{Tij}$ and $r_{Cij}$ are correspondingly equal for Models 1-4 versus Models 5-8. Therefore, the insights in Section~\ref{sec:computeds} on when the aberrant rank test should be preferred over the Mantel-Haenszel test and other times the Mantel-Haenszel test should be preferred can still shed light on the simulation results in Table~\ref{tab:designsensitivity with correlation}.

\clearpage
\begin{center}
    {\large\bf Appendix C: More Details on the Regularity Assumptions}
\end{center}

We give a sufficient condition under which $(T_{1}, T_{2})$ considered in Section~\ref{sec:adaptive} is asymptotically bivariate normal in the sense that in large sample size, the distribution function of $\Big(\frac{T_{1}-\mu_{1,\mathbf{u}}}{\sigma_{1, \mathbf{u}}},\ \frac{T_{2}-\mu_{2,\mathbf{u}}}{\sigma_{2, \mathbf{u}}} \Big)$ can be approximated by that of $\mathcal{N}\left(\left(\begin{array}{c}
0\\
0
\end{array}\right),\left(\begin{array}{cc}
1 & \mathbf{\rho}_{\mathbf{u}} \\
\mathbf{\rho}_{\mathbf{u}} & 1 
\end{array}\right)\right)$, where $\mathbf{\rho}_{\mathbf{u}}=\mathbb{E}\Big(\frac{ T_{1}-\mu_{1, \mathbf{u}}}{\sigma_{1, \mathbf{u}}}\cdot \frac{ T_{2}-\mu_{2, \mathbf{u}}}{\sigma_{2, \mathbf{u}}}\Big \vert \mathcal{F},\mathcal{Z} \Big)$. 

Let $\widetilde{T}_{k,i}=(\sum_{j=1}^{n_{i}}Z_{ij}q_{ijk}-\sum_{j=1}^{n_{i}}p_{ij}q_{ijk})/\sigma_{k,\mathbf{u}}$ for $i=1,\dots,I$ and $k\in\{1,2\}$. Therefore, we have $\frac{T_{1}-\mu_{1,\mathbf{u}}}{\sigma_{1, \mathbf{u}}}=\sum_{i=1}^{I}\widetilde{T}_{1,i}$ and $\frac{T_{2}-\mu_{2,\mathbf{u}}}{\sigma_{2, \mathbf{u}}}=\sum_{i=1}^{I}\widetilde{T}_{2,i}$. Let $T_{\text{joint},i}=(\widetilde{T}_{1,i}, \widetilde{T}_{2,i})^{T}$ for $i=1,\dots,I$. Proposition~\ref{prop: regularity assumption} gives a sufficient condition under which the desired asymptotic bivariate normality holds.

\begin{Proposition}\label{prop: regularity assumption}

We let $\Sigma_{\mathbf{u}}=\left(\begin{array}{cc}
1 & \mathbf{\rho}_{\mathbf{u}} \\
\mathbf{\rho}_{\mathbf{u}} & 1 
\end{array}\right)$. Suppose that the following three assumptions hold: (i) Treatment assignments are independent across matched strata; (ii) As $I\rightarrow \infty$, $\Sigma_{\mathbf{u}}$ has a positive definite limit $\widetilde{\Sigma}$; (iii) For any fixed nonzero vector $\lambda=(\lambda_{1}, \lambda_{2})^{T}$, there exists some $\delta>0$, such that Lyapunov's condition
\begin{equation*}
    \lim_{I\rightarrow \infty}\frac{1}{(\lambda^{T}\Sigma_{\mathbf{u}}\lambda)^{1+\delta/2}}\sum_{i=1}^{I}\mathbb{E}\big\{|\lambda^{T}T_{\text{joint},i}|^{2+\delta} \big\}=0
\end{equation*}
is satisfied. Then we have as $I\rightarrow \infty$, $(T_{1}, T_{2})$ is asymptotically bivariate normal in the sense that 
\begin{equation*}
\Sigma_{\mathbf{u}}^{-1/2}\Big(\frac{T_{1}-\mu_{1,\mathbf{u}}}{\sigma_{1, \mathbf{u}}},\ \frac{T_{2}-\mu_{2,\mathbf{u}}}{\sigma_{2, \mathbf{u}}} \Big)^{T} \xrightarrow{\mathcal{L}} \mathcal{N}(\mathbf{0}, I_{2\times 2}).
\end{equation*}
\end{Proposition}

\begin{proof}
We first show that $\widetilde{\Sigma}^{-1/2}\Big(\frac{T_{1}-\mu_{1,\mathbf{u}}}{\sigma_{1, \mathbf{u}}},\ \frac{T_{2}-\mu_{2,\mathbf{u}}}{\sigma_{2, \mathbf{u}}} \Big)^{T}\xrightarrow{\mathcal{L}} \mathcal{N}(\mathbf{0}, I_{2\times 2})$. Through an application of the Cram\'er-Wold device (Billingsley, 1995; Theorem 29.4), to ensure that $\widetilde{\Sigma}^{-1/2}\Big(\frac{T_{1}-\mu_{1,\mathbf{u}}}{\sigma_{1, \mathbf{u}}},\ \frac{T_{2}-\mu_{2,\mathbf{u}}}{\sigma_{2, \mathbf{u}}} \Big)^{T} \xrightarrow{\mathcal{L}} \mathcal{N}(\mathbf{0}, I_{2\times 2})$,  we just need to ensure that for any nonzero vector $\widetilde{\lambda}=(\widetilde{\lambda}_{1}, \widetilde{\lambda}_{2})^{T}$, the following standardized deviate is asymptotically normal:
\begin{equation}\label{asymptotic normal of standard deviation}
    \frac{\widetilde{\lambda}^{T} \widetilde{\Sigma}^{-1/2}\Big(\frac{T_{1}-\mu_{1,\mathbf{u}}}{\sigma_{1, \mathbf{u}}},\ \frac{T_{2}-\mu_{2,\mathbf{u}}}{\sigma_{2, \mathbf{u}}} \Big)^{T} }{\sqrt{\widetilde{\lambda}^{T}\widetilde{\lambda}}}\xrightarrow{\mathcal{L}}\mathcal{N}(0,1).
\end{equation}
When treatment assignments are independent across matched strata (condition (i)), for each $I=1,2,\dots$, the sequence 
\begin{equation*}
    A_{I}=\{\widetilde{\lambda}^{T}\widetilde{\Sigma}^{-1/2}T_{\text{joint},1},\dots,\widetilde{\lambda}^{T}\widetilde{\Sigma}^{-1/2} T_{\text{joint},I}\}
\end{equation*}
is a sequence of independent random variables, and the collection $\{A_{1},A_{2},\dots,\}$ is a triangular array of random variables. Set $\lambda^{T}=\widetilde{\lambda}^{T} \widetilde{\Sigma}^{-1/2}$ in condition (iii) and apply Lyapunov central limit theorem (Billingsley, 1995; Theorem 27.3) to $\{A_{1},A_{2},\dots,\}$, we have
\begin{equation*}
     \frac{\widetilde{\lambda}^{T} \widetilde{\Sigma}^{-1/2}\Big(\frac{T_{1}-\mu_{1,\mathbf{u}}}{\sigma_{1, \mathbf{u}}},\ \frac{T_{2}-\mu_{2,\mathbf{u}}}{\sigma_{2, \mathbf{u}}} \Big)^{T} }{\sqrt{\widetilde{\lambda}^{T}\widetilde{\Sigma}^{-1/2}\Sigma_{\mathbf{u}}\widetilde{\Sigma}^{-1/2} \widetilde{\lambda}}}\xrightarrow{\mathcal{L}}\mathcal{N}(0,1).
\end{equation*}
Then (\ref{asymptotic normal of standard deviation}) follows immediately from
\begin{align*}
    \frac{\widetilde{\lambda}^{T} \widetilde{\Sigma}^{-1/2}\Big(\frac{T_{1}-\mu_{1,\mathbf{u}}}{\sigma_{1, \mathbf{u}}},\ \frac{T_{2}-\mu_{2,\mathbf{u}}}{\sigma_{2, \mathbf{u}}} \Big)^{T} }{\sqrt{\widetilde{\lambda}^{T}\widetilde{\lambda}}}&=\frac{\widetilde{\lambda}^{T} \widetilde{\Sigma}^{-1/2}\Big(\frac{T_{1}-\mu_{1,\mathbf{u}}}{\sigma_{1, \mathbf{u}}},\ \frac{T_{2}-\mu_{2,\mathbf{u}}}{\sigma_{2, \mathbf{u}}} \Big)^{T} }{\sqrt{\widetilde{\lambda}^{T}\widetilde{\Sigma}^{-1/2}\Sigma_{\mathbf{u}}\widetilde{\Sigma}^{-1/2} \widetilde{\lambda}}}\frac{\sqrt{\widetilde{\lambda}^{T}\widetilde{\Sigma}^{-1/2}\Sigma_{\mathbf{u}}\widetilde{\Sigma}^{-1/2}\widetilde{\lambda}}}{\sqrt{\widetilde{\lambda}^{T}\widetilde{\lambda}}}\\
    &\xrightarrow{\mathcal{L}} \mathcal{N}(0,1)\cdot 1 \quad \text{(by condition (ii) and Slutsky's theorem)}\\
    &\sim \mathcal{N}(0,1).
\end{align*}
So we have shown that $\widetilde{\Sigma}^{-1/2}\Big(\frac{T_{1}-\mu_{1,\mathbf{u}}}{\sigma_{1, \mathbf{u}}},\ \frac{T_{2}-\mu_{2,\mathbf{u}}}{\sigma_{2, \mathbf{u}}} \Big)^{T}\xrightarrow{\mathcal{L}} \mathcal{N}(\mathbf{0}, I_{2\times 2})$. Then note that
\begin{align*}
   &\quad \Sigma_{\mathbf{u}}^{-1/2}\Big(\frac{T_{1}-\mu_{1,\mathbf{u}}}{\sigma_{1, \mathbf{u}}},\ \frac{T_{2}-\mu_{2,\mathbf{u}}}{\sigma_{2, \mathbf{u}}} \Big)^{T}\\
   &= \Sigma_{\mathbf{u}}^{-1/2}\widetilde{\Sigma}^{1/2}\widetilde{\Sigma}^{-1/2} \Big(\frac{T_{1}-\mu_{1,\mathbf{u}}}{\sigma_{1, \mathbf{u}}},\ \frac{T_{2}-\mu_{2,\mathbf{u}}}{\sigma_{2, \mathbf{u}}} \Big)^{T}\\
   &\xrightarrow{\mathcal{L}} I_{2\times 2} \cdot \mathcal{N}(\mathbf{0}, I_{2\times 2})\\
   &\sim \mathcal{N}(\mathbf{0}, I_{2\times 2}).
\end{align*}
That is, the distribution function of $\Big(\frac{T_{1}-\mu_{1,\mathbf{u}}}{\sigma_{1, \mathbf{u}}},\ \frac{T_{2}-\mu_{2,\mathbf{u}}}{\sigma_{2, \mathbf{u}}} \Big)^{T}$ can be approximated by that of $\mathcal{N}(\mathbf{0}, \Sigma_{\mathbf{u}})$. So we are done.

\end{proof}

\begin{center}
{\large\bf Appendix D: More Details on the Two-Stage Programming Method}
\end{center}

In this section, we give a detailed derivation of the second stage of the two-stage programming method. Suppose that the prespecified level $\alpha<1/2$. Once getting the rejection threshold $Q_{\mathbf{\rho}^{*}_{\Gamma},\alpha}$ through solving $(*)$, we apply the method developed in Fogarty and Small (2016) to formulate the problem of checking if the inequality 
\begin{equation}\label{ineq: second stage of two-stage method }
    \min_{\mathbf{u} \in \mathcal{U}}  \max_{k\in \{1,2\}} \frac{t_{k}-\mu_{k,\mathbf{u}}}{\sigma_{k, \mathbf{u}}}\geq Q_{\mathbf{\rho}^{*}_{\Gamma},\alpha}, \quad \text{($Q_{\mathbf{\rho}^{*}_{\Gamma},\alpha}>0$ when $\alpha<1/2$)}
\end{equation} 
would hold at a given $\Gamma$ into checking if the optimal value of $(**)$ is greater than or equal to zero or not.

For each fixed $\Gamma$, to check if (\ref{ineq: second stage of two-stage method }) holds or not, we just need to check if 
\begin{equation*}
    \min_{\mathbf{u} \in \mathcal{U}}  
    \max_{k\in \{1,2\}}\Big(t_{k}-\mu_{k,\mathbf{u}}- Q_{\mathbf{\rho}^{*}_{\Gamma}, \alpha}\sigma_{k, \mathbf{u}}\Big) \geq 0,
\end{equation*} 
which can be transformed into solving the following optimization problem and checking if its optimal value $\geq 0$ or not by introducing an auxiliary variable $y$:
\begin{equation*}
     \begin{split}
        \underset{y, u_{ij}}{\text{minimize}} \quad &y \\
         \text{subject to}\quad & y\geq t_k-\mu_{k,\mathbf{u}}- Q_{\mathbf{\rho}^{*}_{\Gamma},\alpha} \sigma_{k,\mathbf{u}} \quad \forall k \in \{0,1\}\\
        &0\leq u_{ij}\leq 1. \quad\forall i,j
     \end{split}
 \end{equation*}
 Note that the above constraints force $y$ to be larger than $t_k-\mu_{k,\mathbf{u}}- Q_{\mathbf{\rho}^{*}_{\Gamma},\alpha} \sigma_{k,\mathbf{u}}$ for both $k=1$ and $k=2$, and drive us to search for the feasible value of $\mathbf{u}=(u_{11},\dots,u_{In_{I}})\in \mathcal{U}$ that allows for $y$ to be as small as possible. This is a routine way of solving minimax problems (Charalambous and Conn, 1978). Therefore, the above optimization problem indeed seeks to find $\min_{\mathbf{u} \in \mathcal{U}}  
    \max_{k\in \{1,2\}} (t_{k}-\mu_{k,\mathbf{u}}- Q_{\mathbf{\rho}^{*}_{\Gamma}, \alpha}\sigma_{k, \mathbf{u}})$.
    
Recall that $w_{ij}=\exp(\gamma u_{ij})$ and $p_{ij}=w_{ij}/\sum_{j^{\prime}=1}^{n_{i}}w_{ij^{\prime}}$. Then the above optimization problem can be written as 
\begin{equation*}
     \begin{split}
        \underset{y, w_{ij}}{\text{minimize}} \quad &y \\
         \text{subject to}\quad & y\geq t_k-\mu_{k,\mathbf{u}}- Q_{\mathbf{\rho}^{*}_{\Gamma},\alpha} \sigma_{k,\mathbf{u}} \quad \forall k \in \{0,1\}\\
        &1\leq w_{ij}\leq \Gamma, \quad\forall i,j
     \end{split}
 \end{equation*}
 where $\mu_{k, \mathbf{u}}=\sum_{i=1}^{I}\sum_{j=1}^{n_{i}}p_{ij}q_{ijk}=\sum_{i=1}^{I}\sum_{j=1}^{n_{i}}\frac{w_{ij}}{\sum_{j^{\prime}=1}^{n_{i}}w_{ij^{\prime}}}\cdot q_{ijk}$ and 
 \begin{align*}
     \sigma_{k, \mathbf{u}}&=\sqrt{\sum_{i=1}^{I}\sum_{j=1}^{n_{i}}p_{ij}q_{ijk}^{2}-\sum_{i=1}^{I}\big(\sum_{j=1}^{n_{i}}p_{ij}q_{ijk}\big)^{2}} \\ &=\sqrt{\sum_{i=1}^{I}\sum_{j=1}^{n_{i}}\frac{w_{ij}}{\sum_{j^{\prime}=1}^{n_{i}}w_{ij^{\prime}}}\cdot q_{ijk}^{2}-\sum_{i=1}^{I}\big(\sum_{j=1}^{n_{i}}\frac{w_{ij}}{\sum_{j^{\prime}=1}^{n_{i}}w_{ij^{\prime}}}\cdot q_{ijk}\big)^{2}}.
 \end{align*}
There are two terms that make the constraints of the above optimization problem complicated: the square root term $\sigma_{k, \mathbf{u}}$ and the linear-fractional term $w_{ij}/\sum_{j^{\prime}=1}^{n_{i}}w_{ij^{\prime}}$. 

We first consider how to get rid of the square root term $\sigma_{k, \mathbf{u}}$. Recall that we are just concerned about if the optimal value of the above optimization problem $\geq 0$ or not. We introduce a prespecified sufficiently large constant `$M$' and two axillary variables $b_{1}$ and $b_{2}$, and then instead consider the following adjusted optimization problem:
\begin{equation*}
     \begin{split}
        \underset{y, w_{ij}, b_{k}}{\text{minimize}}  \quad &y \\
         \text{subject to}\quad & y\geq (t_k-\mu_{k,\mathbf{u}})^2- Q^{2}_{\mathbf{\rho}^{*}_{\Gamma},\alpha} \sigma^2_{k,\mathbf{u}}-Mb_k \quad \forall k \in \{0,1\}\\
        &1\leq w_{ij}\leq \Gamma \quad\forall i,j\\
        &b_k\in \{0,1\} \quad \forall k \in \{0,1\}\\
        &-Mb_k\leq t_k-\mu_{k,\mathbf{u}}\leq M(1-b_k). \quad\forall k \in \{0,1\}
     \end{split}
 \end{equation*}
Note that when $M$ is sufficiently large, for all $k \in \{0,1\}$ and $1\leq w_{ij}\leq \Gamma$, we have $-Mb_k\leq t_k-\mu_{k,\mathbf{u}}\leq M(1-b_k)$ for either $b_{k}=0$ or $b_{k}=1$. When $b_{k}=0$, we have the constraints $0 \leq t_k-\mu_{k,\mathbf{u}}\leq M$ and $y\geq (t_k-\mu_{k,\mathbf{u}})^2- Q^{2}_{\mathbf{\rho}^{*}_{\Gamma},\alpha} \sigma^2_{k,\mathbf{u}}$. When $M$ is sufficiently large, for any $\mathbf{u}\in [0,1]^{N}$ (or equivalently, for any $(w_{11}, \dots, w_{In_{I}})\in [1,\Gamma]^{N}$) such that $0 \leq t_k-\mu_{k,\mathbf{u}}\leq M $, we have $(t_k-\mu_{k,\mathbf{u}})^2- Q^{2}_{\mathbf{\rho}^{*}_{\Gamma},\alpha} \sigma^2_{k,\mathbf{u}}\geq 0$ if and only if $t_k-\mu_{k,\mathbf{u}}- Q_{\mathbf{\rho}^{*}_{\Gamma},\alpha} \sigma_{k,\mathbf{u}}\geq 0$. When $b_{k}=1$, we have the constraints $-M\leq t_k-\mu_{k,\mathbf{u}}\leq 0$ and $y\geq (t_k-\mu_{k,\mathbf{u}})^2- Q^{2}_{\mathbf{\rho}^{*}_{\Gamma},\alpha} \sigma^2_{k,\mathbf{u}}-M$. When $M$ is sufficiently large, for any $\mathbf{u}_{1}, \mathbf{u}_{2}\in [0,1]^{N}$, we have $(t_k-\mu_{k,\mathbf{u}_{2}})^2- Q^{2}_{\mathbf{\rho}^{*}_{\Gamma},\alpha} \sigma^2_{k,\mathbf{u}_{2}}>(t_k-\mu_{k,\mathbf{u}_{1}})^2- Q^{2}_{\mathbf{\rho}^{*}_{\Gamma},\alpha} \sigma^2_{k,\mathbf{u}_{1}}-M$ and $(t_k-\mu_{k,\mathbf{u}_{1}})^2- Q^{2}_{\mathbf{\rho}^{*}_{\Gamma},\alpha} \sigma^2_{k,\mathbf{u}_{1}}-M<0$. Therefore, the above `$M$' constraint imposes a directional error to ensure that we will not falsely reject the null if evidence pointed in the opposite direction of the alternative, i.e., the cases in which $\min \limits_{\mathbf{u} \in \mathcal{U}} \ \max\limits_{k\in \{1,2\}}\ (t_{k}-\mu_{k,\mathbf{u}})^{2}/\sigma^{2}_{k, \mathbf{u}}\geq Q^{2}_{\mathbf{\rho}^{*}_{\Gamma},\alpha}$ while $\min \limits_{\mathbf{u} \in \mathcal{U}} \ \max\limits_{k\in \{1,2\}}\ (t_{k}-\mu_{k,\mathbf{u}})/\sigma_{k, \mathbf{u}}< Q_{\mathbf{\rho}^{*}_{\Gamma},\alpha}$.

We then consider how to get rid of the linear-fractional term $w_{ij}/\sum_{j^{\prime}=1}^{n_{i}}w_{ij^{\prime}}$. A routine way of transforming a linear-fractional term into linear terms is through applying the Charnes-Cooper transformation (Charnes and Cooper, 1962):
\begin{equation*}
    p_{ij}=\frac{w_{ij}}{\sum_{j^{\prime}=1}^{n_{i}}w_{ij^{\prime}}}=\frac{\exp(\gamma u_{ij})}{\sum_{j^{\prime}=1}^{n_{i}}\exp(\gamma u_{ij^{\prime}})}, \quad s_{i}=\frac{1}{\sum_{j^{\prime}=1}^{n_{i}}w_{ij^{\prime}}}=\frac{1}{\sum_{j^{\prime}=1}^{n_{i}}\exp(\gamma u_{ij^{\prime}})}.
\end{equation*}
Then the above optimization problem can be transformed into the following quadratically constrained linear program as stated in section~\ref{sec:adaptive}:
\begin{equation*}
     \begin{split}
        \underset{y, p_{ij}, s_i, b_k}{\text{minimize}} \quad &y  \quad \quad (**)\\
         \text{subject to}\quad & y\geq (t_k-\mu_{k,\mathbf{u}})^2- Q^{2}_{\mathbf{\rho}^{*}_{\Gamma},\alpha} \sigma^2_{k,\mathbf{u}}-Mb_k \quad \forall k \in \{0,1\}\\
        &\sum_{j=1}^{n_{i}} p_{ij}=1\quad\forall i\\
        &s_i\leq p_{ij}\leq \Gamma s_i\quad\forall i,j\\
        &p_{ij}\geq0 \quad\forall i,j\\
        &b_k\in \{0,1\} \quad \forall k \in \{0,1\}\\
        &-Mb_k\leq t_k-\mu_{k,\mathbf{u}}\leq M(1-b_k). \quad\forall k \in \{0,1\}
     \end{split}
 \end{equation*}
Therefore, to judge if (\ref{ineq: second stage of two-stage method }) holds or not, we just need to check whether the optimal value $y^{*}_{\Gamma}$ of $(**)$ satisfies $y^{*}_{\Gamma} \geq 0$. If it is, we reject the null; otherwise, we fail to reject. 

\begin{center}
{\large\bf Appendix E: Adaptive Approach in Full Matching Case}
\end{center}

In this section, we discuss how the adaptive approach developed in the main text can be easily adjusted to allow for full matching case, i.e., the constraint that $\sum_{j=1}^{n_{i}}Z_{ij}\in \{1,n_{i}-1\}$ for all $i=1,\dots,I$ (i.e., either one treated individual and one or more controls, or one control and one or more treated individuals, within each stratum). Let $I=I_{1}+I_{2}$. We rearrange the index $i$ of each stratum such that $\sum_{j=1}^{n_{i}}Z_{ij}=1$ for $i=1,\dots,I_{1}$ and $\sum_{j=1}^{n_{i}}Z_{ij}=n_{i}-1$ for $i=I_{1}+1,\dots, I_{1}+I_{2}$. Let $\widetilde{\mathcal{Z}}$ be the collection of the treatment assignment indicators $\mathbf{Z}=(Z_{11},\dots,Z_{In_{I}})$ such that $\mathbf{Z}\in \widetilde{\mathcal{Z}}$ if and only if $\sum_{j=1}^{n_{i}}Z_{ij}=1$ for all $i=1,\dots,I_{1}$ and $\sum_{j=1}^{n_{i}}Z_{ij}=n_{i}-1$ for all $i=I_{1}+1,\dots,I_{1}+I_{2}$. Let $p_{ij}=\mathbb{P}(Z_{ij}=1\mid \mathcal{F}, \widetilde{\mathcal{Z}})$, $\widetilde{p}_{ij}=\mathbb{P}(Z_{ij}=0\mid \mathcal{F}, \widetilde{\mathcal{Z}})=1-p_{ij}$ and $T_{k,i}=\sum_{j=1}^{n_{i}}Z_{ij}q_{ijk}$ for $i=1,\dots,I, j=1,\dots, n_{i}, k=1,2$. Therefore, we have $T_{1}=\sum_{i=1}^{I}T_{1,i}$ and $T_{2}=\sum_{i=1}^{I}T_{2,i}$. In this case, we have
\begin{equation*}
\left\{
\begin{array}{ll}
      p_{ij}=\exp(\gamma u_{ij})/\sum_{j^{\prime}=1}^{n_{i}}\exp(\gamma u_{ij^{\prime}}) &\quad \text{for\ $i=1,\dots,I_{1}, j=1,\dots,n_{i}$}, \\
      \widetilde{p}_{ij}=\exp(-\gamma u_{ij})/\sum_{j^{\prime}=1}^{n_{i}}\exp(-\gamma u_{ij^{\prime}})   & \quad \text{for \ $i=I_{1}+1,\dots, I_{1}+I_{2}, j=1,\dots,n_{i}$}. 
\end{array}
\right. 
\end{equation*}
Then we have
\begin{align}\label{correlation formula in full matching case}
    &\widetilde{\mu}_{k, \mathbf{u}}=\mathbb{E}_{\Gamma, \mathbf{u}}(T_{k}\mid \mathcal{F}, \widetilde{\mathcal{Z}})=\sum_{i=1}^{I_{1}}\sum_{j=1}^{n_{i}}p_{ij}q_{ijk}+\sum_{i=I_{1}+1}^{I_{1}+I_{2}}\sum_{j=1}^{n_{i}}q_{ijk}-\sum_{i=I_{1}+1}^{I_{1}+I_{2}}\sum_{j=1}^{n_{i}}\widetilde{p}_{ij}q_{ijk},\nonumber\\
    &\widetilde{\sigma}_{k, \mathbf{u}}^{2}=Var_{\Gamma, \mathbf{u}}(T_{k}\mid \mathcal{F}, \widetilde{\mathcal{Z}})=\sum_{i=1}^{I_{1}}\sum_{j=1}^{n_{i}}p_{ij}q_{ijk}^{2}-\sum_{i=1}^{I_{1}}(\sum_{j=1}^{n_{i}}p_{ij}q_{ijk})^{2}+\sum_{i=I_{1}+1}^{I_{1}+I_{2}}\sum_{j=1}^{n_{i}}\widetilde{p}_{ij}q_{ijk}^{2}-\sum_{i=I_{1}+1}^{I_{1}+I_{2}}(\sum_{j=1}^{n_{i}}\widetilde{p}_{ij}q_{ijk})^{2},\nonumber \\
    & Cov_{\Gamma, \mathbf{u}}(T_{1}, T_{2}\mid \mathcal{F}, \widetilde{\mathcal{Z}})=\sum_{i=1}^{I_{1}}\sum_{j=1}^{n_{i}}p_{ij}q_{ij1}q_{ij2}-\sum_{i=1}^{I_{1}}(\sum_{j=1}^{n_{i}}p_{ij}q_{ij1})(\sum_{j=1}^{n_{i}}p_{ij}q_{ij2})\nonumber \\
    &\qquad \qquad \qquad \qquad \qquad \qquad +\sum_{i=I_{1}+1}^{I_{1}+I_{2}}\sum_{j=1}^{n_{i}}\widetilde{p}_{ij}q_{ij1}q_{ij2}-\sum_{i=I_{1}+1}^{I_{1}+I_{2}}(\sum_{j=1}^{n_{i}}\widetilde{p}_{ij}q_{ij1})(\sum_{j=1}^{n_{i}}\widetilde{p}_{ij}q_{ij2}), \nonumber \\
    & \widetilde{\mathbf{\rho}}_{\mathbf{u}}=\mathbb{E}\Big(\frac{ T_{1}-\widetilde{\mu}_{1, \mathbf{u}}}{\widetilde{\sigma}_{1, \mathbf{u}}}\cdot \frac{ T_{2}-\widetilde{\mu}_{2, \mathbf{u}}}{\widetilde{\sigma}_{2, \mathbf{u}}}\Big \vert \mathcal{F},\widetilde{\mathcal{Z}} \Big)=\frac{Cov_{\Gamma, \mathbf{u}}(T_{1}, T_{2}\mid \mathcal{F}, \widetilde{\mathcal{Z}})}{\widetilde{\sigma}_{1, \mathbf{u}}\widetilde{\sigma}_{2, \mathbf{u}}}.
\end{align}
Similar to the main text, we let $w_{ij}=\exp(\gamma u_{ij})$. We further let $\widetilde{w}_{ij}=\exp(-\gamma u_{ij})=w_{ij}^{-1}$. In the first stage, to find the worst-case correlation $\min_{\mathbf{u}\in \mathcal{U}} \widetilde{\mathbf{\rho}}_{\mathbf{u}}$, we just need to solve the following optimization problem which just slightly adjusts the box constraints of the optimization problem $(*)$:
\begin{align*}
    &\underset{w_{ij}, \widetilde{w}_{ij}}{\text{minimize}}\quad \widetilde{\mathbf{\rho}}_{\mathbf{u}} \quad \quad  (\diamond) \\
    &\text{subject to}\ \ 1 \leq w_{ij} \leq \Gamma \quad  i=1,\dots,I_{1},j=1,\dots,n_{i} \\
    &\qquad \qquad \quad \Gamma^{-1} \leq \widetilde{w}_{ij} \leq 1, \quad  i=I_{1}+1,\dots,I_{1}+I_{2}, j=1,\dots,n_{i} 
\end{align*}
where $\rho_{\mathbf{u}}$ is as defined in (\ref{correlation formula in full matching case}) with 
\begin{equation*}
\left\{
\begin{array}{ll}
      p_{ij}=w_{ij}/\sum_{j^{\prime}=1}^{n_{i}}w_{ij} &\quad \text{for\ $i=1,\dots,I_{1},j=1,\dots,n_{i}$}, \\
      \widetilde{p}_{ij}=\widetilde{w}_{ij}/\sum_{j^{\prime}=1}^{n_{i}}\widetilde{w}_{ij}  & \quad \text{for\ $i=I_{1}+1,\dots, I_{1}+I_{2}, j=1,\dots, n_{i}$}. 
\end{array}
\right. 
\end{equation*}
Similar to $(*)$ in the main text, the optimization problem $(\diamond)$ can be solved approximately in a reasonable amount of time by the L-BFGS-B algorithm (Byrd et al., 1995; Zhu et al., 1997). Denote the optimal value of $(\diamond)$ with sensitivity parameter $\Gamma$ as $\mathbf{\rho}^{\diamond}_{\Gamma}$. Then the corresponding worst-case quantile $\max_{\mathbf{u}\in \mathcal{U}}Q_{\widetilde{\rho}_{\mathbf{u}}, \alpha}$ equals $Q_{\mathbf{\rho}^{\diamond}_{\Gamma},\alpha}$ by Slepian's lemma. 

As discussed in the main text, to determine if we should reject the null with level $\alpha$ and a given $\Gamma$ in a sensitivity analysis, we then need to check if $\min_{\mathbf{u} \in \mathcal{U}}  \max_{k\in \{1,2\}} (t_{k}-\widetilde{\mu}_{k,\mathbf{u}})/\widetilde{\sigma}_{k, \mathbf{u}}\geq Q_{\mathbf{\rho}^{\diamond}_{\Gamma},\alpha}$ at that given $\Gamma$. Adapting a similar argument to that in Appendix D, this procedure can be implemented through setting
\begin{equation*}
\left\{
\begin{array}{ll}
      s_{i}=1/\sum_{j^{\prime}=1}^{n_{i}}\exp(\gamma u_{ij^{\prime}}) &\quad \text{for\ $i=1,\dots,I_{1}$} \\
      \widetilde{s}_{i}=1/\sum_{j^{\prime}=1}^{n_{i}}\exp(-\gamma u_{ij^{\prime}})   & \quad \text{for \ $i=I_{1}+1,\dots, I_{1}+I_{2}$}. 
\end{array}
\right. 
\end{equation*}
and solving the following quadratically constrained linear program with $M$ being a sufficiently large constant:
\begin{equation*}
     \begin{split}
        \underset{y, p_{ij}, \widetilde{p}_{ij}, s_i, \widetilde{s}_{i}, b_k}{\text{minimize}} \quad &y \quad \quad (\diamond \diamond)\\
         \text{subject to}\quad & y\geq (t_k-\widetilde{\mu}_{k,\mathbf{u}})^2- Q^{2}_{\mathbf{\rho}^{\diamond}_{\Gamma},\alpha} \widetilde{\sigma}^2_{k,\mathbf{u}}-Mb_k \quad \forall k \in \{0,1\}\\
        &\sum_{j=1}^{n_{i}} p_{ij}=1\quad\forall i=1,\dots,I_{1},\\
        &\sum_{j=1}^{n_{i}} \widetilde{p}_{ij}=1\quad\forall i=I_{1}+1,\dots,I_{1}+I_{2},\\
        &s_i\leq p_{ij}\leq \Gamma s_i\quad\forall i=1,\dots,I_{1},j=1,\dots,n_{i},\\
        &\Gamma^{-1}\widetilde{s}_i\leq \widetilde{p}_{ij}\leq \widetilde{s}_i\quad\forall i=I_{1}+1,\dots,I_{1}+I_{2},j=1,\dots,n_{i},\\
        &p_{ij}\geq0 \quad\forall i=1,\dots,I_{1},j=1,\dots,n_{i},\\
        &\widetilde{p}_{ij}\geq0 \quad\forall i=I_{1}+1,\dots,I_{1}+I_{2},j=1,\dots,n_{i},\\
        &b_k\in \{0,1\} \quad \forall k \in \{0,1\}\\
        &-Mb_k\leq t_k-\widetilde{\mu}_{k,\mathbf{u}}\leq M(1-b_k),\quad\forall k \in \{0,1\}
     \end{split}
 \end{equation*}
and checking whether the optimal value $y^{\diamond}_{\Gamma}\geq 0$.

\begin{algorithm}[H]\label{algo:adaptive test in full matching cases}
\textbf{Step 0:} Re-order the matched strata such that with $I=I_{1}+I_{2}$, we have $\sum_{j=1}^{n_{i}}Z_{ij}=1$ for $i=1,\dots,I_{1}$ and $\sum_{j=1}^{n_{i}}Z_{ij}=n_{i}-1$ for $i=I_{1}+1,\dots, I_{1}+I_{2}$ \;
\textbf{Input:} Sensitivity parameter $\Gamma$; level $\alpha$ of the sensitivity analysis; treatment assignment indicator vector $\mathbf{Z}=(Z_{11},\dots, Z_{In_{I}})^{T}$; the score vector $\mathbf{q}_{1}=(q_{111},\dots, q_{In_{I}1})^{T}$ associated with $T_{1}=\sum_{i=1}^{I}\sum_{j=1}^{n_{i}}Z_{ij}q_{ij1}$; the score vector $\mathbf{q}_{2}=(q_{112},\dots, q_{In_{I}2})^{T}$ associated with $T_{2}=\sum_{i=1}^{I}\sum_{j=1}^{n_{i}}Z_{ij}q_{ij2}$\;
\textbf{Step 1:} Solve $(\diamond)$ to get the worst-case correlation $\mathbf{\rho}^{\diamond}_{\Gamma}$ along with the corresponding worst-case quantile $Q_{\mathbf{\rho}^{\diamond}_{\Gamma},\alpha}$ \;
\textbf{Step 2:} Solve $(\diamond \diamond)$ with $Q_{\mathbf{\rho}^{\diamond}_{\Gamma},\alpha}$ obtained from Step 1, and get the corresponding optimal value $y^{\diamond}_{\Gamma}$ \;
\textbf{Output:} If $y^{\diamond}_{\Gamma}\geq 0$, we reject the null; otherwise, we fail to reject.
\caption{Two-stage programming procedure in full matching case}
\end{algorithm}

\begin{center}
{\large\bf Appendix F: Simulated Size of a Sensitivity Analysis}
\end{center}

We study the simulated size of the Mantel-Haenszel test, the aberrant rank test and the adaptive test implementing Algorithm~\ref{algo:adaptive} with the above two tests as the component tests under the aberrant null for various $\Gamma$. Specifically, we set $\beta=0$ in Models 1 and 2 or $\delta=1$ in Models 3 and 4. We set $\alpha=0.05$, $c=1$ and $m=4$ (matching with three controls), and as in Models 1-4, we consider two cases: either $F$ is a standard normal distribution or a standard Laplace distribution. Both $I=100$ and $I=1000$ matched strata are considered. Each simulated size of the Mantel-Haenszel test and the aberrant rank test is based on 20,000 replications and each simulated size of the adaptive test is based on 2,000 replications. The simulation results are presented in Table~\ref{tab:sizewithadaptive}.

\begin{table}[ht]
\centering \caption{Simulated size of the Mantel-Haenszel test, the aberrant rank test and the adaptive test implementing Algorithm~\ref{algo:adaptive} with the above two tests as the component tests under the aberrant null. We set $\alpha=0.05$, $c=1$ and $m=4$ (matching with three controls).}
\label{tab:sizewithadaptive}
\smallskip
\smallskip
\smallskip
\footnotesize
\begin{tabular}{ccccccc}
\hline
\multirow{2}{*}{\textbf{Normal}} & \multicolumn{3}{c}{$I=100$ Matched Strata} & \multicolumn{3}{c}{$I=1000$ Matched Strata}\\
& M-H test & Aberrant rank & Adaptive test & M-H test & Aberrant rank & Adaptive test \\
\hline
$\Gamma=1.00$ & 0.051 & 0.054 & 0.045 & 0.051 & 0.053 & 0.042  \\
$\Gamma=1.05$ & 0.037 & 0.040 & 0.044 & 0.018 & 0.021  & 0.020  \\
$\Gamma=1.10$ & 0.027 & 0.031 & 0.035 & 0.005 & 0.007 & 0.003   \\
\hline
\multirow{2}{*}{\textbf{Laplace}} & \multicolumn{3}{c}{$I=100$ Matched Strata} & \multicolumn{3}{c}{$I=1000$ Matched Strata}\\
& M-H test & Aberrant rank & Adaptive test & M-H test & Aberrant rank & Adaptive test \\
\hline
$\Gamma=1.00$ & 0.053 & 0.055 & 0.044 & 0.049 & 0.053 & 0.042  \\
$\Gamma=1.05$ & 0.039 & 0.043 & 0.041 & 0.020 & 0.021 & 0.024\\
$\Gamma=1.10$ & 0.029 & 0.032 & 0.028  & 0.007 & 0.009 &  0.009  \\
\hline
\end{tabular}
\end{table}

We here provide some insights into the simulation results presented in Table~\ref{tab:sizewithadaptive}. Following the previous literature (e.g., Imbens and Rosenbaum, 2005; Heng et al., 2020), the simulated size of each test in each scenario is calculated under the situation when, parallel with the favorable situation, there is no treatment effect and no hidden bias. When $\Gamma=1$, we can see that all three tests can approximately preserve a 0.05 type I error rate control with realistic sample sizes. When $\Gamma>1$, each simulated size of a sensitivity analysis with that prespecified sensitivity parameter $\Gamma$ is less than 0.05 for all three tests, and decreases as the prespecified sensitivity parameter $\Gamma$ increases. This pattern agrees with that of the power of a sensitivity analysis as shown in Table~\ref{tab:powerwithadaptive}. This is because it is more and more improbable that a sensitivity analysis conducted at a larger and larger $\Gamma$ will reject, either correctly or falsely, the null hypothesis of no treatment effect if, in fact, the treatment assignment is random. This is a general pattern that is not only shared by the above three tests, but also most of the plausible non-parametric tests used in matched observational studies (e.g., Heng et al, 2020; Section 3.5). Note that for all three tests, the simulated size drops substantially as $\Gamma$ increases when the sample size is relatively large $(I=1000)$. This is also expected since when the sample size is large, the standard error of a test statistic should be relatively small compared to the magnitude of bias and in this case the size of a test is driven more by the bias. Another way to understand this pattern is from the design sensitivity. Recall that when there is an actual treatment effect, the chance of rejecting the null hypothesis of no treatment effect (i.e., the power of a test) in a sensitivity analysis conducted with $\Gamma>\widetilde{\Gamma}$ goes to zero as the sample size $I$ goes to infinity, where the design sensitivity $\widetilde{\Gamma}$ approaches 1 as the actual treatment effect approaches zero. Therefore, if there is in fact no treatment effect, the chance of rejecting the null hypothesis (i.e., the size of a test) in a sensitivity analysis with any $\Gamma>1$ goes to zero as the sample size increases.

\begin{center}
{\large\bf Appendix G: More Details on Sections~\ref{sec:intro} and \ref{sec:real data}}
\end{center}

In Section~\ref{sec:examples}, we have mentioned that: ``Numerous causal problems share a similar structure with that of the causal determinants of malnutrition, where we care about whether a certain treatment would change the pattern of some aberrant response (e.g. stunted growth) rather than the average treatment effect over the whole population." We here provide two more examples of such type of causal problems. 
\begin{itemize}
    \item According to WHO (2008), anemia in adult men can be defined as blood hemoglobin (Hb) concentrations $<$ 130 g~/ l, and related studies typically focus on the prevalence and severity of anemia, instead of the change of average blood Hb concentrations among the whole population; see Adamu et al. (2017).
    \item A commonly used definition for low birth weight is infant's weight at birth being less than or equal to 2500 g; see Kramer (1987). Related studies are typically concerned about the low birth weight rate and the severity of low birth weight among the study population, instead of the change of average birth weight among that study population; see Paneth (1995) and Schieve et al. (2002). 
\end{itemize}

In Section~\ref{example}, we examine the causal problem of the potential effect of teenage pregnancy on stunting with children's level data from the Kenya 2003 Demographic and Health Surveys (DHS), which is available at Integrated Public Use Microdata Series (IPUMS). We here provide some motivating examples from the previous literature. According to Darteh et al. (2014), a causal effect of teenage pregnancy on stunting could arise ``as a result of the fact that young mothers require adequate nutrition to fully grow into adults; thus, they struggle with their children over the little food the mother eats." Van de Poel et al. (2007) argued that ``Children of younger mothers could be more prone to malnutrition because of physiological immaturity and social and psychological stress that come with child bearing at young age." We also summarize the total number of stunting cases among the treated individuals and controls in the matched data in Table~\ref{tab:number of stunting}.

\begin{table}[ht]
\centering \caption{The total number of stunting cases among the treated individuals and controls in the matched data.}
\label{tab:number of stunting}
\begin{tabular}{|c c c c c|} 
 \hline
   & Stunted & Non-stunted & Total & Percentage  \\ 
 \hline
 Treated & 55 & 95 & 150 & 36.7\% \\ 
 \hline
 Control & 125 & 325 & 450 & 27.8\% \\
 \hline
 Total & 180 & 420 & 600 & 30.0\% \\
 \hline
 \end{tabular}
\end{table}

In Section~\ref{sec:real data}, we report the worst-case p-values of the Mantel-Haenszel test, the aberrant rank test and the adaptive test. The worst-case p-values of the Mantel-Haenszel test and the aberrant rank test are obtained from the results of the asymptotic approximation of the worst-case p-value in Section~\ref{sec:M-H test} and Section~\ref{sec: aberrant rank test}. Especially, for the aberrant rank test, we apply the asymptotic separability algorithm proposed in Gastwirth et al. (2000) to find an approximate worst-case p-value under some mild conditions on the response vector and the treatment assignment probabilities. See Proposition 1 in Gastwirth et al. (2000) for more details about the mild conditions under which the asymptotic separability algorithm is applicable. Note that the adaptive testing procedure (\ref{test:adaptive}) is a procedure that directly determines if we should reject the null hypothesis or not and does not directly involve the worst-case p-value in the traditional sense. Instead, the worst-case p-value reported by the adaptive test in Table~\ref{tab:p-values} is the value of the prespecified $\alpha$ such that the adaptive testing procedure (\ref{test:adaptive}) barely rejects the null hypothesis under level $\alpha$. That is, we report the value $\alpha^{*}$ such that the following equality holds:
\begin{equation*}
\min \limits_{\mathbf{u} \in \mathcal{U}}\ \max\limits_{k\in \{1,2\}} \frac{t_{k}-\mu_{k,\mathbf{u}}}{\sigma_{k, \mathbf{u}}} = Q_{\mathbf{\rho}^{*}_{\Gamma},\alpha^{*}}.
\end{equation*}
It is clear that the adaptive test rejects the null hypothesis in a level-0.05 sensitivity analysis conducted with sensitivity parameter $\Gamma$ if and only if $\alpha^{*}\leq 0.05$. Like the worst-case p-value in a sensitivity analysis in the traditional sense, the value of $\alpha^{*}$ implies how strong the evidence against the null hypothesis obtained from the adaptive test is: a smaller $\alpha^{*}$ corresponds to a smaller chance that the null hypothesis holds in a sensitivity analysis.

\section*{References}%
%

\setlength{\hangindent}{12pt}
\noindent
Adamu, A. L., Crampin, A., Kayuni, N., Amberbir, A., Koole, O., Phiri, A., Nyirenda, M. and Fine, P. (2017). Prevalence and risk factors for anemia severity and type in Malawian men and women: urban and rural differences. \textit{Population Health Metrics.} \textbf{15(1)} 12.

\setlength{\hangindent}{12pt}
\noindent
Billingsley, P. (1995). \textit{Probability and Measure (Third Edition).} A Wiley-Interscience Publication, John Wiley \& Sons.

\setlength{\hangindent}{12pt}
\noindent
Byrd, R. H., Lu, P., Nocedal, J. and Zhu, C. (1995). A limited memory algorithm for bound constrained optimization. \textit{SIAM Journal on Scientific Computing.} \textbf{16(5)} 1190-1208.

\setlength{\hangindent}{12pt}
\noindent
Charalambous, C. and Conn, A. R. (1978). An efficient method to solve the minimax problem directly. \textit{SIAM Journal on Numerical Analysis.} \textbf{15(1)} 162-187.

\setlength{\hangindent}{12pt}
\noindent
Charnes, A. and Cooper, W. W. (1962). Programming with linear fractional functionals. \textit{Naval Research Logistics Quarterly.} \textbf{9(3‐4)} 181-186.

\setlength{\hangindent}{12pt}
\noindent
Darteh, E. K. M., Acquah, E. and Kumi-Kyereme, A. (2014). Correlates of stunting among children in Ghana. \textit{BMC Public Health.} \textbf{14(1)} 504.

\setlength{\hangindent}{12pt}
\noindent
Fink, G., Günther, I. and Hill, K. (2011). The effect of water and sanitation on child health: evidence from the demographic and health surveys 1986–2007. \textit{International Journal of Epidemiology.} \textbf{40(5)} 1196-1204.

\setlength{\hangindent}{12pt}
\noindent
Fogarty, C. B. and Small, D. S. (2016). Sensitivity analysis for multiple comparisons in matched observational studies through quadratically constrained linear programming. \textit{Journal of the American Statistical Association.} \textbf{111(516)} 1820-1830.

\setlength{\hangindent}{12pt}
\noindent
Gastwirth, J. L., Krieger, A. M. and Rosenbaum, P. R. (2000). Asymptotic separability in sensitivity analysis. \textit{Journal of the Royal Statistical Society: Series B (Statistical Methodology).} \textbf{62(3)} 545-555.

\setlength{\hangindent}{12pt}
\noindent
Heng, S., Small, D. S. and Rosenbaum, P. R. (2020). Finding the strength in a weak instrument in a study of cognitive outcomes produced by Catholic high schools. \textit{Journal of the Royal Statistical Society Series A.} \textbf{183(3)} 935-958.

\setlength{\hangindent}{12pt}
\noindent
Imbens, G. W. and Rosenbaum, P. R. (2005). Robust, accurate confidence intervals with a weak instrument: quarter of birth and education. \textit{Journal of the Royal Statistical Society: Series A.} \textbf{168(1)} 109-126.

\setlength{\hangindent}{12pt}
\noindent
Kramer, M. S. (1987). Determinants of low birth weight: methodological assessment and meta-analysis. \textit{Bulletin of the World Health Organization.} \textbf{65(5)} 663.

\setlength{\hangindent}{12pt}
\noindent
Paneth, N. S. (1995). The problem of low birth weight. \textit{The Future of Children.} 19-34.

\setlength{\hangindent}{12pt}
\noindent
Rosenbaum, P. R. (2004). Design sensitivity in observational studies. \textit{Biometrika.} \textbf{91(1)} 153-164.

\setlength{\hangindent}{12pt}
\noindent
Rosenbaum, P. R. (2012). Testing one hypothesis twice in observational studies. \textit{Biometrika.} \textbf{99(4)} 763-774.

\setlength{\hangindent}{12pt}
\noindent
Schieve, L. A., Meikle, S. F., Ferre, C., Peterson, H. B., Jeng, G. and Wilcox, L. S. (2002). Low and very low birth weight in infants conceived with use of assisted reproductive technology. \textit{New England Journal of Medicine.} \textbf{346(10)} 731-737.

\setlength{\hangindent}{12pt}
\noindent
Van der Vaart, A. W. (2000). \textit{Asymptotic Statistics (Vol. 3).} Cambridge University Press.

\setlength{\hangindent}{12pt}
\noindent
WHO (2008). Worldwide Prevalence of Anaemia 1993-2005: WHO Global Database on Anaemia.

\setlength{\hangindent}{12pt}
\noindent
Zhu, C., Byrd, R. H., Lu, P. and Nocedal, J. (1997). Algorithm 778: L-BFGS-B: Fortran subroutines for large-scale bound-constrained optimization. \textit{ACM Transactions on Mathematical Software (TOMS).} \textbf{23(4)} 550-560.

\end{document}